\documentclass[a4paper,USenglish,11pt]{article}
\usepackage{geometry}
\usepackage{graphicx}
\usepackage{etoolbox}
\usepackage{multirow}
\usepackage{booktabs}
\usepackage{graphicx}
\usepackage{bbm}
\usepackage{tabularx,environ}
\usepackage[utf8]{inputenc}
\usepackage{ stmaryrd }

\usepackage{xspace}

 \usepackage{amsthm}
\usepackage{amsmath,amsfonts,amssymb} 

\usepackage{hyperref}
\usepackage[capitalize]{cleveref}
\usepackage{comment}

\usepackage{thmtools} 
\usepackage{thm-restate}

\newtheorem{observation}{Observation}

\usepackage[textwidth=20mm,obeyFinal,colorinlistoftodos]{todonotes}
\usepackage{marginnote}

\newcommand{\appsymb}{$\bigstar$}
\newcommand{\appref}[1]{{\hyperref[proof:#1]{\appsymb}}}

\newcommand{\mytodo}[2]{\xspace}
\newcommand{\myrevtodo}[2]{{%
		\let\marginpar\marginnote
		\reversemarginpar
		\renewcommand{\baselinestretch}{0.8}%
		}}
\newcommand{\myinlinetodo}[2]{\todo[size=\small, color=#1!50!white, inline, 
	caption={}]{#2}\xspace}
\newcommand{\registerAuthor}[3]{%
	\expandafter\newcommand\csname #2com\endcsname[1]{\mytodo{#3}{\textsc{#2}: 
			##1}}%
	\expandafter\newcommand\csname 
	#2revcom\endcsname[1]{\myrevtodo{#3}{\textsc{#2}: ##1}}%
	\expandafter\newcommand\csname 
	#2inline\endcsname[1]{\myinlinetodo{#3}{\textsc{#2}: ##1}}%
	\expandafter\newcommand\csname 
	#2inlineLater\endcsname[1]{\lv{\myinlinetodo{#3}{\textsc{#2}: ##1}}}%
}
\registerAuthor{Niclas Boehmer}{nb}{orange}
\registerAuthor{Klaus Heeger}{kh}{green}

\newcommand{\HRUQ}{\textsl{HR-$\text{Q}^\text{U}$}\xspace}
\newcommand{\HRLQ}{\textsl{HR-$\text{Q}_\text{L}$}\xspace}
\newcommand{\HRLUQ}{\textsl{HR-$\text{Q}_\text{L}^\text{U}$}\xspace}
\newcommand{\HRLUQI}{\textsl{HA-$\text{Q}_\text{L}^\text{U}$}\xspace}
\newcommand{\HALUQ}{\textsl{HA-$\text{Q}_\text{L}^\text{U}$}\xspace}
\newcommand{\HALQ}{\textsl{HA-$\text{Q}_\text{L}$}\xspace}

\newcommand{\HRUQT}{\textsl{HR-$\text{Q}^\text{U}$-T}\xspace}
\newcommand{\HRLQT}{\textsl{HR-$\text{Q}_\text{L}$-T}\xspace}
\newcommand{\HRLUQT}{\textsl{HR-$\text{Q}_\text{L}^\text{U}$-T}\xspace}
\newcommand{\HRLUQIT}{\textsl{HA-$\text{Q}_\text{L}^\text{U}$-T}\xspace}
\newcommand{\HALUQT}{\textsl{HA-$\text{Q}_\text{L}^\text{U}$-T}\xspace}

\newcommand{\HRLUQtwo}{\textsl{HR-$\text{Q}_{\text{L}\le 2}^\text{U}$}\xspace}

\newtheorem{theorem}{Theorem}
\newtheorem{lemma}{Lemma}
\newtheorem{corollary}{Corollary}
\newtheorem{proposition}{Proposition}

\tikzstyle{vertex}=[draw, circle, fill, inner sep = 2pt]
\tikzstyle{squared-vertex}=[draw, fill, inner sep = 2pt]
\usetikzlibrary{arrows.meta}
\newcommand{\hone}{\ensuremath{h^1}}
\newcommand{\htwo}{\ensuremath{h^2}}
\newcommand{\ibefore}{\ensuremath{\mathcal{I}^\mathrm{before}}}
\newcommand{\iafter}{\ensuremath{\mathcal{I}^\mathrm{after}}}
\newcommand{\uq}{u}
\newcommand{\fl}{flexible\xspace}
\newcommand{\infl}{inflexible\xspace}

\newcommand{\lowerq}{\ell}
\DeclareMathOperator{\edge}{edge}
\DeclareMathOperator{\ver}{vert}
\DeclareMathOperator{\full}{full} 

\newcommand{\minus}{\scalebox{0.75}[1.0]{$-$}}
\newcommand{\ab}[1]{AB\textsuperscript{+}\minus#1\xspace}
\newcommand{\ba}[1]{BA\minus#1\xspace}
\allowdisplaybreaks
\DeclareMathOperator{\open}{open}
\DeclareMathOperator{\quota}{quota}
\usetikzlibrary{calc}
\usetikzlibrary{shapes.geometric}
\tikzstyle{vertex}=[draw, circle, fill, inner sep = 2pt]
\tikzstyle{squared-vertex}=[draw, fill, inner sep = 2pt]
\usetikzlibrary{arrows.meta}

\usepackage{algorithm}
\usepackage[noend]{algpseudocode}
\algnewcommand\algorithmicinput{\textbf{Input:}}
\algnewcommand\algorithmicoutput{\textbf{Output:}}
\algnewcommand\Input{\item[\algorithmicinput]}
\algnewcommand\Output{\item[\algorithmicoutput]}
\algnewcommand\algorithmicgoto{\textbf{GoTo }}
\algnewcommand\GoTo{\item[\algorithmicgoto]}

\begin{document}
	\title{\Large \bf {A Fine-Grained View on Stable Many-To-One Matching Problems with 
	Lower 
	and Upper Quotas}\thanks{An extended abstract of the paper appeared in the 
	proceedings of the 16th International Conference on Web and Internet 
	Economics (WINE 2020), Springer LNCS 12495, pages 31-44, 2020. We thank 
	Robert Bredereck and Rolf Niedermeier 
		for 
		useful discussions.}}

\author{Niclas~Boehmer\thanks{Supported by the DFG project MaMu  (NI 
		369/19).}\xspace and Klaus~Heeger\thanks{Supported by DFG 
		Research Training Group 2434 ``Facets of Complexity''.}}

\date{
	\small
	TU Berlin, Algorithmics and Computational Complexity, Germany\\
	\texttt{\{niclas.boehmer,heeger\}@tu-berlin.de}\\
}
	\maketitle              
	\begin{abstract}
		In the Hospital Residents problem with lower and upper quotas
		(\HRLUQ), 
the goal is to find a stable
matching of residents to hospitals where the
number of residents matched to a hospital is either between its lower and upper 
quota or zero [Bir\'{o} et al., TCS 2010].  We analyze this problem from
a parameterized perspective using several natural parameters such as the number 
of hospitals and the number of residents. Moreover, we present a polynomial-time 
algorithm that finds a
stable matching if it exists on instances with maximum lower quota two.
Alongside \HRLUQ,
we also 
consider two closely related models of independent interest, namely, 
the special case of \HRLUQ where 
each hospital has only a lower quota but no upper quota and the variation of \HRLUQ where
hospitals do not have preferences over residents, which is also known as the 
House Allocation problem with lower and upper quotas. Lastly, we investigate 
how the parameterized complexity of these three models changes if 
preferences may contain ties. 
	\end{abstract}

	\section{Introduction}
	Since its introduction by Gale and 
	Shapely~\cite{GaleS62}, the Hospital
	Residents problem, which is also known as the College Admission problem,
	has
	attracted a lot of attention. Besides a rich body of theoretical 
	work~\cite{DBLP:books/ws/Manlove13}, also
	many
	practical applications have been 
	identified~\cite{DBLP:journals/ijgt/Roth08}.
	Applications include various centralized assignment problems, for example, 
	in 
	the context of education 
	\cite{10.1257/000282805774670167,biroGS,Zhang2010AnalysisOT} or in the
	context of assigning 
	career starters to their first work 
	place~\cite{RePEc:inm:orinte:v:33:y:2003:i:3:p:1-11,RePEc:ucp:jpolec:v:92:y:1984:i:6:p:991-1016}.
	In the classical 
	Hospital Residents problem (\HRUQ), we are given a set of residents, each
	with strict
	preferences over hospitals, and a set of hospitals,
	each 
	with an upper quota and strict preferences over residents.
	In a
	feasible matching of residents to hospitals, the number of residents that 
	are assigned to a hospital is at most its 
	upper quota. A hospital-resident pair~$(h,r)$ blocks a matching~$M$ if
	resident~$r$
	prefers hospital~$h$ to the hospital to which $r$ is matched by~$M$ and
	the number of residents 
	matched to $h$ is below its upper quota or $h$
	prefers~$r$ 
	to one of the residents
	matched to it.
	The task in the Hospital Residents problem is to find a stable
	matching, i.e., a feasible matching that does not admit a blocking pair.
	Gale and Shapely~\cite{GaleS62} presented 
	a linear-time 
	algorithm that always finds a stable matching in a Hospital Residents instance.
	
	In practice, some hospitals may also
	have a lower quota, i.e., a minimum number of assigned residents such that 
	the hospital can open and 
	accommodate them. For example, due to economical or
	political reasons, a university might have a
	lower quota on the number of students for each course of study. Moreover,
	practical considerations might impose lower quotas as well, for instance, if
	residents assigned to a hospital need to perform certain tasks together
	which require at least a given number of participating residents.
	Bir{\'{o}} et 
	al.~\cite{DBLP:journals/tcs/BiroFIM10}
	captured these considerations by extending the Hospital
	Residents problem such that each hospital has a lower and upper quota (\HRLUQ).
	Here, feasibility additionally requires that the number of residents
	assigned to a hospital is either zero or at least its lower quota, while
	stability additionally requires that there does not exist a blocking 
	coalition, i.e., a sufficiently large subset of
	residents that
	want to open a currently closed hospital together.
	Bir{\'{o}} et 
	al.~\cite{DBLP:journals/tcs/BiroFIM10} proved that deciding the existence
	of a stable 
	matching in an~\HRLUQ instance is
	NP-complete. We 
	complement their work with a thorough parameterized complexity analysis of
	\HRLUQ and closely related problems, considering various problem-specific
	parameters. Moreover, we 
	study the special case of \HRLUQ where hospitals have only a lower quota (\HRLQ), which
	has not been considered before.
	
	Lower and upper quotas have also been 
	applied to 
	the House Allocation problem~(\HRLUQI)
	 where the goal is to match a set of applicants to a set of houses 
	 \cite{DBLP:journals/mss/CechlarovaF17,DBLP:journals/orl/Kamiyama13,MONTE201314}.
	  In~\HRLUQI,
	houses have a lower and upper quota but no preferences over 
	applicants, while applicants have preferences over houses. One
	possible application 
	of this model is the assignment of kids to 
	different 
	activities, where lower quotas could arise due to
	economical or practical constraints, 
	for
	instance, playing soccer with only three kids is 
	less fun.
	So far, literature on
	\HRLUQI mainly focused on finding Pareto optimal matchings.
	However, in 
	contrast to the 
	classical House Allocation problem, 
	Pareto optimality in \HRLUQI does
	not imply stability. Thus, finding stable matchings here is an interesting
	problem on its own. For notational convenience, we also refer to houses as 
	hospitals and to applicants as residents in context of \HRLUQI.
	
	\subsection{Our Contributions}
	We provide an extensive parameterized complexity analysis of the Hospital 
	Residents
	problem with lower and upper quotas (\HRLUQ) and of the two closely related
	problems~\HRLQ and \HALUQ.  We also (briefly)
	consider a generalization
	of these models where we allow for ties in the preferences (\HRLUQT, 
\HRLQT, and \HALUQT, respectively).
    By applying the framework of parameterized 
    complexity, we 
	analyze the influence of various problem-specific parameters such as the number of residents or the number of hospitals on the
	complexity of these problems.
	Motivated by
	the observation that there might exist stable matchings opening a different 
	set of hospitals
	(of possibly different sizes)~\cite{DBLP:journals/tcs/BiroFIM10}, we also
	consider the problem of
	deciding whether there exists a stable matching where exactly a given set 
	of hospitals~$H_{\open}$ is open and the problem of deciding whether there 
	exists a
	stable matching with exactly $m_{\open}$ ($m_{\text{closed}}$) open
	(closed) 
	hospitals parameterized by $m_{\open}$ ($m_{\text{closed}}$).
	 
	We present an overview of our results in \Cref{ta:sum}. Our most 
	important technical contribution
	is the design of a polynomial-time algorithm for \HRLUQ (and therefore also 
	for \HRLQ) instances
	where all hospitals have lower quota at most two. This answers an open 
	question raised by Bir\'{o} et al. \cite{DBLP:journals/tcs/BiroFIM10} and by Manlove 
	\cite[p.~298]{DBLP:books/ws/Manlove13}. Such \HRLUQ instances are of 
special 
	theoretical interest, as they, for example, subsume a variant of 
	three-dimensional \textsc{Stable Marriage}, where, 
	given two sets
	of agents each with preferences over the agents from the other set, 
	the
	goal
	is to find a stable set of triples, each consisting of two agents from the
	first and one
	agent from the second set. Moreover, there also exist several
	applications where 
	a lower quota of two is of particular interest, for example, assuming 
	that hospitals correspond to (tennis) coaches and residents to (tennis) 
	players, a coach may require that at least two players are assigned to her
	(as she does not always want to play herself).
	
	\begin{table}[t!]
		\begin{center}
			\renewcommand{\arraystretch}{1.5}
			\setlength\tabcolsep{4pt}
			\resizebox{\textwidth}{!}{\begin{tabular}{c|c|c|c|c|c|c|c|c} \hline
				
				& $q_l\leq 2$ & $q_l = 3$ & $H_{\open}$ &$n$ & $m$ &
				$m_{\quota}$ &
				$m_{\open}$ & $m_{\text{closed}}$  \\ \hline\hline
				
				\HRLUQ & \multirow{1}{*}{P} (T. \ref{t:q2})  & 
				\multirow{6}{*}{\shortstack{NP-c. \\ (T. \ref{th::NP-compl})}}  
				&P (P. \ref{pr:HRULQ-H'})& \multirow{6}{*}{\shortstack{W[1]-h. 
						\\ (T. 
						\ref{th:W-n})}}  & 
				\multicolumn{2}{c|}{FPT (C. \ref{co:mFPT})}  
				& 
				W[1]-h. (C. \ref{c:op}) & 
				\multirow{4}{*}{\shortstack{W[1]-h. \\(T. 
						\ref{thm:ha-m-closed})}}
				\\
				
				\HRLUQT & NP-c. (P. \ref{p:q2T})&&NP-c. (T. 
				\ref{th:HRLUQIH'})&& \multicolumn{2}{c|}{W[1]-h. (P. 
					\ref{pr:HRT-W})}& para-h. (C. \ref{c:HAco})&    \\ 
					\cline{1-2} \cline{4-4}\cline{6-8}
				
				\HRLQ & P (T. \ref{t:q2})  & &
				P (O. \ref{ob:HRLQ-H'})& & 
				\multicolumn{2}{c|}{FPT (C. \ref{co:mFPT})}  
				& \multirow{2}{*}{\shortstack{W[1]-h.\\ (C. \ref{c:op})}}
				& 
				\\
				
				\HRLQT & NP-c. (P. \ref{p:q2T})&&P (P. 
				\ref{pr:HRLQT-H'})&&\multicolumn{2}{c|}{FPT (P. 
				\ref{pr:LTies})}&&   \\ \cline{1-2} 
				\cline{4-4}
				\cline{6-9}
				
				\HALUQ & \multirow{2}{*}{\shortstack{NP-c. \\(P. 
						\ref{pr:HRLUQI-NP-const})}}   &&  
				\multirow{2}{*}{\shortstack{NP-c.\\ 
						(T. 
						\ref{th:HRLUQIH'})}} & &
				FPT (P. 
				\ref{pr:HRLUQI-FPTM}) & 
				\multirow{2}{*}{\shortstack{para-h. \\(P. 
						\ref{pr:HRLUQI-NP-const})}}& 
				\multirow{2}{*}{\shortstack{para-h. \\(C. \ref{c:HAco})}} & 
				\multirow{2}{*}{\shortstack{W[1]-h. \\(C. \ref{c:HAco})}} \\
				\HALUQT &  && & & FPT (C. 
				\ref{co:HRLUQIT-m})
				&& 
				&  \\ \hline
			\end{tabular}}
			\caption{Overview of our results.
			The maximum lower quota of a hospital is denoted by~$q_l$, the 
			number of residents by~$n$, the number of hospitals by~$m$, the 
			number of hospitals with lower quota larger than 1 by~$m_{\quota}$, 
			the number of open hospitals by $m_{\open}$, and the number of 
			closed hospitals by $m_{\operatorname{closed}}$.
			``P'' stands for polynomial-time solvability, ``NP-c.'' for NP-completeness, ``FPT'' for fixed-parameter tractability, ``W[1]-h.'' for W[1]-hardness, and ``para-h.'' for paraNP-hardness with respect to the corresponding parameter.
				With the exception of the parameter 
				$m_{\text{closed}}$ for \HRLUQI, \HALUQT, and \HRLUQT, we 
				present XP algorithms for
				all 
				W[1]-hard cases.
				}
			\label{ta:sum}
		\end{center}
	\end{table}
	
	Our rich set of tractability
	and intractability results allows us to draw several high-level conclusions
	about the considered problems.
	First, our results highlight the differences between the three considered 
	models from a computational perspective: While~\HRLQ is very
	similar to \HRLUQ, 
	\HALUQ is computationally more demanding than \HRLUQ. The first observation 
	suggests that the complexity of \HRLUQ comes solely from the
	lower quotas of hospitals. The second observation
	indicates that the hospitals' preferences in the lower and upper quotas 
	setting make the problem easier, as 
	they may act as a ``tie-breaker'' to decide which resident deserves a 
	better spot in a stable matching. 
	
	Second, our results identify the ``difficult parts'' of the
	considered
	problems. Considering that for \HRLQ and \HRLUQ, we can decide 
	in 
	polynomial time whether there exists a stable matching opening exactly a
	given set of 
	hospitals, the complexity of \HRLQ and \HRLUQ comes purely from
	deciding which hospitals to open and not from the 
	task of assigning residents to hospitals. This 
	finding is strengthened by the observation that most of our 
	hardness reductions 
	also work if we ignore blocking pairs, i.e.,  for
	the problem of assigning residents
	to hospitals such that the lower quota of each hospital is respected and no
	blocking coalition of residents to open a closed hospital exists. 

	Third, our results identify the fine line between
	tractability and intractability. For example, parameterizing the three 
	problems by the number of hospitals leads to fixed-parameter tractability, 
	while only considering the number of open or closed hospitals in a stable 
	matching as a parameter results in W[1]-hardness. 
	A similar contrast arises when considering the
	influence of the total
	number of hospitals and the number of hospitals with non-unit lower quota 
	on~\HALUQ. 
	
	Fourth, we analyze what happens if we allow for ties in the preferences. In 
	this case, \HRLUQT generalizes both \HALUQT and \HRLQT. This is also 
	reflected in the 
	complexity of the three problems parameterized by the number $m$ of 
	hospitals: While for \HALUQT and \HRLQT it is possible to construct a 
	fixed-parameter tractable algorithm by bounding the number 
	of resident types in a function of $m$ (hospitals are indifferent 
	between residents), \HRLUQT is W[1]-hard parameterized by $m$. To prove 
	this, as a side 
	result, we establish 
	that the problem of deciding whether a \HRUQT instance admits a stable 
	matching which matches all residents is W[1]-hard parameterized by $m$. 
	
	\subsection{Related Work}
	After the work of Bir{\'{o}} et al.~\cite{DBLP:journals/tcs/BiroFIM10},
	only few papers revisited computational problems related to the NP-hard 
	Hospital
	Residents problem with lower and
	upper quotas. 
	A notable exception is the work of
	Agoston et al.~\cite{DBLP:journals/jco/AgostonBM16} who proposed an ILP 
	formulation to find stable matchings and several 
	preprocessing rules to decide which hospitals must be open in a stable 
matching.
	Apart from this, most of the
	follow-up work applied the idea of lower and upper quotas to other 
	settings, such as the House Allocation problem
	\cite{DBLP:journals/mss/CechlarovaF17,DBLP:journals/orl/Kamiyama13,MONTE201314} or maximum-weight many-to-one matchings in bipartite graphs~\cite{DBLP:journals/algorithmica/ArulselvanCGMM18},
	 or interpreted it differently.
	 For example, Arulselvan et
	al.~\cite{DBLP:journals/algorithmica/ArulselvanCGMM18}
	studied finding maximum-weight many-to-one matchings in bipartite
	graphs where vertices on one side of the bipartition have a lower and upper quota. They
	conducted a parameterized
	analysis of the resulting computational problems and studied quota-
	and degree-restricted cases.
		
	Hamada et al.~\cite{DBLP:journals/algorithmica/HamadaIM16} introduced 
	an alternative version of the Hospital Residents problem with lower and
	upper quotas. In their model, hospitals have lower and upper quotas,
	but are not allowed to be closed. Thus, in a feasible matching, the
	lower and upper quota of each hospital 
	needs to be respected. As deciding whether a stable matching 
	exists is polynomial-time solvable for this model, their main focus lied on 
	finding a feasible matching 
	minimizing the number of blocking pairs. Mnich et
	al.~\cite{DBLP:journals/algorithmica/MnichS20} studied the \textsc{Stable
	Marriage with Covering Constraints} problem, which corresponds to the special case
	of Hamada et al.'s model where each hospital 
	has unit upper quota, from a parameterized perspective considering parameters
	such as the number of
	blocking pairs and the
	number of hospitals with non-zero lower quota. To capture stable matching
	problems 
	with 
	diversity or distributional constraints, the model of Hamada et 
	al.~\cite{DBLP:journals/algorithmica/HamadaIM16} has been adapted and 
	further developed in various directions, for example, by assuming that 
	residents belong to different types and each hospital has type-specific lower and upper
	quotas 
	\cite{DBLP:conf/atal/0001GSW19,DBLP:journals/jet/EhlersHYY14,DBLP:journals/jair/KurataHIY17}.
	
	Another popular stable matching problem is the Hospital Residents problem 
with 
couples~(\textsl{HRC})~\cite{DBLP:conf/wea/BiroMM14,DBLP:journals/disopt/MarxS11},
 where some of the
residents are grouped in pairs and submit their preferences together.
	The \HRLUQ problem where all hospitals have upper quota at most 
two is closely related to the special case of \textsl{HRC}
where all hospitals have upper quota one:
	Switching the roles of residents and hospitals and interpreting couples as 
hospitals with lower quota two, the only difference between the two problems is 
that the preferences of couples are over pairs of hospitals, while the 
preferences of quota-two hospitals are over single residents. Notably,  
\textsl{HRC} is already NP-hard in the 
described special case \cite{DBLP:conf/wea/BiroMM14}, which is in sharp 
contrast to our polynomial-time algorithm which also applies for \HRLUQ where 
all hospitals have upper quota at most two.
	
	From a technical perspective, our work falls in line with previous work on 
	the parameterized complexity of stable matching problems 
	\cite{DBLP:journals/tcs/AdilGRSZ18,DBLP:conf/sagt/BoehmerBHN20,DBLP:conf/isaac/BredereckHKN19,DBLP:journals/algorithmica/MarxS10,DBLP:journals/disopt/MarxS11,MeeksR20,DBLP:journals/algorithmica/MnichS20}.
	
	\subsection{Structure of the Paper}
    The paper is structured as follows.
    We start with the preliminaries in \Cref{se:pre}.
    In \Cref{sec:strict}, we analyze the parameterized complexity of the 
    considered Hospital Residents problems with lower and upper quotas with 
    respect to several parameters related to the number of residents or 
    hospitals.
    In \Cref{sec:q2}, we present a polynomial-time algorithm for the Hospital Residents problem with lower and upper quotas if the lower quota of each hospital is at most two.
    Afterwards, we extend our model by allowing ties in the preferences in 
    \Cref{se:ties}.
	
	\section{Preliminaries} \label{se:pre}

	We consider different models of stable bipartite many-to-one 
	matchings. For the sake of readability, we refer to all of them as 
	different 
	variants of the Hospital Residents problem with lower and upper quotas
	(\HRLUQ).
	In \HRLUQ, we are given a set~$R=\{r_1,\dots r_n\}$ of
	\emph{residents} and a set $H=\{h_1,\dots,
	h_m\}$ of \emph{hospitals}, each with a lower and upper quota.
	Throughout the paper, $n$ denotes the number of residents and~$m$ the 
	number of hospitals.
	We refer to the joint set of
	residents and hospitals as
	\emph{agents}.
	Each resident $r\in R$ \emph{accepts} a subset of hospitals $A(r)\subseteq 
	H$ and 
	each hospital $h\in H$ \emph{accepts} a subset of residents~$A(h) \subseteq 
	R$. We assume that acceptability is symmetric, 
	that is, $h\in A(r)$ if and only if $r\in A(h)$. 
	Each agent $a\in R\cup H$ has a preference list $\succ_a$ in which all 
	agents 
	from~$A(a)$ are ranked in strict order. For three agents $a$, $a_1$, and 
	$a_2$, 
	we say that $a$ \emph{prefers}~$a_1$ to $a_2$ and write $a_1\succ_a a_2$ if 
	$a_1,a_2\in A(a)$ and $a$ ranks $a_1$ above $a_2$. 

	A \emph{matching} $M$ is a subset of $R\times H$ where each resident
	is contained in at most one pair and for each pair $(r,h)\in M$,
	agents $r$ and $h$ accept each other. For a matching~$M$ and a resident~$r\in 
	R$, we denote by $M(r)$ the hospital to which $r$ is
	matched  
	in~$M$,~i.e.,~$M(r)=h$ if $(r,h)\in M$, and we set $M(r):=\square$
	if $r$ is not matched, i.e., $r\neq r'$ for all $(r', h) \in M$. All residents $r$ prefer each 
	hospital~$h\in A(r)$ to being unmatched, i.e., to
	$M(r)=\square$.
	Further, for a hospital~$h\in H$, we denote
	by $M(h)$ the set of residents 
	that are matched to $h$, i.e., $r\in M(h)$ if $(r,h)\in M$. We 
	sometimes 
	write $M$ as a set of pairs of
	the form $(h,\{r_1,\dots, r_k\})$ which denotes that the residents
	$r_1,\dots, r_k$ (and possibly other residents if $h$ appears in more than 
	one 
	tuple) are matched to hospital $h$ in~$M$.
	
	In \HRLUQ, each hospital $h\in H$ has an upper quota 
	$u(h)$ and a lower quota $l(h)$ with $1\leq l(h)\leq u(h)$. We call a 
	matching $M$ \emph{feasible} if,
	for all hospitals $h\in H$, it either holds that $|M(h)|=0$ or 
	$l(h)\leq|M(h)|\leq u(h)$. We say that a hospital $h\in H$ is \emph{closed} 
	in $M$ if $|M(h)|=0$ and we say that it is \emph{open} otherwise. Moreover,
	we call an open hospital $h\in H$ \emph{full} if $|M(h)|=u(h)$ and an open 
	hospital~$h\in H$ \emph{undersubscribed} if~$|M(h)|<u(h)$.
	In a matching~$M$, a hospital-resident pair~$(r, h)\in R\times H$ is a
	\emph{blocking pair} if 
	$h$ is open in $M$, both~$r$ and $h$ find each other acceptable, $r$ 
	prefers $h$ to $M(r)$, and~$h$ is 
	either undersubscribed or prefers~$r$ to
	at least one resident 
	from $M(h)$. 
	Moreover, we call $(h,\{r_1,\dots, r_k\})$ with $k=l(h)$ a \emph{blocking
	coalition} if 
	$h$ is closed in $M$ and, for all~$i\in [k]$, resident $r_i$
	prefers $h$ to $M(r_i)$.
	In this case, we also write that $\{r_1,\dots,
	r_k\}$ forms a blocking coalition to open $h$. A
	feasible matching is called \emph{stable} if it neither admits a blocking 
	pair nor a blocking coalition.
	It is possible to express different problems related to the classical
	\textsc{Stable Marriage} problem~\cite{GaleS62} in this framework.
	For example, the
	\textsc{Stable Marriage with Incomplete Lists} problem corresponds 
	to finding 
	a stable matching in an \HRLUQ instance where all hospitals have upper 
	quota one.
	
	We also consider the 
	Hospital Residents problem with upper and lower quotas and ties (\HRLUQT). 
	In this case, the preference list $\succsim_{a}$ of each agent $a\in R\cup 
	H$ is a 
	weak order 
	over the set of agents they accept. We say that an agent~$a$
	is \emph{indifferent} between two agents $a_1$ and $a_2$ and write 
	$a_1\sim_a a_2$ if both $a_1 \succsim_a a_2$ and $a_2 \succsim_a a_1$. 
	Further, we say that an agent $a$ \emph{prefers} agent~$a_1$ to 
	agent $a_2$ and write $a_1\succ_a a_2$ if $a_1\succsim_{a} a_2$ and $a$ is 
	not indifferent between $a_1$ and $a_2$. Note 
	that the resulting stability notion is known as weak stability in the 
	literature \cite[Chapter 3]{DBLP:books/ws/Manlove13}.  
	
	We now describe how the other two models considered in this paper can be 
	formulated as variants of \HRLUQ. 
	
	\paragraph{House Allocation problem with lower and upper quotas.}
	\HALUQ corresponds to \HRLUQ with one-sided 
	preferences, 
i.e., all hospitals are indifferent among all 
	residents and residents have strict preferences over hospitals. While the definition of a blocking coalition still applies in this setting, a hospital-resident pair $(r, h)\in R\times H$ is only
	\emph{blocking} if 
	$h$ is open in $M$, $r$ accepts~$h$, resident $r$ 
	prefers $h$ to~$M(r)$, and~$h$ is
	undersubscribed. Note that
\HRLUQ does not subsume \HALUQ, as, in 
\HRLUQ, no ties in the preferences are allowed.
    However, \HRLUQT subsumes \HALUQT by setting the preference of each hospital to be a tie involving all acceptable residents.
	
	\paragraph{Hospital Residents problem with lower quotas.} \HRLQ is
	the special case of \HRLUQ where each hospital has upper quota $n+1$. 
	Thereby, no hospital can be full in a matching. Consequently, in a
	matching~$M$, a resident $r$ forms a blocking pair with each open
	hospital~$h$
	she prefers to $M(r)$. Thus, in every stable matching, all residents
	need to be matched to their most preferred open hospital. This in 
	turn implies that the preferences of hospitals over residents can be 
	omitted, as they have no 
	influence on the stability of a matching. 
	Hence, \HRLQ is equivalent to the House
	Allocation problem with lower quotas $(=\HALQ)$ and thus lies in the ``intersection'' of \HRLUQ and \HRLUQI.

	\paragraph{First observations.}
	\label{se:obs}
	As already observed, \HRLQ instances can be expressed both as~\HRLUQ and 
	\HALUQ instances.
	Notably, most instances constructed in our reductions fulfill an additional 
	property which directly transfers the hardness results to a variant 
of~\HRLQ 
	where only blocking coalitions may make a matching unstable.
	\begin{observation} \label{ob:equiv}
		In \HRLQ instances where for each hospital $h\in H$, the 
		number of
		residents
		accepting $h$ is equal to its lower quota $l(h)$, no feasible matching admits a blocking pair,
		while a feasible matching may admit blocking coalitions.
	\end{observation}
	Unfortunately, a stable matching may fail to exist in \HRLQ instances (and
	therefore also in \HRLUQ and \HALUQ instances),
	even if all hospitals have lower quota
	at most two:
	\begin{observation}\label{ob:counter}
		Let $R=\{r_1,r_2,r_3\}$ and $H=\{h_1,h_2,h_3\}$ be a \HRLQ instance 
		with each hospital having lower quota two and the following 
		preferences: $r_1: h_1 \succ h_2;$ $r_2: h_2 \succ h_3;$ $r_3: h_3 
		\succ h_1.$ This instance does not admit a stable matching.
	\end{observation}
	Note that this example resembles the Condorcet paradox.
	In the following hardness reductions, we 
	will frequently use this construction 
	as a penalizing component to ensure that certain residents 
	are matched to some designated set of hospitals in a stable matching.
	For a visualization of the example, see \Cref{fig:counter}. 
					\begin{figure}
		\begin{center}
			\begin{tikzpicture}
			
			\node[vertex, label=180:$r_1$] (r1) at (0, 2) {};
			\node[vertex, label=180:$r_2$] (r2) at (0, 1) {};
			\node[vertex, label=180:$r_3$] (r3) at (0, 0) {};
			\node[squared-vertex, label=90:$h_3$, label=0:{$[2,\infty]$}] 
			(h3) at (2.5,0) {};
			\node[squared-vertex, label=90:$h_2$, label=0:{$[2,\infty]$}] 
			(h2) at (2.5,1) {};
			\node[squared-vertex, label=90:$h_1$, label=0:{$[2,\infty]$}] 
			(h1) at (2.5,2) {};
			\draw (h1) edge node[pos=0.78, fill=white, inner sep=2pt] 
			{\scriptsize $1$} (r1);
			\draw (h2) edge node[pos=0.78, fill=white, inner sep=2pt] 
			{\scriptsize $2$} (r1);
			\draw (h2) edge node[pos=0.78, fill=white, inner sep=2pt] 
			{\scriptsize $1$} (r2);
			\draw (h3) edge node[pos=0.78, fill=white, inner sep=2pt] 
			{\scriptsize $2$} (r2);		
			\draw (h3) edge node[pos=0.78, fill=white, inner sep=2pt] 
			{\scriptsize $1$} (r3);
			\draw (h1) edge node[pos=0.78, fill=white, inner sep=2pt] 
			{\scriptsize $2$} (r3);				
			\end{tikzpicture}
			
		\end{center}
		\caption{Visualization of a \HRLQ instance without a stable matching 
			(\Cref{ob:counter}).}
		\label{fig:counter}
	\end{figure}
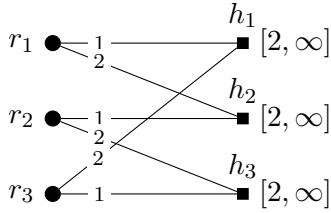
	In this and all 
	following 
	figures, each resident is represented by a dot and each hospital by a 
	square. The quotas of each hospital $h\in H$ are depicted next to the 
	corresponding square as $[l(h),u(h)]$. The preferences of the agents are 
	encoded on the edges:
	The number on an edge between agent $a$ and agent $a'$ that is closer to 
	agent 
	$a$ indicates the rank of $a'$ in $a$'s preference relation, that is, the 
	number of agents $a$ prefers to $a'$ plus one. Note that in this 
	example, hospitals do not have preferences.
	
	Moreover, for all 
	three models, there might exist stable matchings with a different number of 
	open/closed hospitals and a different number of assigned residents. 
	Consider an \HRLQ instance consisting of one 
	hospital~$h_1$ with lower quota 
	one, one hospital $h_2$ with lower quota two, and one hospital~$h_3$ with 
lower quota four together with four residents with 
	the following preferences: 
	$
	r_1: h_3 \succ h_1;$ $
	r_2: h_2 \succ h_3;$ $
	r_3: h_3 \succ h_2;$ $r_4:h_3.
	$
	This instance admits two stable 
	matchings~$M_1=\{(h_3,\{r_1,r_2,r_3,r_4\})\}$ 
and
	$M_2=\{(h_1,\{r_1\}),(h_2,\{r_2,r_3\})\}$.

	\paragraph*{Basic concepts of parameterized complexity.}
A \emph{parameterized problem}~$P\subseteq \{0,1 \}^* \times \mathbb{N}$ is \emph{fixed-parameter tractable} if there exists an algorithm running in
time~$f(k)|\mathcal{I}|^{O(1)}$ for some computable function~$f$ which decides every instance~$(\mathcal{I}, k)$ of $P$. Moreover, a parameterized problem~$P$ is in XP if there exists an algorithm running in time $|\mathcal{I}|^{f(k)}$ for some computable function~$f$ which decides every instance~$(\mathcal{I}, k)$ of~$P$.

There is also a theory of hardness of parameterized
problems that includes the notion of W$[1]$-hardness. If a parameterized problem is
W[$1$]-hard, then it is widely believed
not to be
fixed-parameter tractable. The usual approach to prove W[$1$]-hardness of
a given parameterized problem is
to reduce a known W[$1$]-hard problem to it, using a
parameterized reduction. In this paper, we do not use parameterized reductions 
in full generality but only a special case of
parameterized reductions, that is, standard many-one
reductions which run in polynomial time and fulfill that the parameter of the
output
instance is upper-bounded by a computable function of the parameter of the input
instance.

A parameterized problem is \emph{paraNP-hard} if it is NP-hard for a constant parameter value.

For more details on parameterized complexity, we refer to standard textbooks~\cite{DBLP:books/sp/CyganFKLMPPS15,DBLP:series/txcs/DowneyF13}.

	\section{Parameterized Complexity}
	\label{sec:strict}
	
	In this section, we analyze the parameterized
	computational complexity of \HRLQ, \mbox{\HRLUQ},
	and \mbox{\HRLUQI}.
	We start by proving in \Cref{sec:NP-hard} that all three problems are 
	NP-complete. 
	Then, in \Cref{sec:resi}, we 
	analyze the influence of the number of residents on the complexity of the 
	three problems. 
	Lastly, in \Cref{se:hosp}, we consider the number of hospitals and various 
	parameters related to it such as the number of hospitals with non-unit 
	lower quota and the number of open (or closed) hospitals in a stable 
	matching.
	
	\subsection{An NP-Completeness Result}\label{sec:NP-hard}
	Bir{\'{o}} et al.~\cite{DBLP:journals/tcs/BiroFIM10} proved that
	\HRLUQ is 
	NP-complete, even if each hospital has upper quota at most three.
	However, their reduction
	does not settle the computational
	complexity of~\HRLQ or \HALUQ.
	To answer this question, note that \HRLQ and therefore also \HRLUQ and 
	\HALUQ subsume hedonic games
	(see 
	\cite{DBLP:reference/choice/AzizS16} for definitions):
	We introduce a resident for
	each agent in the given hedonic game and a hospital for each possible 
	coalition with lower quota equal 
	to the size of the coalition. We replace the coalitions in the agents'
	preferences by the corresponding hospitals.
	Core stable outcomes in the hedonic game then correspond
	to 
	stable matchings in the constructed \HRLQ instance, which notably falls 
under
	\Cref{ob:equiv}. As deciding the existence of a core stable outcome is 
	NP-complete, even if all coalitions 
	have size three
	\cite{DBLP:journals/siamdm/NgH91}, this 
	implies that all three problems are NP-complete, even if each hospital has 
	lower quota (and upper 
	quota) at most three. 
	By slightly adopting the reduction from Ng and Hirschberg~\cite{DBLP:journals/siamdm/NgH91},
	one can also bound the number of residents acceptable to a hospital and the 
	number of hospitals acceptable to a resident. For the sake of completeness 
	and as a warm-up to 
	illustrate the basic features of the three problems we consider, we present 
	here 
	an alternative reduction from 
	\textsc{Satisfiability}.

	\begin{theorem}
	\label{th::NP-compl}
	\HRLQ, \HRLUQ, and \HALUQ are NP-complete, even if each resident accepts
	at most four hospitals, each hospital accepts at most three
	residents, and the lower (and upper) quota of every hospital
	is at most 
	three.
	\end{theorem}
	\begin{proof}
		We reduce from the NP-hard variant of \textsc{Satisfiability} where 
each 
		clause 
		contains 
		exactly three literals and each variable occurs exactly twice 
		positively 
		and twice negatively~\cite{DBLP:journals/eccc/ECCC-TR03-049}. An 
		instance of \textsc{Satisfiability} consists of 
		a  
		set 
		$X=\{x_1,\dots x_q\}$ of variables and a set $C=\{c_1,\dots c_p\}$ of 
		clauses. For
		each 
		$i\in [q]$, we denote as $c^\text{pos}_{i,1}$ and 
		$c^\text{pos}_{i,2}$ 
		the two clauses in which variable $x_i$ appears positively and as
		$c^\text{neg}_{i,1}$ and $c^\text{neg}_{i,2}$ the two clauses in which
		variable $x_i$ appears negatively. From this, we construct an instance 
		of \HRLQ. 
		
		\textbf{Construction:}
		For each $i\in [q]$, we introduce two hospitals $h_{i}$ and  
		$\bar{h}_i$ with lower quota three and two hospitals $h_{i}^*$ and 
		$\bar{h}_{i}^*$ 
		with 
		lower quota two. Moreover, we add three hospitals~$h^1_i$, $h^2_i$, and~$h^3_i$ with lower quota two that will help us to build a penalizing 
		component. Finally, we create one hospital $h_c$ for each clause $c\in 
		C$ 
		with lower quota three, which is supposed to be closed in every stable 
		matching. 
		
		Turning to the residents, for each $i\in [q]$, we add two 
		\emph{variable residents} ($r_{i}$
		and~$\bar{r}_i$),
		two \emph{dummy residents} ($d^1_i$ and
		$d^2_i$), and three \emph{penalizing residents} ($s^*_i,s^1_i,s^2_i$). 
		The
		assignment of the variable residents encode the truth assignments and 
		their 
		preferences are as follows:
		$$r_i: h_{i} \succ h_{c^\text{pos}_{i,1}} \succ  
		h_{c^\text{pos}_{i,2}} \succ h_{i}^*, \qquad  \bar{r}_i: 
		\bar{h}_i \succ h_{c^\text{neg}_{i,1}} \succ  
		h_{c^\text{neg}_{i,2}} \succ \bar{h}_{i}^*.$$
		The two dummy residents enable us to choose for each $i\in [q]$ which 
		of 
		the two hospitals~$h_{i}$ and $\bar{h}_i$ should be open and 
		thereby which of the two residents $r_i$ and $\bar{r}_i$ 
		are 
		matched to her top-choice. The preferences of the dummy residents are as
		follows: 
		$$d^1_i: h_{i} \succ \bar{h}_i, \qquad  d^2_i: 
		\bar{h}_i \succ h_{i}.$$
		Lastly, we construct a penalizing component for each variable $i\in 
		[q]$ 
		consisting of three agents whose preferences are as follows: 
		$$s^*_i: h_{i}^*\succ \bar{h}_{i}^* \succ h_i^1 \succ h_i^2,
		\qquad  s_i^1:
		h_{i}^2 \succ h_{i}^3, \qquad s_i^2:
		h_{i}^3 \succ h_{i}^1.$$ Note that if $s^*_i$ is not matched to
		$h_{i}^*$ or $\bar{h}_{i}^*$, then no stable matching of these
		three residents~$s_i^*$, $s_i^1$, and $s_i^2$ to their acceptable
		hospitals $h_i^1$, $h_i^2$, and $h_i^3$ exists. Using this, we prove in 
		the following that in a stable matching, for
		all $i\in [q]$, either $r_i$ is assigned to $h_i^*$ and 
		$\bar{r}_i$ is assigned to $\bar{h}_i$ (which 
		corresponds to setting $x_i$ to false) or 
		$\bar{r}_i$ is assigned to $\bar{h}_i^*$ and $r_i$ 
		is assigned to~$h_i$ (which corresponds to setting $x_i$ to 
		true). The resulting assignment needs to satisfy every clause~$c\in 
		C$, as otherwise the residents corresponding to the literals from $c$ 
		form a blocking coalition to open $h_c$.  Note that 
		in the 
		constructed instance, for all hospitals, the number of residents 
		accepting it is equal to its lower quota.
		As every variable appears exactly twice positively and twice negatively, every resident accepts at most four hospitals.
		As every clause contains exactly three literals, every hospital accepts at most three residents.
		We now prove that
		the 
		given \textsc{Satisfiability} instance has a satisfying assignment if 
		and only if
		there exists a stable matching in the constructed~\HRLQ instance. 
		
		{\bfseries ($\Rightarrow$)} Let $Z$ be the set of variables that are 
		set 
		to 
		true in a satisfying assignment of the given \textsc{Satisfiability} instance.
		From 
		this we construct a stable matching $M$:
		\begin{align*}
		\{(h_{i},\{r_i, d_i^1, d_i^2\}),(\bar{h}_{i}^*,\{\bar{r}_i, 
		s^*_i\})  \mid x_i\in Z\} & \cup
		\{(\bar{h}_i,\{\bar{r}_i, 
		d_i^1, d_i^2\}),(h_{i}^*,\{r_i, 
		s^*_i\})  \mid x_i\notin Z\}\\ & \cup
		\{(h_i^3, \{s_i^1, s_i^2\})\mid i\in [q] \}.
		\end{align*}
		As the constructed instance falls under \Cref{ob:equiv}, no blocking 
		pair exists.
		We now iterate over all closed hospitals and argue why there 
		does not exist a blocking coalition to open this hospital. For each $i\in
		[q]$, one of $h_{i}$ and $\bar{h}_i$ is closed. However, 
		as 
		one of the three residents (either $d_i^1$ or~$d_i^2$) that find such a
		hospital acceptable is matched 
		to her top-choice in $M$, no blocking coalition to open $h_i$ or $\bar{h}_i$ exists.
		Moreover, for each $i\in [q]$, one of $h_{i}^*$ and $\bar{h}_i^*$ 
		is closed. However, 
		as 
		one of the two residents (either $r_i$ or $\bar{r}_i$) that find 
		such a 
		hospital acceptable is matched 
		to her top-choice in~$M$, no blocking coalition to open $h_i^*$ or $\bar{h}_i^*$ exists.
        Next we consider a clause hospital~$h_c$ for some $c\in C$.
		As~$Z$ induces a satisfying assignment, for at
		least 
		one 
		literal occurring in~$c$, the corresponding resident is matched to its 
		top-choice. 
		Thus, there is no blocking coalition to open~$h_c$.
		Finally, there does not exist a blocking coalition to open~$h_i^1$ and~$h_i^2$, as both or one of the two residents that find one 
		of 
		these hospitals acceptable are matched better in $M$.
		
		 {\bfseries ($\Leftarrow$)} Assume that there exists a stable matching 
		 $M$ 
		of 
		residents to hospitals. First of all, note that every resident $s_i^*$
		needs 
		to be matched to~$h_i^*$ or $\bar{h}_i^*$, as otherwise no stable 
		matching of the residents~$s_i^*$, $s_i^1$, and~$s_i^2$ to the
		remaining acceptable hospitals $h_i^1$, $h_i^2$, and~$h_i^3$ can exist 
		(see also \Cref{ob:counter}).
		Consequently, for each $i\in [q]$,
		either $h_{i}^*$ or~$\bar{h}_{i}^*$ needs to be open and at 
		least one of the residents 
		$r_i$ and $\bar{r}_i$ needs to be matched to one of the two. 
Moreover, 
		no clause hospital can be open in a stable matching: Let us assume
		that 
		some clause hospital is open and, without loss of generality, some 
		resident $r_i$ is matched to 
		it. Then, as argued above, $\bar{r}_i$ needs to be assigned to 
		$\bar{h}_{i}^*$. However, such an assignment is blocked by the two 
		unassigned 
		dummy residents~$d_i^1$ and~$d_i^2$ and~$\bar{r}_i$ to 
		open~$\bar{r}_i$'s top-choice~$\bar{h}_i$.
		Thus, for each~$i \in [q]$, either~$r_i$ is matched to~$h_{i}$ 
		and~$\bar{r}_i$ is matched to~$\bar{h}_i^*$, or~$r_i$ is 
		matched 
		to~$h_{i}^*$ and~$\bar{r}_i$ is matched to 
		$\bar{h}_i$.
		
		Let $Z=\{
		x_i\in X \mid M(r_{i})\neq h_{i}^*\}$. We claim that setting all 
		variables in 
		$Z$ 
		to true and all others to false induces a satisfying assignment of the 
		given propositional formula. For the sake of contradiction, let us 
		assume 
		that there exists a clause $c=\{\ell_1,\ell_2,\ell_3\}\in C$ which is 
		not 
		fulfilled. However, this implies that all three corresponding variable 
		residents are all matched to their least preferred acceptable 
		hospital. From this is follows that $M$
		cannot be stable, as these three variable residents then form a 
		blocking coalition to open~$h_c$. 
	\end{proof}
	Note that the \HRLQ instance constructed in the reduction falls under 
	\Cref{ob:equiv}. This implies that all three models we consider remain 
	computationally hard for lower quota at most three if stability only 
	requires that no 
	blocking coalition exists (but blocking pairs might exist).
	\subsection{Parameterization by Number of Residents}\label{sec:resi}
	
	After establishing the NP-hardness of \HRLQ, \HRLUQ, and \HALUQ, we now analyze
	their computational complexity parameterized by the number of residents. 
While there 
	exists a
	straightforward XP algorithm for this parameter that guesses for each 
	resident the hospital she is assigned to, all three problems are W[1]-hard.
	
	\begin{theorem}
		\label{th:W-n}
		Parameterized by the number $n$  of residents, \HRLQ,
		\HRLUQ, and~\HRLUQI are W[1]-hard, even if every hospital has lower
		(and upper) quota at most four and accepts at most four residents.
	\end{theorem}
		\begin{proof}
			We prove hardness by a parameterized reduction from 
			\textsc{Multicolored 
				Independent Set}. In an instance of \textsc{Multicolored 
Independent Set}, we are 
			given an 
			undirected graph~$G=(V=\{v_1,\dots, v_n\},E)$ together with a
			partition 
			$(V^1,\dots, V^k)$ of
			$V$ into $k$ different colors and the question is whether there 
			exists an independent set of size~$k$ containing 
exactly one vertex from each color, 
			i.e., 
			a 
			subset $V'\subseteq V$ with $|V'|=k$ and $|V'\cap V^c|=1$ for all
			$c\in 
			[k]$ such that no two vertices in $V'$ are adjacent.
			\textsc{Multicolored Independent Set} parameterized by~$k$ is 
W[1]-hard~\cite{Pietrzak03}.
			To
			simplify
			notation, we reduce from a restricted still W[1]-hard version of 
			\textsc{Multicolored Independent Set} where there exist some 
			(arbitrary) integers~$p$ and~$q$
			such 
			that each vertex $v$ is adjacent to exactly~$p$ vertices, and 
each color $c\in [k]$ contains exactly~$q$ vertices.
Moreover, without loss of generality, we assume that there are no edges between 
vertices of the same color.
			For each vertex~$v\in V$, let $u_1^{v},\dots,u_{p}^{v}$ be a list 
			of all
			vertices 
			incident to $v$. For a vertex $v\in V$ and an integer $i\in 
			[p]$, we write $\text{z}(v,i)$ to denote the $i$-th vertex that is 
			incident to $v$, i.e., $\text{z}(v,i):=u_i^{v}$.
			Moreover, for each
			color $c\in [k]$, let	$v_1^c,\dots, v_{q}^c$ be a list of all 
			vertices 
			with this 
			color. From an instance of \textsc{Multicolored Independent Set}, 
			we now construct an instance 
			of \HRLQ. 
			
			\textbf{Construction:} For each vertex
			$v\in V$, we introduce a \emph{vertex hospital}
			$h_v$ with lower quota
			three. Moreover, 
			for each edge $\{v,v'\}\in V$,
			we introduce an \emph{edge hospital}
			$h_{\{v,v'\}}$ with lower quota
			four. Furthermore, we 
			introduce for each color $c\in [k]$ a penalizing component 
			consisting 
			of 
			three \emph{penalizing hospitals} $h_1^c$, $h_2^c$, and $h_3^c$ 
			with lower
			quota two.
			
			Turning to the residents, we introduce for each color $c\in [k]$ 
			two 
			\emph{color residents} $r^c_1$ and~$r^c_2$.
			One of the color residents ranks the vertex
			hospitals corresponding to vertices of color $c$ in some ordering 
			and 
			the 
			other color resident ranks them in reversed order. Both residents 
			rank
			directly 
			in front of each vertex hospital all edge hospitals involving 
			the 
			corresponding vertex in an arbitrary order. That 
			is:                                                                 
			\begin{align*}
			r^c_1:\mkern4mu &  h_{\{v_1^c, \text{z}(v_1^c,1)\}}\succ \dots \succ
			h_{\{v_1^c, 
				\text{z}(v_1^c,p)\}}\succ h_{v_1^c}\succ \dots
			\succ 
			h_{\{v_{2}^c, \text{z}(v_{2}^c,1)\}}\succ \dots \succ 
			h_{\{v_{2}^c, 
				\text{z}(v_{2}^c,p)\}}
			\succ h_{v_{2}^c} \succ \\
			&  h_{\{v_{q-1}^c, \text{z}(v_{q-1}^c,1)\}}\succ \dots \succ
			h_{\{v_{q-1}^c, 
				\text{z}(v_{q-1}^c,p)\}}\succ h_{v_{q-1}^c}\succ \dots
			\succ 
			h_{\{v_{q}^c, \text{z}(v_{q}^c,1)\}}\succ \dots \succ 
			h_{\{v_{q}^c, 
				\text{z}(v_{q}^c,p)\}}
			\succ h_{v_{q}^c},\\
			r^c_2:\mkern4mu & h_{\{v_{q}^c, \text{z}(v_{q}^c,1)\}}\succ \dots 
			\succ
			h_{\{v_{q}^c,
				\text{z}(v_{q}^c,p)\}}
			\succ h_{v_{q}^c}\succ \dots \succ h_{\{v_{q-1}^c, 
				\text{z}(v_{q-1}^c,1)\}}\succ \dots \succ h_{\{v_{q-1}^c, 
				\text{z}(v_{q-1}^c,p)\}}\succ h_{v_{q-1}^c} \succ \\
			& h_{\{v_{2}^c, \text{z}(v_{2}^c,1)\}}\succ \dots \succ
			h_{\{v_{2}^c,
				\text{z}(v_{2}^c,p)\}}
			\succ h_{v_{2}^c}\succ \dots \succ h_{\{v_1^c, 
				\text{z}(v_1^c,1)\}}\succ \dots \succ h_{\{v_1^c, 
				\text{z}(v_1^c,p)\}}\succ h_{v_1^c}.
			\end{align*}
			Moreover, for each color $c\in [k]$, we introduce a penalizing 
			component ensuring that no edge hospital can be open in a stable 
			matching. The penalizing component consists of three 
			\emph{penalizing residents} 
			$s^c_*$, $s^c_1$,
			$s^c_2$:
			$$s^c_*:h_{v_1^c}\succ\dots \succ h_{v_{q}^c}\succ h_1^c\succ
			h_2^c,
			\qquad s^c_1: h_2^c\succ h_3^c, \qquad s^c_2: h_3^c\succ h_1^c.$$
			See \Cref{fig:residents-hardness} for an example of the construction.
			The penalizing component enforces that for each color at least 
			(and, in 
			fact, exactly one) hospital needs to be open. The preferences of 
			the 
			color residents are constructed in a way such that no two 
			hospitals that correspond to adjacent vertices can be open. Note 
			that for each constructed hospital, the number of residents 
			accepting it is equal its lower quota.
			Thereby, the set of open vertex hospitals corresponds to a 
			multicolored independent set.
			We now prove that there exists a solution to the given 
			\textsc{Multicolored 
				Independent Set} instance if and only if there exists a stable 
matching 
			in 
			the 
			constructed \HRLQ instance. 
			
			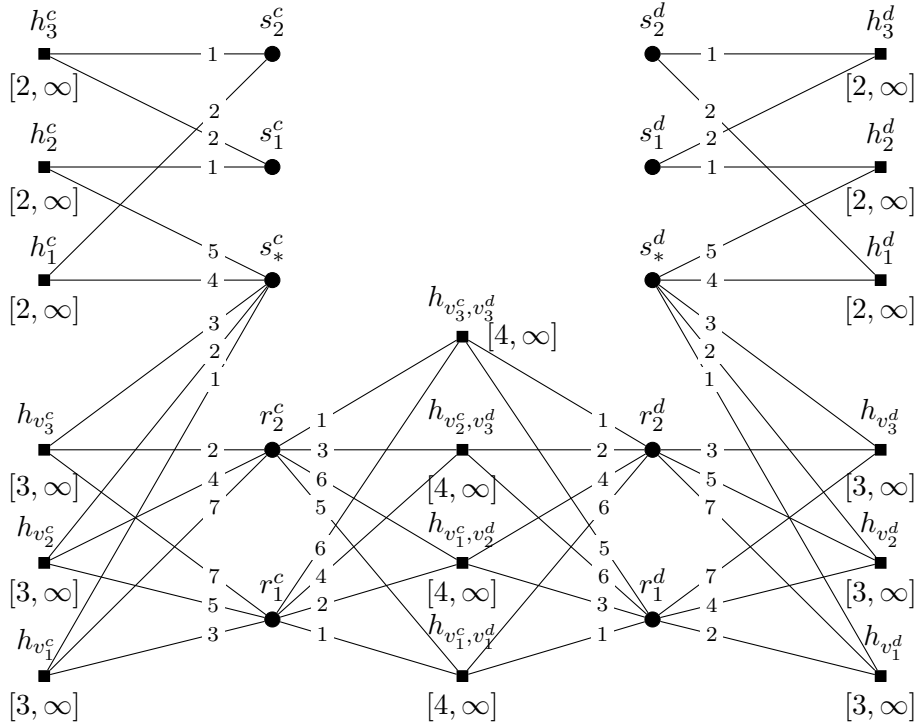
\begin{figure}
			  \begin{center}
			    \begin{tikzpicture}
                  \node (ver-dist) at (0, 1.5) {};
                  \node (hor-dist) at (5, 0) {};
                  \node (hd) at (3, 0) {};

			      \node[vertex, label=90:$r^c_1$] (r1) at (0, 0) {};
			      \node[vertex, label=90:$r^c_2$] (r2) at ($(r1) + 1.5*(ver-dist)$) {};
			      \node[vertex, label=90:$s^c_*$] (r3) at ($(r2) + 1.5*(ver-dist)$) {};
			      \node[vertex, label=90:$s^c_1$] (r4) at ($(r3) + (ver-dist)$) {};
			      \node[vertex, label=90:$s^c_2$] (r5) at ($(r4) + (ver-dist)$) {};

			      \node[vertex, label=90:$r^d_1$] (s1) at ($(r1) + (hor-dist)$) {};
			      \node[vertex, label=90:$r^d_2$] (s2) at ($(s1) + 1.5*(ver-dist)$) {};
			      \node[vertex, label=90:$s^d_*$] (s3) at ($(s2) + 1.5*(ver-dist)$) {};
			      \node[vertex, label=90:$s^d_1$] (s4) at ($(s3) + (ver-dist)$) {};
			      \node[vertex, label=90:$s^d_2$] (s5) at ($(s4) + (ver-dist)$) {};

			      \node[squared-vertex, label=90:$h_1^c$, label=270:{$[2,\infty]$}] (hc1) at ($(r3) - (hd)$) {};
			      \node[squared-vertex, label=90:$h_2^c$, 
			      label={[yshift=-0.05cm]270:{$[2,\infty]$}}] (hc2) at ($(hc1) 
			      + (ver-dist)$) {};
			      \node[squared-vertex, label=90:$h_3^c$, 
			      label={[yshift=-0.05cm]270:{$[2,\infty]$}}] (hc3) at ($(hc2) 
			      + 
			      (ver-dist)$) {};

			      \node[squared-vertex, label=90:$h_1^d$, label=270:{$[2,\infty]$}] (hd1) at ($(s3) + (hd)$) {};
			      \node[squared-vertex, label=90:$h_2^d$, 
			      label={[yshift=-0.05cm]270:{$[2,\infty]$}}] (hd2) at ($(hd1) 
			      + (ver-dist)$) {};
			      \node[squared-vertex, label=90:$h_3^d$, 
			      label={[yshift=-0.05cm]270:{$[2,\infty]$}}] (hd3) at ($(hd2) 
			      + (ver-dist)$) {};

			      \node[squared-vertex, label={[yshift=0.15cm]90:$h_{v_1^c, v_1^d}$}, label=270:{$[4,\infty]$}] (ve1) at ($0.5*(hor-dist) -0.5*(ver-dist)$) {};
			      \node[squared-vertex, label=90:$h_{v_1^c, v_2^d}$, label=270:{$[4,\infty]$}] (ve2) at ($(ve1) + (ver-dist) $) {};
			      \node[squared-vertex, label=90:$h_{v_2^c, v_3^d}$, 
			      label={[yshift=-0.13cm]270:{$[4,\infty]$}}] (ve3) at ($(ve2) 
			      +(ver-dist) $) {};
			      \node[squared-vertex, label=90:$h_{v_3^c, v_3^d}$, 
			      label={[xshift=0.1cm]0:{$[4,\infty]$}}] (ve4) at ($(ve3) 
			      +(ver-dist) $) {};

			      \node[squared-vertex, label={[xshift=-0.1cm]90:$h_{v_1^c}$}, label=270:{$[3,\infty]$}] (vc1) at ($(r1) - (hd) - 0.5*(ver-dist)$) {};
			      \node[squared-vertex, label={[xshift=-0.1cm]90:$h_{v_2^c}$}, label=270:{$[3,\infty]$}] (vc2) at ($(vc1) + (ver-dist)$) {};
			      \node[squared-vertex, label={[xshift=-0.1cm]90:$h_{v_3^c}$}, 
			      label={[yshift=-0.1cm]270:{$[3,\infty]$}}] (vc3) at ($(vc2) 
			      + (ver-dist)$) {};

			      \node[squared-vertex, label={[xshift=0.05cm]90:$h_{v_1^d}$}, 
			      label=270:{$[3,\infty]$}] (vd1) at ($(s1) + (hd) - 
			      0.5*(ver-dist)$) {};
			      \node[squared-vertex, label=90:$h_{v_2^d}$, label=270:{$[3,\infty]$}] (vd2) at ($(vd1) + (ver-dist)$) {};
			      \node[squared-vertex, label=90:$h_{v_3^d}$, 
			      label={[yshift=-0.1cm]270:{$[3,\infty]$}}] (vd3) at ($(vd2) 
			      + (ver-dist)$) {};

        \draw (vc1) edge node[pos=0.76, fill=white, inner sep=2pt] {\scriptsize $3$} (r1);
        \draw (vc1) edge node[pos=0.76, fill=white, inner sep=2pt] {\scriptsize $7$} (r2);
        \draw (vc2) edge node[pos=0.76, fill=white, inner sep=2pt] {\scriptsize $5$} (r1);
        \draw (vc2) edge node[pos=0.76, fill=white, inner sep=2pt] {\scriptsize $4$} (r2);
        \draw (vc3) edge node[pos=0.76, fill=white, inner sep=2pt] {\scriptsize $7$} (r1);
        \draw (vc3) edge node[pos=0.76, fill=white, inner sep=2pt] {\scriptsize $2$} (r2);

        \draw (vd1) edge node[pos=0.76, fill=white, inner sep=2pt] {\scriptsize $2$} (s1);
        \draw (vd1) edge node[pos=0.76, fill=white, inner sep=2pt] {\scriptsize $7$} (s2);
        \draw (vd2) edge node[pos=0.76, fill=white, inner sep=2pt] {\scriptsize $4$} (s1);
        \draw (vd2) edge node[pos=0.76, fill=white, inner sep=2pt] {\scriptsize $5$} (s2);
        \draw (vd3) edge node[pos=0.76, fill=white, inner sep=2pt] {\scriptsize $7$} (s1);
        \draw (vd3) edge node[pos=0.76, fill=white, inner sep=2pt] {\scriptsize $3$} (s2);

        \draw (vc1) edge node[pos=0.76, fill=white, inner sep=2pt] {\scriptsize $1$} (r3);
        \draw (vc2) edge node[pos=0.76, fill=white, inner sep=2pt] {\scriptsize $2$} (r3);
        \draw (vc3) edge node[pos=0.76, fill=white, inner sep=2pt] {\scriptsize $3$} (r3);
        \draw (hc1) edge node[pos=0.76, fill=white, inner sep=2pt] {\scriptsize $4$} (r3);
        \draw (hc2) edge node[pos=0.76, fill=white, inner sep=2pt] {\scriptsize $5$} (r3);

        \draw (hc2) edge node[pos=0.76, fill=white, inner sep=2pt] {\scriptsize $1$} (r4);
        \draw (hc3) edge node[pos=0.76, fill=white, inner sep=2pt] {\scriptsize $2$} (r4);

        \draw (hc3) edge node[pos=0.76, fill=white, inner sep=2pt] {\scriptsize $1$} (r5);
        \draw (hc1) edge node[pos=0.76, fill=white, inner sep=2pt] {\scriptsize $2$} (r5);

        \draw (ve1) edge node[pos=0.76, fill=white, inner sep=2pt] {\scriptsize $1$} (r1);
        \draw (ve1) edge node[pos=0.76, fill=white, inner sep=2pt] {\scriptsize 
        $5$} (r2);
        \draw (ve2) edge node[pos=0.76, fill=white, inner sep=2pt] {\scriptsize $2$} (r1);
        \draw (ve2) edge node[pos=0.76, fill=white, inner sep=2pt] {\scriptsize 
        $6$} (r2);
        \draw (ve3) edge node[pos=0.76, fill=white, inner sep=2pt] {\scriptsize $4$} (r1);
        \draw (ve3) edge node[pos=0.76, fill=white, inner sep=2pt] {\scriptsize $3$} (r2);
        \draw (ve4) edge node[pos=0.76, fill=white, inner sep=2pt] {\scriptsize $6$} (r1);
        \draw (ve4) edge node[pos=0.76, fill=white, inner sep=2pt] {\scriptsize $1$} (r2);

        \draw (vd1) edge node[pos=0.76, fill=white, inner sep=2pt] {\scriptsize $1$} (s3);
        \draw (vd2) edge node[pos=0.76, fill=white, inner sep=2pt] {\scriptsize $2$} (s3);
        \draw (vd3) edge node[pos=0.76, fill=white, inner sep=2pt] {\scriptsize $3$} (s3);
        \draw (hd1) edge node[pos=0.76, fill=white, inner sep=2pt] {\scriptsize $4$} (s3);
        \draw (hd2) edge node[pos=0.76, fill=white, inner sep=2pt] {\scriptsize $5$} (s3);

        \draw (hd2) edge node[pos=0.76, fill=white, inner sep=2pt] {\scriptsize $1$} (s4);
        \draw (hd3) edge node[pos=0.76, fill=white, inner sep=2pt] {\scriptsize $2$} (s4);

        \draw (hd3) edge node[pos=0.76, fill=white, inner sep=2pt] {\scriptsize $1$} (s5);
        \draw (hd1) edge node[pos=0.76, fill=white, inner sep=2pt] {\scriptsize $2$} (s5);

        \draw (ve1) edge node[pos=0.76, fill=white, inner sep=2pt] {\scriptsize $1$} (s1);
        \draw (ve1) edge node[pos=0.76, fill=white, inner sep=2pt] {\scriptsize $6$} (s2);
        \draw (ve2) edge node[pos=0.76, fill=white, inner sep=2pt] {\scriptsize $3$} (s1);
        \draw (ve2) edge node[pos=0.76, fill=white, inner sep=2pt] {\scriptsize $4$} (s2);
        \draw (ve3) edge node[pos=0.76, fill=white, inner sep=2pt] {\scriptsize 
        $6$} (s1);
        \draw (ve3) edge node[pos=0.76, fill=white, inner sep=2pt] {\scriptsize $2$} (s2);
        \draw (ve4) edge node[pos=0.76, fill=white, inner sep=2pt] {\scriptsize 
        $5$} (s1);
        \draw (ve4) edge node[pos=0.76, fill=white, inner sep=2pt] {\scriptsize $1$} (s2);

        \end{tikzpicture}

			  \end{center}
              \caption{An example for the reduction showing W[1]-hardness of 
              \HRLQ, \HRLUQ, and \HALUQ from \Cref{th:W-n}.
              Let $G=(\{v_1^c,v_2^c,v_3^c,v_1^d,v_2^d,v_3^d\},\{\{v_1^c, 
              v_1^d\},\{v_1^c, v_2^d\},\{v_2^c, v_3^d\},\{v_3^c, 
              v_3^d\}\})$ and 
              $(V^c=\{v_1^c,v_2^c,v_3^c\},V^d=\{v_1^d,v_2^d,v_3^d\})$. 
              The picture shows the output of the reduction on this instance.}
              \label{fig:residents-hardness}
			\end{figure}

			{\bfseries ($\Rightarrow$)} Assume that $V'=\{v_{i_1},\dots, 
			v_{i_{k}}\}$ is an 
independent set
			in $G$ with $v_{i_c}\in V^c$ for all~$c\in [k]$. From this we
			construct a stable matching $M$ as follows:
			$$M=\{(h_{v_{i_c}},\{r^c_1,r^c_2, s^c_*\}),
			(h^{c}_3,\{s^c_1,s^c_2\})\mid
			c\in [k]\}.$$
			As the constructed instance falls under \Cref{ob:equiv}, no 
			blocking pair can exist.
			It remains to argue that for no closed hospital $h$
			there exists a
			coalition of residents to open it in $M$. Note that
			there does not exist a coalition to open $h^c_1$ or $h^c_2$ for any
			$c\in
			[k]$, as $s^c_*$ is
			matched 
			to a hospital she prefers to both $h^c_1$ and $h^c_2$. Note further
			that for
			each color $c\in [k]$, the only hospitals that both $r^c_1$ and
			$r^c_2$
			prefer 
			to the hospital~$h_{v_{i_c}}$ (which is the hospital $r^c_1$ and $r^c_2$ are
matched to in $M$) are the
			edge 
			hospitals
			$h_{\{v_{i_c},
				\text{z}(v_{i_c},1)\}},\dots, h_{\{v_{i_c},
				\text{z}(v_{i_c},p)\}}$
			corresponding to 
			$v_{i_c}$ and the vertices that are adjacent to~$v_{i_c}$. Hence,
			as 
			no two color residents
			prefer the 
			same vertex hospital, there cannot exist a blocking coalition to 
			open a 
			vertex hospital. Moreover, as $V'$ is an independent set, there 
do not exist 
			two 
			adjacent vertices in $V'$ and thereby also no edge 
			hospital 
			corresponding to a vertex pair from $V'$. Thus, there does not 
			exist an edge hospital that is preferred by four color residents to 
			the 
			hospital 
			they are matched to in $M$. Thus, $M$ is stable.
			
			{\bfseries ($\Leftarrow$)} Assume that there exists a stable 
			matching $M$ in 
			the 
			constructed \HRLQ instance. First of all note that, for each color 
			$c\in [k]$, resident $s^c_*$
			needs to be
			matched 
			to a vertex hospital in~$M$, as otherwise there does not exist a
			stable matching of the residents $s_*^c$,~$s_1^c$, and~$s_2^c$ to
			the hospitals~$h^c_1$, $h^c_2$, and $h^c_3$ (see 
			\Cref{ob:counter}). Thus, for each color $c\in [k]$,
there 
			exists 
			exactly one
			vertex $v_{i}^c\in V^c$ such that the
			hospital~$h_{v_{i}^c}$ is open in $M$. We claim that $V'=\{
			v_{i}^c 
			\mid c\in[k] \wedge i\in [q] \wedge \text{ $h_{v_{i}^c}$ is open in 
			}M\}$ forms an 
independent set in $G$. For the sake of 
			contradiction, 
			let us assume that there exists a pair of vertices~$v,v'\in V'$
			with 
			$v\in 
			V^{c}$ and $v'\in V^{d}$ for two $c\neq d\in [k]$ that are
			adjacent. By 
			construction, 
			$r^{c}_1$ and $r^{c}_2$ are matched to $h_{v}$
			and 
			$r^{d}_1$ and $r^{d}_2$ are matched to~$h_{v'}$ in $M$.
			However, as 
			$v$ 
			and $v'$ are adjacent, this implies that all four 
			residents~$r^{c}_1$,~$r^{c}_2$,~$r^{d}_1$, and $r^{d}_2$
			prefer the 
edge hospital
			$h_{\{v,v'\}}$ to the hospital they are matched to in~$M$, which 
			contradicts our assumption that $M$ is a stable matching.
		\end{proof}
	Again, the \HRLQ instance constructed in the reduction falls under 
	\Cref{ob:equiv}, which implies that the three models are W[1]-hard 
	parameterized by $n$ even if stability only forbids the existence of 
	blocking coalitions. 
	Compared to \Cref{th::NP-compl}, the hardness statement from \Cref{th:W-n} does not
	bound the number of hospitals accepted by each resident.
	In 
	fact, all three problems are
	fixed-parameter tractable by the combined parameter number of residents plus maximum number of hospitals accepted by a resident,
	as the size of the instance (ignoring hospitals which no resident accepts) 
	is bounded in a function of this parameter.
	
	In the following \Cref{se:hosp}, among others, we address the computational 
	complexity of deciding whether 
	there exists a stable matching where exactly a given subset of hospitals is 
	open and the parameterized complexity of deciding whether there exists a 
	stable matching 
	opening~$m_{\text{open}}$~hospitals parameterized by $m_{\text{open}}$. 
	Obviously, both questions can be analogously asked for residents. However, 
	the reduction from \Cref{th:W-n} proves that deciding the existence of a 
	stable matching assigning all residents is NP-hard and W[1]-hard 
	parameterized by the number of residents. Thus, both questions asked for 
	the residents are computationally intractable for all our three models. 
	
	\subsection{Influence of Hospitals}\label{se:hosp} 
	After studying the parameterization by the number of residents,
	we turn to the number of hospitals and several closely related parameters.
	We start by considering
	the problem of finding a stable matching opening exactly a given
	set of hospitals.
	\subsubsection{Which hospitals should be open?} \label{se:hopen}
	It is possible 
	to think of finding a stable matching as a two-step 
	process. First,
	decide which hospitals are open and second,
	compute a stable matching between the residents and the selected set of open hospitals
	respecting all quotas.
	This observation leads to the question what happens if the first step
	has been already done, e.g., by an
	oracle or by some
	authority, and we are left with the task of finding a stable matching where
	exactly a given set
	of hospitals is open. We show that while for \HRLQ
	and
	\HRLUQ
	this problem is solvable in
	polynomial time, it is NP-hard for \HRLUQI.
	
	As already observed in \Cref{se:obs}, in an \HRLQ instance $(H,R)$, every stable matching assigns all
	residents to their most preferred open hospital.
	Thereby, deciding whether there exists a stable matching where exactly a
	given set $H_{\open}\subseteq H$ of hospitals is open reduces to
	assigning each 
	resident to her most preferred hospital in $H_{\open}$ and checking
	whether the 
	resulting matching is stable in $(H,R)$. 
	\begin{observation} \label{ob:HRLQ-H'}
		Given a subset of hospitals $H_{\open}\subseteq H$, deciding
		whether there 
		exists a stable matching in an \HRLQ instance $(H,R)$ in which exactly 
		the hospitals from $H_{\open}$ are open is solvable in
		$\mathcal{O}(nm)$ time.
	\end{observation}

	For \HRLUQ, a slightly more involved reasoning is needed, which utilizes 
	the famous Rural Hospitals Theorem 
	\cite{RePEc:ucp:jpolec:v:92:y:1984:i:6:p:991-1016,10.2307/1913160}:
	\begin{proposition} \label{pr:HRULQ-H'}
		Given a subset of hospitals $H_{\open}\subseteq H$, deciding
		whether there 
		exists a stable matching in an \HRLUQ instance $(H,R)$ in which exactly 
		the hospitals from $H_{\open}$ are open is solvable in
		$\mathcal{O}(nm)$ time.
	\end{proposition}
	\begin{proof}
		Given an \HRLUQ instance $\mathcal{I}=(H,R)$, we calculate the
		resident-optimal stable matching~$M$ of the residents $R$ to the 
		hospitals 
		$H_{\open}$ ignoring their lower quota
		by applying the Gale and Shapely algorithm~\cite{GaleS62}. We return YES if $M$ is  
		stable in $\mathcal{I}$ and obeys the 
		quotas and otherwise NO.
		If the algorithm returns YES, then this 
		answer is clearly correct.
		It remains to show that if there exists a stable matching, then the
		algorithm returns YES.
		Assume that there exists a stable
		matching $M'$ in $\mathcal{I}$ that opens exactly the hospitals 
		$H_{\open}$, and let $M$ be the computed matching.
		As $M'$ is a stable
		matching that opens exactly the hospitals from $H_{\open}$, matching 
		$M'$
		is also a stable matching
		in 
		the instance $\mathcal{I}' :=(H_{\open},R)$ without lower quotas.
		By the Rural Hospitals Theorem, it follows that every stable matching 
		in $\mathcal{I}'$---and therefore also $M$---matches the same
		number of residents as $M'$ to each hospital and thus is a feasible 
		matching in 
		$\mathcal{I}$.
		Matching $M$ does not admit a blocking pair in $\mathcal{I}$ by the 
		definition of $M$.
		Lastly, due to the resident-optimality of $M$, it follows that any 
		blocking coalition for $M$ in $\mathcal{I}$ is also a blocking 
		coalition for $M'$ in $\mathcal{I}$. Thus, as $M'$ is stable in 
		$\mathcal{I}$, matching $M$ is also stable in~$\mathcal{I}$.
	\end{proof}

	So far, our hardness results for all three models including
	\HRLQ indicated
	that the complexity of our problems mainly comes from the hospital's lower 
	quotas. In fact, for all questions we consider in this and the following 
	section, \HRLQ and 
	\HRLUQ are tractable in the same cases. Moreover, we have shown in
	\Cref{ob:HRLQ-H'} 
	and \Cref{pr:HRULQ-H'} that the difficult part of \HRLQ and 
	\HRLUQ is to 
	decide which hospitals to open and not how to assign the residents to the 
	open 
	hospitals. 
	
	In contrast to this, as proven in the following theorem, for \HALUQ, by 
	exploiting the upper quotas of the hospitals, the problem of assigning the 
	residents to the open hospitals becomes computationally hard. Generally 
	speaking, the reason for this is that we have a lot of flexibility how to 
	assign the residents here, as blocking pairs need to involve an 
	undersubscribed hospital. For instance, as done in the following reduction, 
	there may exist a hospital that every resident wants to be matched to. In 
	this case, by enforcing an upper quota on the number of residents matched 
	to such a ``popular'' hospital and as this popular hospital does not have 
	preferences that 
	could act as a ``tie-breaker'',  there exists, in principle, an exponential 
	number 
	of 
	possibilities 
	which residents to assign to the popular hospital. 
	
	Thus, in contrast to the Hospital Residents problem with lower quotas where 
	upper quotas do not really seem to add 
	complexity to the problem, as the preferences of hospitals already provide 
	some information which residents to assign to a hospital, upper quotas add 
	complexity to the House Allocation problem with lower quotas. We show this 
	in the following 
	theorem by proving 
	that \HALUQ remains NP-complete even if we know which hospitals are open in 
	a stable matching:
	
	\begin{theorem}
		\label{th:HRLUQIH'}
		Given a subset of hospitals $H_{\open}\subseteq H$, deciding
		whether there 
		exists a stable matching in an \HRLUQI instance $(H,R)$ in which 
		exactly 
		the hospitals from $H_{\open}$ are open
		is NP-complete, even if
		$|H_{\open}|=4$, each resident 
		accepts at most five hospitals, all hospitals have lower quota at most 
		two, and we know that if there exists a stable matching, then it opens
		exactly the hospitals from~$H_{\open}$.
	\end{theorem}
	\begin{proof}
       	We reduce from the NP-hard \textsc{Independent Set} problem on 
       	3-regular graphs, that is, 
       	graphs were all vertices have exactly three neighbors 
       	\cite{DBLP:books/fm/GareyJ79}.
		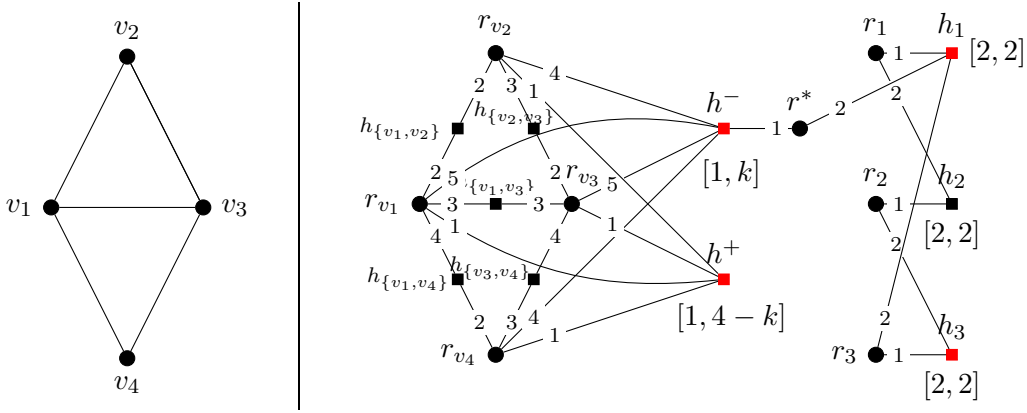
\begin{figure}
			\begin{minipage}{0.25\textwidth}
				
				\begin{center}
					\begin{tikzpicture}
					\node (ver-dist) at (0, 1.5) {};
					\node (hor-dist) at (5, 0) {};
					\node (hd) at (3, 0) {};
					
					\node[vertex, label=180:$v_1$] (r1) at (0, 0) {};
					\node[vertex, label=90:$v_2$] (r2) at ($(r1) + (1,2)$) {};
					\node[vertex, label=0:$v_3$] (r3) at ($(r1) + (2, 0)$) {};
					\node[vertex, label=270:$v_4$] (r5) at ($(r1) + (1, -2)$) 
					{};
					
					\draw (r1) -- (r2) -- (r3) -- (r1);
					\draw (r2) -- (r3) -- (r5);
					\draw (r5) -- (r1);
					
					\end{tikzpicture}
				\end{center}
			\end{minipage}\vline
			\begin{minipage}{0.65\textwidth}
				\begin{center}
					\begin{tikzpicture}
					\node (ver-dist) at (0, 1.5) {};
					\node (hor-dist) at (5, 0) {};
					\node (hd) at (3, 0) {};
					
					\node[vertex, label=180:$r_{v_1}$] (r1) at (0, 0) {};
					\node[vertex, label=90:$r_{v_2}$] (r2) at ($(r1) + (1,2)$) 
					{};
					\node[vertex, label={[xshift=0.15cm]90:$r_{v_3}$}] (r3) at 
					($(r1) + (2, 0)$) {};
					\node[vertex, label=180:$r_{v_4}$] (r4) at ($(r1) + 
					(1,-2)$) {};
					
					\node[squared-vertex, label=180:\scriptsize$h_{\{v_1, 
					v_2\}}$] (e12) at ($0.5*(r1)+ 0.5*(r2)$) {};
					\node[squared-vertex, 
					label={[yshift=-0.1cm]90:\scriptsize$h_{\{v_1, v_3\}}$}] 
					(e13) at ($0.5*(r1)+ 0.5*(r3)$) {};
					\node[squared-vertex, 
					label={[xshift=0.5cm,yshift=0.2cm]180:\scriptsize$h_{\{v_2, 
							v_3\}}$}] (e23) at ($0.5*(r3)+ 0.5*(r2)$) {};
					\node[squared-vertex, 
					label={[xshift=0.1cm]180:\scriptsize$h_{\{v_1, v_4\}}$}] 
					(e15) at ($0.5*(r5)+ 0.5*(r1)$) {};
					\node[squared-vertex, 
					label={[xshift=0.18cm,yshift=0.15cm]180:\scriptsize$h_{\{v_3,
							v_4\}}$}] 
					(e34) at ($0.5*(r3)+ 0.5*(r4)$) {};
					
					\draw (e12) edge node[pos=0.6, fill=white, inner sep=2pt] 
					{\scriptsize $2$} (r1);
					\draw (e12) edge node[pos=0.6, fill=white, inner sep=2pt] 
					{\scriptsize $2$} (r2);
					\draw (e23) edge node[pos=0.6, fill=white, inner sep=2pt] 
					{\scriptsize $3$} (r2);
					\draw (e23) edge node[pos=0.6, fill=white, inner sep=2pt] 
					{\scriptsize $2$} (r3);
					\draw (e13) edge node[pos=0.6, fill=white, inner sep=2pt] 
					{\scriptsize $3$} (r1);
					\draw (e13) edge node[pos=0.6, fill=white, inner sep=2pt] 
					{\scriptsize $3$} (r3);
					\draw (e15) edge node[pos=0.6, fill=white, inner sep=2pt] 
					{\scriptsize $4$} (r1);
					\draw (e15) edge node[pos=0.6, fill=white, inner sep=2pt] 
					{\scriptsize $2$} (r4);
					\draw (e34) edge node[pos=0.6, fill=white, inner sep=2pt] 
					{\scriptsize $4$} (r3);
					\draw (e34) edge node[pos=0.6, fill=white, inner sep=2pt] 
					{\scriptsize $3$} (r4);
					
					\node[red, squared-vertex, label=90:$h^+$, 
					label={[yshift=-0.1cm, 
						xshift=0.1cm]270:{$[1, 4-k]$}}] (hg) at (4, -1) {};
					\node[red, squared-vertex, label=90:$h^-$, 
					label={[yshift=-0.2cm, xshift=0.1cm]270:{$[1, k]$}}] (hb) 
					at (4, 1) {};
					
					\draw (hb) edge[bend right=25] node[pos=0.9, fill=white, 
					inner sep=2pt] {\scriptsize $5$} (r1);
					\draw (hb) edge node[pos=0.76, fill=white, inner sep=2pt] 
					{\scriptsize $4$} (r2);
					\draw (hb) edge node[pos=0.76, fill=white, inner sep=2pt] 
					{\scriptsize $5$} (r3);
					\draw (hb) edge node[pos=0.85, fill=white, inner sep=2pt] 
					{\scriptsize $4$} (r4);
					
					\draw (hg) edge[bend left = 20] node[pos=0.9, fill=white, 
					inner sep=2pt] {\scriptsize $1$} (r1);
					\draw (hg) edge node[pos=0.85, fill=white, inner sep=2pt] 
					{\scriptsize $1$} (r2);
					\draw (hg) edge node[pos=0.76, fill=white, inner sep=2pt] 
					{\scriptsize $1$} (r3);
					\draw (hg) edge node[pos=0.76, fill=white, inner sep=2pt] 
					{\scriptsize $1$} (r4);
					
					\node[vertex, label=90:$r^*$] (rs) at (5, 1) {};
					\node[vertex, label=180:$r_3$] (re3) at (6, -2) {};
					\node[vertex, label=90:$r_1$] (re1) at (6, 2) {};
					\node[vertex, label=90:$r_2$] (re2) at (6, 0) {};
					
					\node[red, squared-vertex, label=90:$h_1$, 
					label=0:{$[2,2]$}] (h1) at (7, 2) {};
					\node[squared-vertex, label=90:$h_2$, label=270:{$[2,2]$}] 
					(h2) at (7, 0) {};
					\node[red, squared-vertex, label=90:$h_3$, 
					label=270:{$[2,2]$}] (h3) at (7, -2) {};
					
					\draw (h1) edge node[pos=0.76, fill=white, inner sep=2pt] 
					{\scriptsize $1$} (re1);
					\draw (h2) edge node[pos=0.76, fill=white, inner sep=2pt] 
					{\scriptsize $2$} (re1);
					\draw (h2) edge node[pos=0.76, fill=white, inner sep=2pt] 
					{\scriptsize $1$} (re2);
					\draw (h3) edge node[pos=0.76, fill=white, inner sep=2pt] 
					{\scriptsize $2$} (re2);
					\draw (h3) edge node[pos=0.76, fill=white, inner sep=2pt] 
					{\scriptsize $1$} (re3);
					\draw (h1) edge node[pos=0.9, fill=white, inner sep=2pt] 
					{\scriptsize $2$} (re3);
					
					\draw (h1) edge node[pos=0.76, fill=white, inner sep=2pt] 
					{\scriptsize $2$} (rs);
					\draw (hb) edge node[pos=0.76, fill=white, inner sep=2pt] 
					{\scriptsize $1$} (rs);
					\end{tikzpicture}
				\end{center}
				
			\end{minipage}
			
			\caption{An example for the reduction showing NP-hardness of 
				\HALUQ with four open hospitals from \Cref{th:HRLUQIH'}.
				For the sake of illustration, the input graph~$T$ (depicted in 
				the left picture) is not 3-regular.
				The output of the reduction is depicted in the right picture.
				The lower and upper quota two for each edge hospital is not 
				drawn 
				for the sake of readability.
				The set $H_{\open}$ is marked in red.
			}
			\label{fig:hardness-open-hospitals}
		\end{figure}
	
		\textbf{Construction:} Let $(G=(V,E),k)$ be an instance of \textsc{Independent
		Set}, where $G$ is a 3-regular graph.
		For each~$v\in V$, let $e^v_1$, $e^v_2$, and $e^v_3$ be a list of all 
		edges incident to $v$. We introduce a \emph{good hospital}~$h^+$ with 
		lower
		quota one and upper quota $n-k$ and a \emph{bad hospital}~$h^-$ with
		lower
		quota one and upper quota $k$. Moreover, we introduce for each edge
		$e\in E$ an \emph{edge hospital} $h_e$ with lower and upper quota two.
		
		Turning to the residents, we introduce for each vertex $v\in V$ a
		\emph{vertex resident} $r_v$ with the following preferences: 
		$$r_v:h^+\succ h_{e^v_1} 
		\succ h_{e^v_2} \succ h_{e^v_3} \succ h^-.$$
		Finally, we introduce a
		penalizing component ensuring that no vertex resident can be 
		matched to an edge hospital. The penalizing component consists of three 
		hospitals $h_1$, $h_2$,
		$h_3$, each with lower and upper quota two, and four
		residents $r^*$, $r_1$, $r_2$, and $r_3$: 
		$$r^*:h^-\succ h_1, \qquad r_1:h_1\succ h_2, \qquad r_2:h_2\succ h_3, 
		\qquad r_3:h_3\succ h_1.$$
		Note that the penalizing component ensures that $k$ vertex residents 
		need to be matched to the bad hospital.
		We set $H_{\open}:=\{h^+, h^-, h_1,h_3\}$.
		See \Cref{fig:hardness-open-hospitals} for an example.
		Intuitively, the vertex
        residents assigned to the bad hospital in a stable 
        matching form an independent set, as two 
        vertex residents corresponding to adjacent vertices matched to the bad 
        hospital would form a blocking 
        coalition to open the respective edge hospital together.
        We now prove that there exists an independent set $V'$ of size $k$ in 
the given graph~$G$ if and only if the constructed \HALUQ instance admits a
stable matching opening exactly the hospitals from $H_{\text{open}}$.
        
		{\bfseries ($\Rightarrow$)} Let $V'\subseteq V$ be an independent set 
		of size $k$ 
		in 
		the given graph $G$. From this, we construct a stable matching $M$ 
		opening $H_{\open}$ in
		the 
		constructed instance by assigning all vertex residents corresponding to 
		vertices in $V'$ to the bad hospital, all other vertex residents to
		the good hospital, $(h_1,\{r^*,r_1\})$, and
		$(h_3,\{r_2,r_3\})$.
        As no hospital is undersubscribed, it remains to
		argue why there does not exists a blocking coalition to
		open a closed hospital: As $r_1$ is matched to her top-choice, there does not exist a
		blocking coalition to open $h_2$. Moreover, for no edge are both 
residents
		corresponding to the two endpoints matched to the bad hospital~$h^-$ 
as
		$V'$ is an independent set. Thus, no blocking coalition to open an edge 
		hospital exists.
		
	{\bfseries ($\Leftarrow$)} Let $M$ be a stable matching in the constructed 
		instance. First of all note that~$r^*$ needs to be matched to $h_1$, as 
		otherwise~$r_1$,~$r_2$, and~$r_3$ cannot be assigned to the 
		hospitals~$h_1$,~$h_2$, and~$h_3$ in a stable way (see 
		\Cref{ob:counter}). 
		This implies that $k$ vertex residents need to be matched to the bad 
		hospital~$h^-$, implying that the remaining $n-k$ vertex residents are matched to the good hospital~$h^+$.
		Thus, all edge hospitals are closed, since the residents which accept 
		them need to be either matched to the good or bad hospital as discussed 
		above.
		The
		$k$~vertices~$\{v_1, \dots, v_k\}$ corresponding to the residents 
matched to $h^-$ in $M$ form an
		independent set, as a pair of vertex residents $r_{v_i}$ and
		$r_{v_{j}}$ with~$\{v_i,v_j\}\in E$ both matched to the bad hospital 
		forms a blocking coalition to open $h_{
		\{v_i,v_j\}}$ in~$M$.
	\end{proof}
	
		Note that this reduction can be easily adapted to yield NP-completeness 
		for \HRLUQI restricted to instances with upper quota at most two by 
		splitting the hospitals $h^+$ and~$h^-$ into multiple hospitals with 
		upper and lower quota one.
		However, in this case, the bounds on the size of $H_{\open}$ and
		the number 
		of hospitals acceptable to a single resident do not hold any more.

	\subsubsection{Parameterization by the number of hospitals (with non-unit 
	lower quota)}
	Together
	with the number $n$ of residents, the number $m$ of hospitals is a very
	important and straightforward structural parameter of the studied problems.
	As in some applications this 
	parameter is much smaller than the number of residents, checking for
	fixed-parameter tractability is of special interest here.
	
	For both \HRLQ and \HRLUQ, it is possible to iterate over all possible 
	subsets $H_{\open}\subseteq H$ of hospitals and use
	\Cref{ob:HRLQ-H'} and 
	\Cref{pr:HRULQ-H'}, respectively, to decide whether there exists a stable 
	matching in which exactly the hospitals from $H_{\open}$ are open.
	Let 
	$H^{\quota}\subseteq H$ denote the set of hospitals with non-unit
	lower quota. 
	In fact, it is only necessary to iterate over all possible 
	subsets~$H_{\open}\subseteq H^{\quota}$ with non-unit lower quota and 
	subsequently add all hospitals 
	with
	lower quota one to $H_{\open}$. For each resulting set $H_{\open}$, we apply
	the procedures 
	described in
	\Cref{ob:HRLQ-H'} and \Cref{pr:HRULQ-H'} to compute a matching which we 
	then check for stability and feasibility. 
	
	\begin{corollary} \label{co:mFPT}
\HRLQ and \HRLUQ are solvable in
		$\mathcal{O}(n m \cdot 2^{m_{\quota}})$ time, where $m_{\quota}$ is the 
		number of hospitals with non-unit lower quota.
	\end{corollary}

	Turning to \HRLUQI, despite the fact that it is NP-complete to 
	decide whether there exists a stable matching even if the set of open 
	hospitals is given,
	\HRLUQI parameterized by the number~$m$ of hospitals turns out to be
	fixed-parameter tractable. The algorithm utilizes that the number of 
	different resident types in a \HRLUQI instance can be bounded in a function 
of 
	$m$, as a resident is fully 
	characterized by her 
	preferences over hospitals. This observation can be used to construct an 
	integer linear program (ILP) where the
	number of variables is bounded in a function of $m$. 
	Employing
	Lenstra's algorithm 
	\cite{DBLP:journals/mor/Kannan87,DBLP:journals/mor/Lenstra83}
	shows that the problem is fixed-parameter tractable parameterized by
	the number of hospitals. 	
	\begin{proposition}
		\label{pr:HRLUQI-FPTM}
		Parameterized by the number $m$ of hospitals, \HALUQ is
		fixed-parameter 
		tractable.
	\end{proposition}
	\begin{proof}
		Let $(H,R)$ be a given \HALUQ instance.
		Note that as hospitals are indifferent among all residents, a
		resident $r\in R$ is fully characterized by her preferences over
		hospitals from~$H$. Thereby, the number of resident types is bounded by 
		the 
		number of ordered subsets of $m$ elements which is~$\mathcal{O}(m \cdot m!)$ (one can create each possible preference by taking the first $\ell$ hospitals from some permutation of the hospitals).
		Let
		$t_1$, \dots, $t_q$ be a list of all resident types and, for $i\in 
		[q]$, let $A(t_i)$ denote the set of hospitals which residents of type $t_i$ 
		accept. For two hospitals 
$h\neq h'\in H$, we write $h\succ_{t_i} h'$ if residents of type $t_i$ prefer 
$h$ 
to $h'$. For each~$i\in [q]$, let $n_i$ denote the number of residents in the given
		instance of type $t_i$. We design an ILP
		solving the given \HALUQ instance as follows.
		We introduce a variable
		$x_{i,h}$ for each hospital~$h\in H$ and
		each $i\in [q]$ representing the number of residents of type
		$t_i$ assigned to hospital~$h$. Moreover, for each hospital~$h\in H$, we introduce a binary
		variable~$o_h$ which is $1$ if $h$
		is
		open and $0$ otherwise.
		Furthermore, to prevent blocking pairs, we also use an additional 
		binary variable $y_h$ for each hospital $h\in H$ which shall be $1$ if 
		and only if $h$ is undersubscribed.
		In the following, we extend our notation by also writing~$h\succ_{t_i} 
		h'$ if $h'\notin A(t_i)$ and
		$h\in A(t_i)$ for a
		resident type~$t_i$. Using this notation, the problem can be
		solved using the
		following ILP:
		\begin{align}
			\sum_{i\in [q]} x_{i,h} +y_hn\geq o_hu(h), \qquad & \forall h \in H \label{cond:undersubscribed}\\
			x_{i,h}+\sum_{\substack{h'\in H \\ h'\succ_{t_i} h}} x_{i,h'} + (1 
			- y_h) n\geq
			n_i, \qquad & \forall i \in
			[q], h\in A(t_i) \label{ILP:no-bp}\\
			\sum_{\substack{i\in [q], h'\in H: \\ h\succ_{t_i} h'}} 
			x_{i,h'}\leq 
			l(h)+o_hn, \qquad & \forall h\in H \label{ILP:no-bc}\\
			o_h l(h)\leq \sum_{i\in [q]} x_{i,h}  \leq o_h u(h), \qquad& 
			\forall h\in 
			H \label{ILP:quotas}\\
			\sum_{h\in H} x_{i,h} \leq n_i, \qquad& \forall i\in [q] 			\label{ILP:num-residents}
			\\
			x_{i,h}  = 0, \qquad &\forall i \in [q], h\in H\setminus
			A(t_i) 
			\label{ILP:acceptablitiy}\\
			x_{i,h}\in\{0,1,\dots,n_i\}, \qquad  o_h \in \{0, 1\}, \qquad y_h \in \{0, 1\}, \qquad &\forall i\in [q],
			\forall h \in H 			
			\label{ILP:integrality}
    \end{align}
    Condition (\ref{cond:undersubscribed}) ensures $y_h = 1 $ holds for every undersubscribed hospital.
			Condition (\ref{ILP:no-bp}) ensures that no blocking pair between a 
			resident of type~$t_i$ and hospital~$h\in H$ exists:
			If $y_h  = 0$ (i.e.\, $h$ is closed or full), then $h$ is not part 
			of a blocking pair and the inequality is fulfilled.
			Otherwise the inequality
			enforces that for each resident type $t_i$ and hospital $h\in H$, 
			all residents of type~$t_i$ are assigned to $h$ or
			hospitals they prefer to $h$.
			Condition (\ref{ILP:no-bc}) ensures that no blocking coalition exists by
			enforcing that for all closed hospitals $h\in H$ the number of 
			residents that
			are assigned to hospitals they find worse than $h$ is below $l(h)$.
			Conditions (\ref{ILP:quotas})-(\ref{ILP:integrality}) ensure that 
			the matching encoded in the
			variables is a feasible assignment by checking whether the number 
			of resident that are assigned to a hospital is either zero or 
			between its lower and upper quota, by enforcing that for each 
			resident type the number of assigned residents of this type is 
			smaller or equal to the number of residents of this type in the 
			instance and by enforcing that no resident is 
			assigned to a hospital that she does not accept.
			
			Thus, the ILP admits a feasible solution if and only if the given 
\HALUQ instance admits a stable matching. As the number of variables used in 
the ILP lies in $\mathcal{O}(m^2 \cdot m!)$, it is
possible to apply Lenstra's algorithm
		\cite{DBLP:journals/mor/Kannan87,DBLP:journals/mor/Lenstra83}
		to solve the
		problem in $\mathcal{O}(f(m)\cdot n)$ time for some computable 
		function~$f$.
	\end{proof}	

	However, it is not
	possible to follow a similar approach to construct a fixed-parameter 
	tractable algorithm for the the number $m_{\text{quota}}$ of hospitals with 
	non-unit lower quota. In 
	fact,~\HRLUQI is NP-complete even for only three
	hospitals 
	with non-unit lower quota. 
	\begin{proposition}\label{pr:HRLUQI-NP-const}
	 \HRLUQI is NP-complete, even if only 
		three hospitals have 
		lower and 
		upper quota two and all other hospitals have upper quota one.
	\end{proposition}
    \begin{proof}
				We reduce from the
			NP-hard  
			\textsc{Clique} problem \cite{DBLP:books/fm/GareyJ79}, where given 
			a graph $G$ and an integer~$k$, the
			task is to decide whether there exists a subset of vertices of size 
			$k$ in $G$ that are all pairwise adjacent. Given an 
			instance~$((V=\{v_1,\dots v_n \},E),k)$ of
			\textsc{Clique}, we construct an
			\HRLUQI  instance as 
			follows. For each vertex $v\in V$, let $e^v_1,\dots, e^v_{p_v}$ be 
			a 
			list of all edges incident to $v$. 
			
			\textbf{Construction:} For each vertex $v\in V$, we add a \emph{vertex 
			hospital} 
			$h_v$ with 
			lower and 
			upper quota one. Similarly, for each edge $e\in E$, we add an 
			\emph{edge 
				hospital} with lower and upper quota one. Moreover, we add $k$ 
			\emph{vertex 
				selection hospitals} $h^\text{vert}_1,\dots, h^\text{vert}_{k}$ 
				and
			$\binom{k}{2}$ \emph{edge selection hospitals} 
			$h^\text{edge}_1,\dots,
			h^\text{edge}_{\binom{k}{2}}$, all with lower and upper quota one.
			
			Turning to the residents, for each vertex $v\in V$, we introduce a 
			\emph{vertex resident} $r_v$ with the following preferences:
			$$r_v : h^\text{vert}_1 \succ \dots \succ h^\text{vert}_{k}\succ 
			h_{e^v_1} \succ \dots \succ h_{e^v_{p_v}} \succ h_v.$$
			Moreover, we introduce an \emph{edge resident} $r_e$ for 
			each edge 
			$e\in E$ with the following preferences:
			$$h^\text{edge}_1\succ \dots \succ 
			h^\text{edge}_{\binom{k}{2}}\succ 
			h_e.$$
			Notably, in every stable matching, all vertex selection hospitals 
			will be filled with vertex residents and all edge selection 
			hospitals 
			will be 
			filled with edge residents. 
			In addition, we add for each~$i \in [k]$ a \emph{filling agent} 
			$r_i^\text{fill}$ with
			the 
			following 
			preferences:
			$$r_i^\text{fill} : h_{v_1}\succ \dots \succ h_{v_n} .$$			
			Finally, we introduce a penalizing component consisting of three 
			hospitals $h_1$, $h_2$, and~$h_3$ all with lower and upper quota 
			two and four residents $r^*$, 
			$r_1$, $r_2$, and $r_3$:
			$$r^*:h_{v_1}\succ \dots \succ h_{v_n}\succ h_1, \qquad 
			r_1:h_1\succ 
			h_2,
			\qquad r_2:h_2\succ h_3, 
			\qquad r_3:h_3\succ h_1.$$
			See \Cref{fig:hardness-few-hospitals} for an example of the reduction.
			The general idea behind the reduction is that the penalizing 
			component enforces that a vertex resident or a filling resident is 
			assigned to every vertex hospital.
			Thus, no vertex resident can be matched to an edge hospital, and therefore,
			$\binom{k}{2}$ 
            edge hospitals need to be closed in any stable matching.
            The vertex residents
corresponding to the
			end points of the closed edge hospitals
			need to be matched to a vertex selection hospital (since otherwise the edge hospital and vertex resident form a blocking pair), implying that
			the vertices matched to vertex selection hospitals form a 
			clique.
			We now show that there exists a clique of size~$k$ in the given
graph if and only if there exists a stable matching in the constructed \HALUQ 
instance.
			
			\begin{figure}
			  \begin{minipage}{0.2\textwidth}

			  \begin{center}
			    \begin{tikzpicture}
                 \node[vertex, label=180:$v_1$] (r1) at (0, 0) {};
			      \node[vertex, label=90:$v_2$] (r2) at ($(r1) + (1,2)$) {};
			      \node[vertex, label=0:$v_3$] (r3) at ($(r1) + (2, 0)$) {};
			      \node[vertex, label=270:$v_4$] (r5) at ($(r1) + (1, -2)$) {};

			      \draw (r1) -- (r2) -- (r3) -- (r1);
			      \draw (r2) -- (r3) -- (r5);
			      \draw (r5) -- (r1);

        \end{tikzpicture}
\end{center}
			  \end{minipage} \hspace*{0.6cm}\vline \hspace*{0.15cm}
			  \begin{minipage}{0.65\textwidth}
\begin{center}
			    \begin{tikzpicture}
                  \node (ver-dist) at (0, 1.5) {};
                  \node (hor-dist) at (1.5, 0) {};
                  \node (hd) at (3, 0) {};

			      \node[vertex, label=270:$r_{v_1}$] (r1) at (0, 0) {};
			      \node[vertex, label=270:$r_{v_2}$] (r2) at ($(r1) + (ver-dist)$) {};
			      \node[vertex, label={[yshift=0.15cm]90:$r_{v_3}$}] (r3) at 
			      ($(r2) + (ver-dist)$) {};
			      \node[vertex, label=90:$r_{v_4}$] (r4) at ($(r3) + (ver-dist)$) {};
			      
			      \node[vertex, 
			      label={[xshift=-0.15cm,yshift=-0.1cm]90:$h_{v_1}$}] (s1) at 
			      ($(r1) + ( hor-dist)$) {};
			      \node[vertex, label=90:$h_{v_2}$] (s2) at ($(s1) + (ver-dist)$) {};
			      \node[vertex, 
			      label={[xshift=0.15cm,yshift=-0.1cm]90:$h_{v_3}$}] (s3) at 
			      ($(s2) + (ver-dist)$) {};
			      \node[vertex, label={[yshift=-0.1cm]90:$h_{v_4}$}] (s4) at 
			      ($(s3) + (ver-dist)$) {};

                  \node[squared-vertex, label=0:$h_1^{\ver}$] (hv1) at ($(r4) + 
                  (hor-dist) + (ver-dist)$) {};
                  \node[squared-vertex, label=0:$h_2^{\ver}$] (hv2) at ($(r1) - 
                  (ver-dist) + (hor-dist)$) {};
        
        \draw (hv1) edge node[pos=0.92, fill=white, inner sep=2pt] {\scriptsize $1$} (r1);
        \draw (hv2) edge node[pos=0.76, fill=white, inner sep=2pt] {\scriptsize $2$} (r1);
        
        \draw (hv1) edge node[pos=0.92, fill=white, inner sep=2pt] {\scriptsize $1$} (r2);
        \draw (hv2) edge node[pos=0.76, fill=white, inner sep=2pt] {\scriptsize $2$} (r2);
        
        \draw (hv1) edge node[pos=0.76, fill=white, inner sep=2pt] {\scriptsize $1$} (r3);
        \draw (hv2) edge node[pos=0.92, fill=white, inner sep=2pt] {\scriptsize $2$} (r3);
        
        \draw (hv1) edge node[pos=0.76, fill=white, inner sep=2pt] {\scriptsize $1$} (r4);
        \draw (hv2) edge node[pos=0.92, fill=white, inner sep=2pt] {\scriptsize $2$} (r4);
                  
			      \node[squared-vertex, 
			      label={[xshift=-0.1cm,yshift=0.1cm]270:$h_{\{v_1, v_2\}}$}] 
			      (e12) at ($(r1) - (hor-dist)$) {};
			      \node[squared-vertex, label={[xshift=-0.175cm]270:$h_{\{v_1, 
			      v_3\}}$}] (e13) at ($(e12) + (ver-dist)$) {};
			      \node[squared-vertex, label={[xshift=-0.175cm]270:$h_{\{v_2, 
			      v_3\}}$}] (e23) at ($(e13) + (ver-dist)$) {};
			      \node[squared-vertex, 
			      label={[xshift=-0.1cm,yshift=-0.05cm]90:$h_{\{v_1, v_4\}}$}] 
			      (e14) at ($(e23) + (ver-dist)$) {};
			      \node[squared-vertex, 
			      label={[xshift=-0.1cm,yshift=-0.05cm]90:$h_{\{v_3, v_4\}}$}] 
			      (e34) at ($(e14)+ (ver-dist)$) {};
			      
			      \node[vertex, label={[yshift=0.05cm]270:$r_{\{v_1, v_2\}}$}] 
			      (r12) at ($(e12) - (hor-dist)$) {};
			      \node[vertex, label={[yshift=0.05cm]270:$r_{\{v_1, v_3\}}$}] 
			      (r13) at ($(e13) - (hor-dist)$) {};
			      \node[vertex, 
			      label={[xshift=0.1cm,yshift=-0.05cm]90:$r_{\{v_2, v_3\}}$}] 
			      (r23) at ($(e23) - (hor-dist)$) {};
			      \node[vertex, label={[yshift=-0.05cm]90:$r_{\{v_1, v_4\}}$}] 
			      (r14) at ($(e14) - (hor-dist)$) {};
			      \node[vertex, 
			      label={[xshift=0.1cm,yshift=-0.05cm]90:$r_{\{v_3, v_4\}}$}] 
			      (r34) at ($(e34) - (hor-dist)$) {};

                  \node[squared-vertex, 
                  label={[xshift=-0.1cm]90:$h_1^{\edge}$}] (he1) at ($(r23) - 
                  (hor-dist)$) {};
        
        \draw (he1) edge node[pos=0.76, fill=white, inner sep=2pt] {\scriptsize $1$} (r12);
        
        \draw (he1) edge node[pos=0.76, fill=white, inner sep=2pt] {\scriptsize $1$} (r13);
        
        \draw (he1) edge node[pos=0.76, fill=white, inner sep=2pt] {\scriptsize $1$} (r23);
        
        \draw (he1) edge node[pos=0.76, fill=white, inner sep=2pt] {\scriptsize $1$} (r34);
        
        \draw (he1) edge node[pos=0.76, fill=white, inner sep=2pt] {\scriptsize $1$} (r14);

        \draw (e12) edge node[pos=0.76, fill=white, inner sep=2pt] {\scriptsize 
        $2$} (r12);
        \draw (e13) edge node[pos=0.76, fill=white, inner sep=2pt] {\scriptsize 
        $2$} (r13);
        \draw (e23) edge node[pos=0.76, fill=white, inner sep=2pt] {\scriptsize 
        $2$} (r23);
        \draw (e14) edge node[pos=0.76, fill=white, inner sep=2pt] {\scriptsize 
        $2$} (r14);
        \draw (e34) edge node[pos=0.76, fill=white, inner sep=2pt] {\scriptsize 
        $2$} (r34);

        \node[vertex, label={[xshift=-0.1cm]45:$r^*$}] (rs) at ($(s3) + 0.5*(hor-dist) - 0.5*(ver-dist)$) {};
        \node[vertex, label=180:$r_3$] (re3) at ($(s3) + (hor-dist) + 0.5*(ver-dist)$) {};
        \node[vertex, label=90:$r_2$] (re2) at ($(re3) - (ver-dist)$) {};
        \node[vertex, label=90:$r_1$] (re1) at ($(re2) -(ver-dist)$) {};

        \node[blue, squared-vertex, label=90:$h_1$] (h1) at ($(re1)+ 
        (hor-dist)$) {};
        \node[blue, squared-vertex, label=90:$h_2$] (h2) at ($(re2)+ 
        (hor-dist)$) {};
        \node[blue, squared-vertex, label=90:$h_3$] (h3) at ($(re3)+ 
        (hor-dist)$) {};

        \draw (h1) edge node[pos=0.76, fill=white, inner sep=2pt] {\scriptsize $1$} (re1);
        \draw (h2) edge node[pos=0.76, fill=white, inner sep=2pt] {\scriptsize $2$} (re1);
        \draw (h2) edge node[pos=0.76, fill=white, inner sep=2pt] {\scriptsize $1$} (re2);
        \draw (h3) edge node[pos=0.76, fill=white, inner sep=2pt] {\scriptsize $2$} (re2);
        \draw (h3) edge node[pos=0.76, fill=white, inner sep=2pt] {\scriptsize $1$} (re3);
        \draw (h1) edge node[pos=0.9, fill=white, inner sep=2pt] {\scriptsize $2$} (re3);

        \draw (h1) edge node[pos=0.76, fill=white, inner sep=2pt] {\scriptsize $5$} (rs);
        
        \draw (s1) edge node[pos=0.76, fill=white, inner sep=2pt] {\scriptsize $1$} (rs);
        \draw (s2) edge node[pos=0.76, fill=white, inner sep=2pt] {\scriptsize $2$} (rs);
        \draw (s3) edge node[pos=0.76, fill=white, inner sep=2pt] {\scriptsize $3$} (rs);
        \draw (s4) edge node[pos=0.76, fill=white, inner sep=2pt] {\scriptsize $4$} (rs);
        
        \draw (e12) edge node[pos=0.76, fill=white, inner sep=2pt] {\scriptsize $3$} (r1);
        \draw (e12) edge node[pos=0.82, fill=white, inner sep=2pt] {\scriptsize 
        $3$} (r2);
        \draw (e23) edge node[pos=0.76, fill=white, inner sep=2pt] {\scriptsize $4$} (r2);
        \draw (e23) edge node[pos=0.76, fill=white, inner sep=2pt] {\scriptsize 
        $4$} (r3);
        \draw (e13) edge node[pos=0.76, fill=white, inner sep=2pt] {\scriptsize $4$} (r1);
        \draw (e13) edge node[pos=0.76, fill=white, inner sep=2pt] {\scriptsize 
        $3$} (r3);
        \draw (e34) edge node[pos=0.76, fill=white, inner sep=2pt] {\scriptsize $5$} (r3);
        \draw (e34) edge node[pos=0.76, fill=white, inner sep=2pt] {\scriptsize 
        $4$} (r4);
        \draw (e14) edge node[pos=0.9, fill=white, inner sep=2pt] {\scriptsize 
        $5$} (r1);
        \draw (e14) edge node[pos=0.76, fill=white, inner sep=2pt] {\scriptsize 
        $3$} (r4);
        
        \draw (s1) edge node[pos=0.85, fill=white, inner sep=2pt] {\scriptsize $6$} (r1);
        \draw (s2) edge node[pos=0.85, fill=white, inner sep=2pt] {\scriptsize $5$} (r2);
        \draw (s3) edge node[pos=0.85, fill=white, inner sep=2pt] {\scriptsize $6$} (r3);
        \draw (s4) edge node[pos=0.85, fill=white, inner sep=2pt] {\scriptsize $5$} (r4);
        \end{tikzpicture}
\end{center}

			  \end{minipage}

			  \caption{An example for the reduction showing NP-hardness of 
			  \HALUQ with four open hospitals from \Cref{pr:HRLUQI-NP-const}.
              The input instance is~$(T,2)$, where~$T$ is depicted in the left 
              picture.
              The output of the reduction is depicted in the right picture.
              Hospitals with lower (and upper) quota two are marked in blue, while all other hospitals have lower (and upper) quota one.
              }
              \label{fig:hardness-few-hospitals}
			\end{figure}
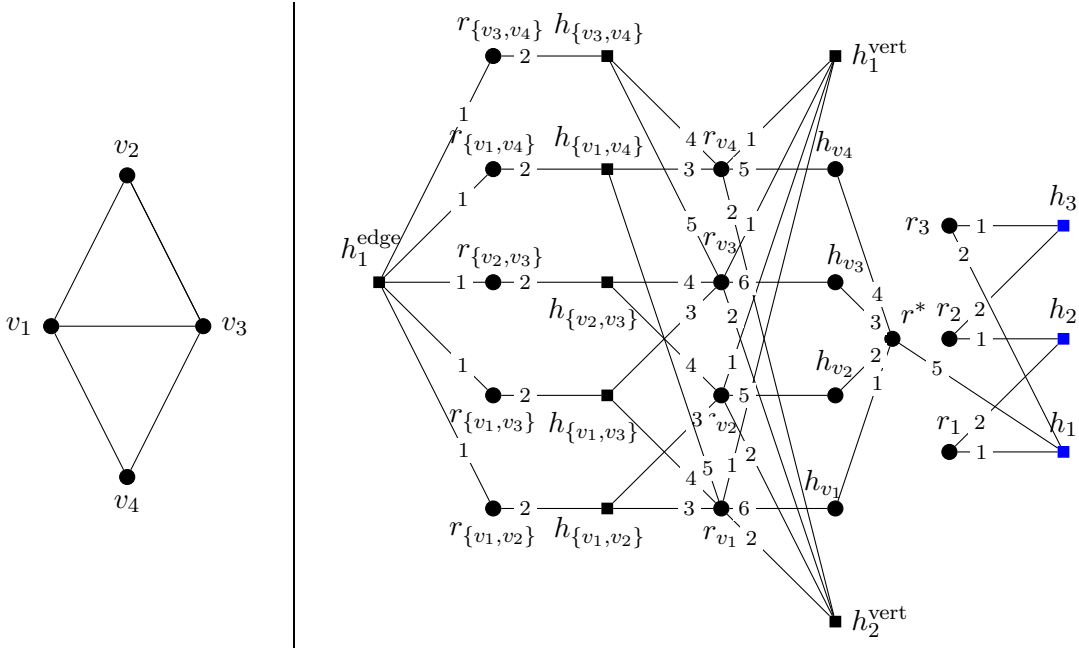

			{\bfseries ($\Rightarrow$)} Let $V'=\{v_{i_1},\dots 
			v_{i_k}\}\subseteq V$ be a 
			clique of 
			size $k$ in $(V,E)$ and $E'=\{e_{j_1}, \dots, 
			e_{j_{\binom{k}{2}}}\}$ 
			the set of all edges lying in 
			the clique. We claim that the following matching $M$ is stable:
			\begin{align*}
			M= & \{(h^\text{vert}_\ell,\{r_{v_{i_\ell}}\})\mid \ell\in 
			[k] \}\cup \{(h^\text{edge}_\ell,\{r_{e_{j_\ell}}\})\mid 
			\ell\in [{{k}\choose{2}}] \}\ \\
			& \cup  \{(h_v,\{r_v\})\mid v\in V\setminus V' \}  \cup 
			\{(h_e,\{r_e\})\mid
			e\in E\setminus E' \} \\
			& \cup \{(h_{v_{i_\ell}}, \{r_\ell^\text{fill}\}) \mid \ell \in 
			[k]\}
			\cup \{(h_1,\{r^*, r_1\})\} \cup \{(h_3,\{r_2, r_3\})\}
			\end{align*} 
			As no hospital is undersubscribed, only blocking coalitions can 
			block 
			$M$. Consequently, it is enough to iterate over all closed hospitals
			and 
			argue why there does not exist a blocking coalition to open them. 
			The 
			only closed hospitals are $h_2$ and the edge hospitals corresponding
			to 
			edges lying in the clique. As $r_1$ is matched to its most 
			preferred 
			hospital, there does not exist a blocking coalition to open $h_2$. 
			Moreover, as all edges~$e\in E$ such that $h_e$ is closed have both 
			their endpoints in~$V'$,
			both vertex residents that find $h_e$ acceptable are 
			matched to 
			vertex selection hospitals and therefore do not want to open this 
			hospital. 
			
			{\bfseries ($\Leftarrow$)} Let $M$ be a stable matching in the 
			constructed  
			\HALUQ
			instance. First of all note that $r^*$ needs to be matched to 
			$h_1$, as 
			otherwise no stable matching of the three 
			residents~$r_1$,~$r_2$, and~$r_3$ to the three hospitals $h_1$, 
$h_2$, 
			and $h_3$ 
			can exist (see \Cref{ob:counter}). Thus, all vertex hospitals need 
			to 
			be 
			full, which is only possible if all vertex residents are matched
			to 
			vertex hospitals or vertex selection hospitals. From this it 
			follows that 
			exactly  
			${k}\choose{2}$ edge hospitals need to be closed, as all edge 
			residents 
			prefer every edge selection hospital to their designated edge hospital
and no vertex resident can be matched to an edge hospital. 
			Because $M$ is stable, no vertex resident forms a blocking
			coalition 
			to open one of the ${k}\choose{2}$ closed edge hospitals.
			Therefore, all vertex residents corresponding to endpoints of the
			corresponding ${k}\choose{2}$ edges are assigned to one of
			the~$k$ 
			vertex selection hospitals. This implies that the vertices 
			corresponding to vertex residents assigned 
			to 
			vertex selection hospitals form a clique.
	\end{proof}
	As all hospitals from the \HALUQ instance constructed in the above 
	reduction have lower 
	quota at most two, this result also strengthens the NP-hardness of \HALUQ 
	from 
	\Cref{th::NP-compl} with respect to the lower and upper quota of the 
	hospitals. 

	\subsubsection{Number of open (or closed) hospitals in a stable matching} 
	\label{ss:oc}
	As there may exist stable matchings of different sizes in the studied 
	many-to-one 
	matching problems, one might want to find a matching with the 
	lowest/highest number of open hospitals. For instance, this could be 
	useful in applications where opening a hospital comes at some cost or 
	where organizers get money or receive points for each open hospital.
	The associated decision problem is
	to decide for some given number $m_{\open}$ whether
	there exists a stable matching in which exactly
	$m_{\open}$ hospitals are
	open. Obviously, it is also possible to ask the question for the dual
	parameter, i.e., the number~$m_{\text{closed}}$ of closed hospitals in a 
	stable matching. 
	
	For the \HRLUQI model, \Cref{th:HRLUQIH'} and the proof of 
	\Cref{pr:HRLUQI-NP-const}
	already imply parameterized hardness results for both parameters: By \Cref{th:HRLUQIH'}, it is NP-complete to decide whether there exists a stable matching opening four hospitals, implying paraNP-hardness for the parameter~$m_{\open}$. Moreover,
	W[1]-hardness for the parameter $m_{\text{closed}}$ follows from the 
	reduction from \Cref{pr:HRLUQI-NP-const} (as only ${k}\choose{2}$ edge 
	hospitals and one hospital from the penalizing component are closed in a 
	stable matching in the constructed instance).
	\begin{corollary} \label{c:HAco}
		Deciding whether there exists a stable matching with four open 
		hospitals in a \HRLUQI instance is NP-hard. Deciding whether there 
		exists a stable matching with~$m_{\text{closed}}$ closed hospitals in a 
		\HRLUQI instance is W[1]-hard parameterized by~$m_{\text{closed}}$, 
even 
		if one knows that otherwise no stable matching exists and only three 
		hospitals have non-unit upper quota.
	\end{corollary}
	It remains open whether $\HALUQ$
	parameterized by $m_{\text{closed}}$ is in XP or paraNP-hard.

	For the other two models, the parameterized hardness for the parameter 
	$m_{\open}$ follows from
	\Cref{th:W-n} (which shows W[1]-hardness
	parameterized by the number of residents), as every stable matching
	can open at most one hospital per resident.
	
	\begin{corollary} \label{c:op}
		Deciding whether there exists a stable matching with $m_{\open}$
		open hospitals in a \HRLQ or \HRLUQ instance is W[1]-hard parameterized 
by~$m_{\text{open}}$.  
	\end{corollary}
	Since the created instance in the reduction in the proof of \Cref{th:W-n} 
	fulfills \Cref{ob:equiv}, these hardness results still
	hold
	even 
	if we relax the stability condition to allow
	blocking pairs in a stable matching (but no blocking coalitions).	
	Turning now to the dual parameter $m_{\text{closed}}$,
	none of the previous reductions has any implication on the complexity of
	\HRLQ and \HRLUQ with respect to this parameter, as we have always used the 
	closed hospitals to encode some 
	constraints on the solution. This may lead to the initial hypothesis 
	that restricting the 
	number of closed hospitals makes the problem tractable.
	Unfortunately, it turns out that this is not the case, as we show W[1]-hardness for this parameter by
	somewhat switching the roles of open and closed hospitals.

	\begin{proposition}
		\label{thm:ha-m-closed}
		Deciding whether there exists a 
		stable matching 
		with $m_{\text{closed}}$ closed hospitals in a \HRLQ or \HRLUQ
		instance 
		is 
		W[1]-hard parameterized by $m_{\text{closed}}$, even if one knows that 
		otherwise
		there does not exist any stable matching.
	\end{proposition}
	
	\begin{proof}
		We reduce from \textsc{Multicolored Clique}. Let $G=(V,E)$ be an 
		undirected graph with a partitioning of the vertices in $k$ different 
		colors $(V^1,\dots, V^k)$. \textsc{Multicolored Clique} asks whether there exists a clique
		$V'\subseteq V$ of size $k$ containing one vertex from each color. 
		Without loss of generality, we assume that 
		there do not exist any edges between vertices of the same color.  
		Moreover, for each $c\neq d\in [k]$, let
		$E^{c,d}$ denote the set of edges with one endpoint colored
		in $c$ and the other endpoint colored in $d$, i.e., $e=\{u,v\}\in 
		E^{c,d}$ if~$(u\in V^c\wedge v\in V^{d})\vee (v\in V^c\wedge u\in
		V^{d})$.
		Parameterized by $k$, \textsc{Multicolored Clique} is 
		W[1]-hard~\cite{Pietrzak03}. Given an instance of \textsc{Multicolored 
		Clique}, we now construct an instance of \HRLQ as follows.
		
		\textbf{Construction:} We start by introducing a penalizing
		component consisting of four hospitals: $h^*$ with lower quota 
		$\binom{k}{2}+2$
		and $h_1$, $h_2$, and $h_3$ each with lower quota two. Moreover, we 
		introduce four penalizing residents $r^*$, $r_1$, $r_2$, and $r_3$ with 
		the following preferences:
		$$r^*: h^*\succ h_1, \quad r_1:h_1\succ h_2, \quad r_2: h_2\succ h_3,\quad
		r_3:h_3\succ h_1.$$
		 For each color $c\in [k]$, we insert a vertex
		hospital~$h_v$ with lower quota $n^3|V^c|$ for each vertex of this 
color $v\in
		 V^c$.
        For each pair~$(v, v')$ of different vertices of the same color (i.e., $v \neq v'$ and $v, v'\in V^c$ for some $c\in [k]$), we introduce $n^3$ residents~$r_{v,v'}^i$ ($i\in [n^3]$) with the following preferences:
		$$r_{v,v'}^i: h_v\succ h_{v'}, \qquad \forall i\in [n^3].$$
		Note that for every vertex hospital $h_v$ for $v\in V^c$, there exist 
		$n^3(|V^c|-1)$ residents with $h_v$ as their top-choice.
		In addition, for each edge
		$e=\{u,v\}\in 
		E^{c,d}$ for some $c<d\in [k]$, we introduce an edge hospital $h_e$ 
		with lower quota
		$|E^{c, d}|+1$. We add a penalizing resident~$r^*_e$ with
		the following preferences: 
		$$r_e^*:h_e\succ h^*.$$
		The penalizing component ensures that at most $\binom{k}{2}$ edge 
		hospitals are closed in a stable matching.
		Moreover,  for each edge
		$e=\{u,v\}\in 
		E^{c,d}$ for some $c<d\in [k]$ and each $e'\in 
		E^{c,d}\setminus\{e\}$,  
		we introduce an edge resident
		$r_{e,e'}$ with the following preferences: $$r_{e,e'}: h_e\succ 
		h_u\succ 
		h_v\succ h_{e'}.$$ 
		Hence, for every edge hospital $h_e$ with $e\in E^{c,d}$, there exist 
		$|E^{c,d}|-1$ residents with $h_e$ as top-choice.
		See \Cref{fig:closed-hospitals} for an example.
		We set $m_{\text{closed}}:=k+\binom{k}{2}+1$.
		In a stable matching, all but one 
		hospital corresponding to vertices of one color should be open. The 
		closed hospital 
		corresponds to the selected vertex from this color for the clique. 
		Similarly, for each color combination $c < d\in [k]$, we
		introduced 
		a gadget consisting of one hospital for each edge with this color 
		combination. In a stable matching, all but one hospital in this gadget 
		should be open and the closed hospital corresponds to the edge with 
		this color combination lying in the constructed clique. 
		We show that there exists a multicolored clique in the given graph if 
and only if there exists a stable matching with $m_{\text{closed}}$ hospitals 
in the constructed~\HRLQ instance.

\begin{figure}
	\begin{center}
		\begin{tikzpicture}
		\node[star, draw, label=180:$r_{v,w}$] (rvw) at (-4, 0) {};
		\node[star, draw, label=180:$r_{w, v}$] (rwv) at ($(rvw) - (0, 2)$) {};
		\node[star, draw, label=0:$r_{x, y}$] (rxy) at (6, 0) {};
		\node[star, draw, label=0:$r_{y, x}$] (ryx) at ($(rxy) + (0 , -2)$) {};
		
		\node[squared-vertex, label=90:$h_{v}$,label={[yshift=0.35cm]90:{$[128, 
				\infty]$}}] (hv) at ($(rvw) + (1, 0)$) {};
		\node[squared-vertex, 
		label=270:$h_{w}$,label={[yshift=-0.35cm,xshift=-0.45cm]270:{$[128, 
				\infty]$}}] (hw) at ($(hv) + (0, -2)$) {};
		\node[squared-vertex, 
		label=90:$h_{x}$,label={[yshift=0.35cm,xshift=0.1cm]90:{$[128, 
				\infty]$}}] (hx) at ($(rxy) - (1, 0)$) {};
		\node[squared-vertex, 
		label=270:$h_{y}$,label={[yshift=-0.45cm]270:{$[128, 
				\infty]$}}] (hy) at ($(hx) + (0, -2)$) {};
		
		\node[squared-vertex, label=90:$h_{\{v, 
			x\}}$,label={[yshift=0.45cm]90:{$[4, 
				\infty]$}}] (hvx) at ($(rvw) + (3, 1)$) {};
		\node[squared-vertex, label=180:$h_{\{v, 
			y\}}$,label={[yshift=0.05cm,xshift=-0.1cm]90:{$[4, 
				\infty]$}}] (hvy) at ($(hvx) - (0, 2)$) {};
		\node[squared-vertex, label=315:$h_{\{w, 
			y\}}$,label={[yshift=.cm,xshift=1.65cm]270:{$[4, 
				\infty]$}}] (hwy) at ($(hvy) - (0, 2)$) {};
		
		\node[vertex, label=90:$r_{e_1, e_2}$] (rvxvy) at (3, 2) {};
		\node[vertex, label=90:$r_{e_1, e_3}$] (rvxwy) at ($(rvxvy) - (0, 1)$) 
		{};
		\node[vertex, label=270:$r_{e_2, e_1}$] (rvyvx) at ($(rvxwy) - (0, 1)$) 
		{};
		\node[vertex, label=270:$r_{e_2, e_3}$] (rvywy) at ($(rvyvx) - (0, 1)$) 
		{};
		\node[vertex, label=270:$r_{e_3, e_1}$] (rwyvx) at ($(rvywy) - (0, 1)$) 
		{};
		\node[vertex, label=270:$r_{e_3, e_2}$] (rwyvy) at ($(rwyvx) - (0, 1)$) 
		{};
		
		\node[vertex, label={[xshift=0.18cm]180:$r_{\{v, x\}}^*$}] (rvx) at 
		($(hvx) + (-1.5, -2)$) {};
		\node[vertex, label=0:$r_{\{v, y\}}^*$] (rvy) at ($(hvy) + (-1.5, -2)$) 
		{};
		\node[vertex, label=270:$r_{\{w, y\}}^*$] (rwy) at ($(hwy) + (0, -1)$) 
		{};
		
		\draw (hv) edge node[pos=0.76, fill=white, inner sep=2pt] {\scriptsize 
		$1$} (rvw);
		\draw (hw) edge node[pos=0.76, fill=white, inner sep=2pt] {\scriptsize 
		$2$} (rvw);
		\draw (hv) edge node[pos=0.76, fill=white, inner sep=2pt] {\scriptsize 
		$2$} (rwv);
		\draw (hw) edge node[pos=0.76, fill=white, inner sep=2pt] {\scriptsize 
		$1$} (rwv);
		
		\draw (hx) edge node[pos=0.76, fill=white, inner sep=2pt] {\scriptsize 
		$1$} (rxy);
		\draw (hy) edge node[pos=0.76, fill=white, inner sep=2pt] {\scriptsize 
		$2$} (rxy);
		\draw (hx) edge node[pos=0.76, fill=white, inner sep=2pt] {\scriptsize 
		$2$} (ryx);
		\draw (hy) edge node[pos=0.76, fill=white, inner sep=2pt] {\scriptsize 
		$1$} (ryx);
		
		\draw (hvx) edge node[pos=0.67, fill=white, inner sep=2pt] {\scriptsize 
		$1$} (rvxvy);
		\draw (hvx) edge node[pos=0.6, fill=white, inner sep=2pt] {\scriptsize 
		$1$} (rvxwy);
		\draw (hvx) edge node[pos=0.7, fill=white, inner sep=2pt] {\scriptsize 
		$4$} (rvyvx);
		\draw (hvx) edge node[pos=0.9, fill=white, inner sep=2pt] {\scriptsize 
		$4$} (rwyvx);
		
		\draw (hvy) edge node[pos=0.76, fill=white, inner sep=2pt] {\scriptsize 
		$1$} (rvyvx);
		\draw (hvy) edge node[pos=0.8, fill=white, inner sep=2pt] {\scriptsize 
		$1$} (rvywy);
		\draw (hvy) edge node[pos=0.8, fill=white, inner sep=2pt] {\scriptsize 
		$4$} (rvxvy);
		\draw (hvy) edge node[pos=0.8, fill=white, inner sep=2pt] {\scriptsize 
		$4$} (rwyvy);
		
		\draw (hwy) edge node[pos=0.76, fill=white, inner sep=2pt] {\scriptsize 
		$1$} (rwyvx);
		\draw (hwy) edge node[pos=0.76, fill=white, inner sep=2pt] {\scriptsize 
		$1$} (rwyvy);
		\draw (hwy) edge node[pos=0.9, fill=white, inner sep=2pt] {\scriptsize 
		$4$} (rvxwy);
		\draw (hwy) edge node[pos=0.7, fill=white, inner sep=2pt] {\scriptsize 
		$4$} (rvywy);
		
		\draw (hvx) edge node[pos=0.76, fill=white, inner sep=2pt] {\scriptsize 
		$1$} (rvx);
		\draw (hvy) edge node[pos=0.76, fill=white, inner sep=2pt] {\scriptsize 
		$1$} (rvy);
		\draw (hwy) edge node[pos=0.76, fill=white, inner sep=2pt] {\scriptsize 
		$1$} (rwy);
		
		\draw (hv) edge node[pos=0.6, fill=white, inner sep=2pt] {\scriptsize 
			$2$} (rvxvy);
		\draw (hx) edge node[pos=0.76, fill=white, inner sep=2pt] {\scriptsize 
		$3$} (rvxvy);
		\draw (hv) edge node[pos=0.83, fill=white, inner sep=2pt] {\scriptsize 
		$2$} (rvxwy);
		\draw (hx) edge node[pos=0.76, fill=white, inner sep=2pt] {\scriptsize 
		$3$} (rvxwy);
		
		\draw (hv) edge node[pos=0.76, fill=white, inner sep=2pt] {\scriptsize 
		$2$} (rvyvx);
		\draw (hy) edge node[pos=0.76, fill=white, inner sep=2pt] {\scriptsize 
		$3$} (rvyvx);
		\draw (hv) edge node[pos=0.8, fill=white, inner sep=2pt] {\scriptsize 
		$2$} (rvywy);
		\draw (hy) edge node[pos=0.76, fill=white, inner sep=2pt] {\scriptsize 
		$3$} (rvywy);
		
		\draw (hw) edge node[pos=0.76, fill=white, inner sep=2pt] {\scriptsize 
		$2$} (rwyvx);
		\draw (hy) edge node[pos=0.76, fill=white, inner sep=2pt] {\scriptsize 
		$3$} (rwyvx);
		\draw (hw) edge node[pos=0.76, fill=white, inner sep=2pt] {\scriptsize 
		$2$} (rwyvy);
		\draw (hy) edge node[pos=0.76, fill=white, inner sep=2pt] {\scriptsize 
		$3$} (rwyvy);
		
		\node[squared-vertex, 
		label=270:$h^*$,label={[yshift=.cm,xshift=0.6cm]270:{$[3, 
				\infty]$}}] (hs) at ($(hw) - (0, 2)$) {};
		\node[squared-vertex, label=0:$h_1$] (h1) at ($(hs) - (0, 1)$) {};
		\node[squared-vertex, label=0:$h_2$] (h2) at ($(h1) - (0, 1)$) {};
		\node[squared-vertex, label=0:$h_3$] (h3) at ($(h2) - (0, 1)$) {};
		\draw (hs) edge node[pos=0.8, fill=white, inner sep=2pt] {\scriptsize 
		$2$} (rvx);
		\draw (hs) edge node[pos=0.76, fill=white, inner sep=2pt] {\scriptsize 
		$2$} (rvy);
		\draw (hs) edge node[pos=0.76, fill=white, inner sep=2pt] {\scriptsize 
		$2$} (rwy);
		
		\node[vertex, label=180:$r^*$] (rs) at ($(hs) - (1, 0)$) {};
		\node[vertex, label=180:$r_1$] (r1) at ($(rs) - (0, 1)$) {};
		\node[vertex, label=180:$r_2$] (r2) at ($(r1) - (0, 1)$) {};
		\node[vertex, label=180:$r_3$] (r3) at ($(r2) - (0, 1)$) {};
		
		\draw (hs) edge node[pos=0.76, fill=white, inner sep=2pt] {\scriptsize 
		$1$} (rs);
		\draw (h1) edge node[pos=0.76, fill=white, inner sep=2pt] {\scriptsize 
		$2$} (rs);
		
		\draw (h1) edge node[pos=0.76, fill=white, inner sep=2pt] {\scriptsize 
		$1$} (r1);
		\draw (h2) edge node[pos=0.76, fill=white, inner sep=2pt] {\scriptsize 
		$2$} (r1);
		
		\draw (h2) edge node[pos=0.76, fill=white, inner sep=2pt] {\scriptsize 
		$1$} (r2);
		\draw (h3) edge node[pos=0.8, fill=white, inner sep=2pt] {\scriptsize 
			$2$} (r2);
		
		\draw (h3) edge node[pos=0.76, fill=white, inner sep=2pt] {\scriptsize 
		$1$} (r3);
		\draw (h1) edge node[pos=0.86, fill=white, inner sep=2pt] {\scriptsize 
			$2$} (r3);
		\end{tikzpicture}
		
	\end{center}
	\caption{An example of the reduction in the proof of \Cref{thm:ha-m-closed} 
		with input graph $(\{v, w, x, y\}, \{e_1 = \{v, x\}, e_2= \{v, y\}, e_3 
		= 
		\{w, y\}\})$ and color distribution $(V^c=\{v,w\},V^d=\{x,y\})$.
		Residents depicted by a star represent $4^3\cdot 1 = 64$ residents with 
		these 
		preferences.}\label{fig:closed-hospitals}
\end{figure}
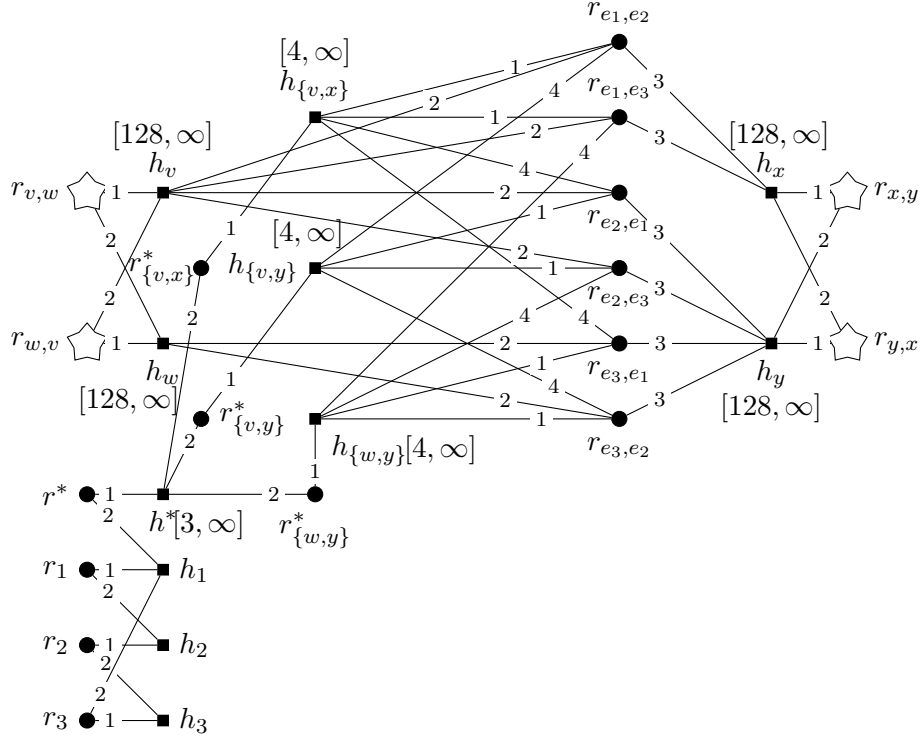

		{\bfseries ($\Rightarrow$)} Let $V'\subseteq V$ be a clique in $G$ with 
		$|V'\cap 
		V^c|=1$ for all $c\in [k]$.
		We denote by $v^c_*$ the vertex from~$V' \cap V^c$.
		Let $E'$ denote the set of all edges
		lying in the clique, i.e., all edges where both endpoints belong to 
		$V'$. Note that for each two different colors $c< d\in [k]$ it holds 
		that
		$|E'\cap 
		E^{c,d}|=1$. We construct a stable matching
		$M$ as follows. For each $c\in [k]$ and each~$v\in
		V^c\setminus V'$, we match to $h_v$ all residents with $h_v$ as top 
		choice, i.e., $\bigcup_{v'\neq v\in
		V^c\wedge i\in [n^3]} \{r_{v,v'}^i\}$, and  all residents which have 
		$h_{v_*^c}$ as their top-choice and $h_v$ as their second choice, i.e.,
		$\bigcup_{ i\in [n^3]} \{r_{v_*^c,v}^i\}$. Moreover, for each $c<d\in 
		[k]$ and
			each edge $e\in E^{c, d}\setminus E'$, we 
			match to $h_e$ all edge
			residents~$r_{e,e'}$ with $h_e$ as top choice, i.e.,
			$\bigcup_{e'\neq e\in E^{c, d}}
			\{r_{e,e'}\}$, the penalizing resident $r^*_e$, and the 
			resident~$r_{e^*, e}$ where~$\{e^*\}=E^{c, d}\cap E'$.
	Finally, we match $r^*$ and $r_1$ to hospital $h_1$ and $r_2$ and $r_3$ to 
	hospital $h_3$.
	Since there are exactly $n^3 |V^c|$ residents matched to each open vertex hospital $h_v$ with $v\in V^c$, exactly $|E^{c, d}| + 1$ residents matched to each open edge hospital $h_e$ with $e\in E^{c, d}$, and two residents matched to $h_1$ and $h_3$, it follows that $M$ is feasible.
	We now argue that $M$ is stable.

	Since all residents are matched to their most-preferred open hospital, it follows that there is no blocking pair.
	Next, we show that there exists no blocking coalition to open a closed 
	vertex hospital~$h_v$ for a vertex $v\in V^c$ for some $c\in [k]$.
	Since $h_v$ is closed, it follows that $v\in V'$.
	The only vertex
	residents which prefer to be matched to $h_v$ to their hospital in $M$ are 
	the $n^3 (|V^c| - 1)$ vertex residents of the form $r^i_{v, v'}$ for $i\in 
	[n^3]$ and $v' \in V^c\setminus \{v\}$.
	Moreover, there exist at most
	$(n-1)n^2$ edge residents that find $h_v$ acceptable (number of edges 
	incident to $v$ times two). As $(n-1)n^2 <
	n^3$, there
	cannot exist a blocking coalition of $n^3|V^c|$ residents to open~$h_v$.
	We now turn to a closed edge hospital~$h_{e^*}$.
	Since $h_{e^*}$ is closed, it follows that both endpoints of~${e^*}$ are contained in $V'$.
	Let ${e^*}\in E^{c, d}$ for some $c< d\in [k]$.
	The residents which prefer~$h_{e^*}$ to their hospital in~$M$ are the
	edge residents~$r_{e^*, e}$ for $e\in E^{c, d}\setminus\{e^*\}$ and $r^*_{e^*}$. However,
	these are only $|E^{c, d}| < l (h_{e^*})$ residents.
	Thus, there
	does not exist a blocking coalition to open $h_{e^*}$.
	Finally, 
	there exist $\binom{k}{2}$ penalizing residents~$r_e^*$ that are
	unmatched. However, 
	as $h^*$ has lower quota~$\binom{k}{2}+2$ and is only preferred by those $\binom{k}{2}$
	unmatched 
	penalizing residents and $r^*$, there does not exist a blocking coalition 
	to open~$h^*$.
	 
	 {\bfseries ($\Leftarrow$)} Assume that there exists a stable matching $M$ 
	 in the 
	 constructed \HRLQ instance. Note first of all that every stable matching 
	 matches
	 $r^*$ to $h_1$, as otherwise two residents from~$r_1$, $r_2$, and~$r_3$ 
	 form a blocking coalition to open one of $h_1$, $h_2$, or $h_3$ (see 
	 \Cref{ob:counter}). This implies that~$h^*$ needs to be closed
	 which in turn implies that at most $\binom{k}{2}$ edge hospitals can be 
closed (because otherwise $r^*$ together with the at least $\binom{k}{2}$ 
residents $r_e^*$ such that $h_e$ is closed form a blocking coalition to open 
$h^*$). For each of the $\binom{k}{2}$ color combinations $c< d\in [k]$,
     there are $|E^{c, d}|^2$ residents that find at least one edge 
     hospital~$h_e$ with $e\in E^{c,d}$ acceptable.
	 Thus, the
	 residents only suffice to open $|E^{c, d}|-1$ of the edge hospitals 
	 corresponding to edges from $E^{c,d}$. Note further that it is 
	 only possible to open $|E^{c, d}|-1$ edge hospitals if for all $e\in E^{c, d}$ edge resident $r_e$ is matched to an edge hospitals $h_{e'}$ for some $e' \in E^{c, d}$.
	 This implies that
	 all edge residents need to be matched to edge hospitals.
	 Note further that for each color~$c\in [k]$, at least $|V^c|-1$ vertex 
	  hospitals corresponding to vertices from $V^c$ are open (if there are two 
	  closed vertex hospitals~$h_v$ and $h_{v'}$ for $v\neq v'\in V^c$, then 
	  $\{r_{v, w}^i : i\in [n^3], w \in V^c \setminus \{v\}\} \cup 
	  \{r_{v', v}^i : i \in [n^3]\}$ form a blocking coalition to open 
	  $h_{v}$. Consequently, at most one vertex
	  hospitals from each color is closed in every stable matching.
	  If less than $k$ vertex hospitals are closed or the $k$ closed vertex 
	  hospitals do not form a clique, then there 
	  exists an edge~$e = \{v, w\}\in E^{c, d}$ for some~$c< d\in [k]$ such
	  that $h_e$ is closed and $h_v$ or $h_w$ is open in~$M$ (as for each 
	  $c<d\in [k]$ there exists a hospital $h_e$ with $e\in E^{c,d}$ that is 
	  closed). We assume without loss of generality that $h_v$ is open in~$M$.
	  As $h_e$ is
	  closed and all edge residents are matched to edge hospitals, for every edge $e' \in E^{c, d}\setminus \{e\}$, resident $r_{e,e'}$ together with $h_v$
	  forms a blocking pair, as
	  $r_{e, e'}$ prefers being matched to $h_v$ to being matched to~$M(r_{e,e'})=h_{e'}$.
	  Thus, the $k$ vertices corresponding to the closed
	  vertex hospitals form a clique.
	\end{proof}
	While we have seen in this section that for both \HRLQ and \HRLUQ deciding 
	whether a stable matching with $m_{\open}$ open ($m_{\text{closed}}$ 
	closed) 
	hospitals exists is W[1]-hard parameterized by $m_{\open}$ 
	($m_{\text{closed}}$), 
	the two problems still lie in XP. For both parameters, we can iterate over 
	all possible subsets of open hospitals $H_{\open}\subseteq H$ of allowed 
	size and 
	subsequently employ the algorithms from \Cref{ob:HRLQ-H'} and 
	\Cref{pr:HRULQ-H'}.

	\section{A Restricted Case: Lower Quota Two} \label{sec:q2}

	In this section, we consider the special case of our models where all 
	hospitals have lower quota at most two. For \HALUQ, we have already seen in 
	\Cref{pr:HRLUQI-NP-const} that the problem remains NP-hard even in this 
	case. However, for the other two models, all our hardness proofs
	used hospitals with lower quota three.
	
	In the following, let \HRLUQtwo denote 	
	the restriction of \HRLUQ to instances in which all hospitals
	have 
	lower quota one or two.
	We present a 
	polynomial-time algorithm for \HRLUQtwo, which computes a stable matching
	if one
	exists. As \HRLQ is a special
	case 
	of \HRLUQ, this algorithm also applies to \HRLQ instances where all
	lower quotas are at most two. The main result of this section is the 
following:
    \begin{restatable}[]{theorem}{qq}
		\label{t:q2}
		If the lower quota of each hospital is at most two, then \HRLUQ (and
		thereby also
		\HRLQ)
		is solvable in $\mathcal{O}(n^3 m)$ time.
	\end{restatable} We split the proof of this result in two sections. In the 
first section, we give a description of our algorithm. In the second section, 
we prove its correctness.
	
	At the end of the section, we also derive a weakened version of the Rural 
	Hospitals 
	Theorem~\cite{RePEc:ucp:jpolec:v:92:y:1984:i:6:p:991-1016,10.2307/1913160} 
	for \HRLUQtwo, showing that every stable matching in a \HRLUQtwo instance  
	matches the same set of 
	resident and opens the same number of hospitals.
	
	In the following, we refer to all 
	hospitals with lower quota one as \emph{quota-one hospitals} and to all 
	hospitals 
	with lower quota two as \emph{quota-two hospitals}.
	
	\subsection{Warm-up: Irving's Algorithm for Stable Roommates}
	\label{sec:irving}
	
	As a warm-up, we first consider the case in which each hospital with lower 
	quota two finds at most two residents acceptable and each quota-one 
	hospital also has upper quota one.
	Such an instance of \HRLUQtwo can be reduced to a \textsc{Stable Roommates} 
	instance.
	The \textsc{Stable Roommates} problem takes as input a set of agents, where 
	each agent has preferences over an acceptable subset of the other agents, 
	and the goal is to find a stable matching where each agent is matched to at 
	most one other agent and stability is defined as the absence of a pair of 
	agents preferring each other over their partner in the matching.
	The reduction works as follows.
	We insert an agent~$a_r$ for each resident $r\in R$
	 and an agent $a_h$ for each quota-one hospital $h\in H$. For resident~$r\in R$, we construct the preferences of $a_r$ from the preferences of
	 $r$ by replacing each quota-one hospital $h\in H$ by $a_h$ and each 
	 quota-two hospital $h\in H$ which accepts $r$ and one other resident~$r'$
	 by the agent $a_{r'}$ corresponding to resident $r'$. Moreover, for each 
	 quota-one hospital~$h\in H$, we construct the preferences of $a_h$ from
	 the preferences of $h$ by replacing each resident~$r\in R$ by the
	 corresponding agent $a_r$ (note that parts of this reduction are 
	 ``inverse'' to 
	 the one sketched at the beginning of \Cref{sec:NP-hard}).
	 It is easy to see that any stable matching for the \HRLUQtwo instance 
	 induces a stable matching for the \textsc{Stable Roommates} instance and 
	 vice versa.
	Thus, \HRLUQtwo is a generalization of \textsc{Stable Roommates}, and our algorithm to solve \HRLUQtwo will be a generalization of Irving's algorithm~\cite{DBLP:journals/jal/Irving85} to solve \textsc{Stable Roommates}.
	In the following, we give a brief description of Irving's algorithm as a warm-up.
	The algorithm consists of two phases:
	
	In the first phase of Irving's algorithm, every agent proposes to the first 
	agent on its preference list.
	Whenever an agent receives multiple proposals, it rejects all proposals but the best.
	If agent $a$ rejects the proposal of $a'$, then the acceptability of $a$ 
	and $a'$ is deleted, and $a'$ proposes to its most preferred agent that is 
	still contained in its preferences unless its preferences got empty.
	This phase ends when every agent either holds and receives exactly one proposal or has empty preferences.

    In the second phase of Irving's algorithm, the algorithm searches for a 
    certain substructure called \emph{rotation}.
    A rotation is a sequence of pairs of agents $(a_1, b_1), \dots, (a_k , 
    b_k)$ such that $a_i$ is the last agent on $b_i$'s preferences, and 
    $b_{i+1}$ is the second agent on $a_i$'s preferences (where all indices are 
    taken modulo~$k$).
    Note that, by the definition of Phase 1, one can show that $b_i$ is the 
    first agent in $a_i$'s preferences.
    Whenever such a rotation is found, it is \emph{eliminated}, meaning that $b_i$ deletes the mutual acceptability to all agents which are after~$a_{i-1}$ in its preferences (note that this always includes $a_i$).
    The algorithm successively eliminates rotations.
    If the preferences of an agent become empty through the elimination of a rotation, then the algorithm concludes that there does not exist a stable matching.
    Otherwise, at some point the preferences of each agent contain at most one 
    agent inducing a stable matching.

    Note that the elimination of a rotation can also be interpreted as deleting 
    the mutual acceptabilities of $a_i$ and $b_i$ for each $i\in [k]$ and 
    restarting Phase~1 afterwards.
    This interpretation is closer to our algorithm for \HRLUQtwo.
    
    Our algorithm generalizes Irving's algorithm as follows.
    In the first phase, residents propose and reject proposals as in Irving's algorithm.
    However, a hospital~$h$ proposes to the first $\uq (h)$ hospitals on its 
    preferences, and only rejects a proposal if it already received $\uq(h)$ 
    proposals.
    Furthermore, quota-two hospitals only start proposing after they received a proposal.

    In the second phase, we search for some substructures (which are 
    generalizations of rotations).
    However, we have to take into account that quota-two hospitals are not 
    necessarily open if they received only one proposal.
    Therefore, we avoid that possibly closed quota-two hospitals are
    contained in 
    a rotation by replacing quota-two 
    hospitals occurring in a rotation by the first or second resident on their 
    preferences.

    Between the first and second phase (we call this later Phase~1b), we replace each quota-two hospital which has to be open (i.e., which received at least two proposals) by several quota-one hospitals to simplify the presentation of Phase~2.

    We now describe our algorithm for \HRLUQtwo\ in detail.

	\begin{algorithm}[t]
		\caption{Algorithm for \HRLUQtwo (high-level 
		description)}\label{alg:euclid}
		\begin{algorithmic}[1]
			\Input{An \HRLUQtwo instance $\mathcal{I}$}
			\Output{A stable matching in $\mathcal{I}$ or NO if $\mathcal{I}$ 
			does not 
				admit a stable matching.}
			\State Apply \textbf{Phase 1a - Propose\&Reject}
			
			\State $S \gets \{\text{residents with non-empty 
				preferences}\}$\Comment{Initialization}
			
			\While{there exists some resident $r$ with at least two hospitals 
			on her 
				preferences}
			\State Apply \textbf{Phase 1a - Propose\&Reject}
			\While{a hospital holding at least two proposals 
				exists}
			\For{\textbf{each} hospital $h$ holding at least two proposals}
			\State Split $h$ into 
			$\uq(h)$ hospitals $h^1$, \dots, $h^{\uq(h)}$ \Comment{Phase 1b}
			\EndFor
			\State Apply \textbf{Phase 1a - Propose\&Reject}
			\EndWhile
			
			\If{there exists some resident $r$ with at least two hospitals on 
			her 
				preferences}
			\State Find a generalized rotation $R$. \Comment{Phase 2}
			\State Eliminate $R$.
			\EndIf
			\EndWhile
			
			\If{all residents from $S$ have exactly one hospital on their 
			preferences 
				left}
			\State \textbf{return} matching $M$ that matches every resident $r$ 
			with 
			non-empty preferences to the 
			hospital from $\mathcal{I}$ corresponding to the remaining hospital 
			on her preferences.
			\EndIf
			\State \textbf{return} NO 
		\end{algorithmic}
		\label{alg} 
	\end{algorithm}
    \subsection{Description of the Algorithm for 
\HRLUQtwo
    }
	
    Similar to Irving's algorithm, our algorithm is split into two phases, 
    where 
    the first 
	phase is again split into Phase 1a and Phase 1b.
	Phase 1a identifies hospital-resident pairs which cannot be part of a 
	stable matching using a propose-and-reject approach.
	Subsequently, for each such hospital-resident pair $(r,h)$,  hospital $h$ is deleted
	from the preferences of $r$ and vice versa.
	Furthermore, Phase 1a identifies some quota-two hospitals which 
	are open in every stable matching.
	Phase 1b further simplifies the instance by replacing quota-two hospitals 
	that 
	are open in every stable matching by multiple copies of this hospital with lower quota one.
	Phase~1a and Phase 1b are applied repeatedly until no hospital from which 
	we know that it is open in every stable matching exists. After that, in 
	Phase 2, we identify 
	substructures which we call ``generalized 
	rotations'' and subsequently eliminate them by deleting the acceptability 
	of some hospital-resident pairs. 
    \Cref{alg} gives a high-level description of our algorithm.
	While Phase~1 keeps the number of stable matchings identical, 
Phase 2 may 
	reduce the number of stable matchings in the instance, but still guarantees 
	that at least one stable matching survives (if there exists one in the
	original instance).

	The algorithm applies Phase 1 and Phase 2 alternately
	until every resident has at most one hospital on her preferences.
	The algorithm returns NO if, after the initialization, the 
preferences of a resident got empty, as one can show that
all residents having non-empty preferences after the initialization are matched 
in every 
stable 
	matching. Otherwise, the algorithm
	constructs a stable matching where all
	residents with empty preferences are unmatched and all residents with
	non-empty preferences are matched to the hospital on their preference
	list (if this hospital was created by splitting a hospital $h$, then the resident is matched to $h$).

	We now describe Phases 1a, 1b, and 2 in more detail.
	We start with an
	observation that allows us to assume that quota-one hospitals have also upper
	quota one: 
	\begin{observation}
		Given an \HRLUQ instance, replacing all quota-one
		hospitals~$h\in H$ by $u(h)$ copies
		$h^1,\dots, h^{u(h)}$ each with lower and upper quota one and with 
		the same preferences as $h$ and replacing $h$ by $h^1\succ \dots \succ 
		h^{u(h)}$ in the preferences of all residents results in an equivalent 
		instance.
	\end{observation} 	
	
	\subsubsection{Phase 1a - Propose\&Reject}
	In Phase 1a, residents and hospitals propose
	to one another. Residents always propose to hospitals and hospitals always 
	to residents. If a resident $r\in R$ proposes to a hospital $h\in H$, then $h$ can either \emph{accept} or \emph{reject} the proposal from~$r$. We
	say that a hospital~$h$ \emph{holds} a proposal $r$ if~$r$ proposed to $h$
	and $h$ 
	did not reject the proposal (until now). We say that a resident~$r$
	(currently)
	\emph{issues} a proposal if there exists a hospital~$h$ that holds 
	the 
	proposal 
	$r$. The notation also applies if the roles of residents and hospitals are 
	swapped. Considering quota-two hospitals, we distinguish between
	\emph{activated} and \emph{deactivated} hospitals.
	Activated hospitals propose to residents, while deactivated hospitals do not.
	Initially, all quota-two
	hospitals
	are deactivated. 
	
	\paragraph{Algorithm (Phase 1a).} We proceed in multiple rounds. In 
	each
	round, a resident or quota-one hospital with non-empty preferences that
	does not 
	currently issue a proposal or an activated quota-two hospital is selected.
  If a resident or quota-one hospital is selected, then it
	proposes to the first hospital or resident on its preference list. If an 
activated quota-two hospital~$h$ is selected, then the
	hospital
	proposes to the first $u(h)$ residents on its preference list
	unless $h$ received exactly one proposal and this one proposal comes from a
	resident~$r$ which is among the first $u(h)$ residents in $h$'s
	preferences: In this case,~$h$ only proposes to the first $u(h)-1$ 
	residents that are not~$r$.
	
	If a resident or a quota-one hospital receives a proposal, then it
	\emph{accepts} the proposal if it does not hold a proposal or if it prefers the new
	proposal to the one it currently holds. Similarly, it \emph{rejects} a
	proposal if it
	either already holds or later receives a better proposal. A quota-two
	hospital $h$ \emph{accepts} a proposal $r$ if it does not hold 
	$u(h)$ proposals it prefers to $r$. It \emph{rejects} a proposal $r$ if it  
	holds or at some point receives $u(h)$ proposals it prefers to $r$, or if
	the hospital has been rejected by all but one resident on its preference 
	list. If 
	an agent $a$ proposes to an agent $a'$ and $a'$ rejects the proposal, then we
	delete $a'$ from the preference list of $a$ and $a$ from the 
	preference list of $a'$. 
	
	A quota-two hospital $h$ gets activated if it receives a proposal or if one 
of 
	its 
	proposals gets rejected. If $h$ currently 
	holds exactly one proposal by one of its~$u(h)$ most preferred residents 
	$r$, then it gets deactivated if it currently issues 
	$u(h)-1$ proposals
	or has proposed to all residents on its preference list except $r$. 
	Otherwise, it gets deactivated if it currently issues $u(h)$ proposals or 
	has proposed to all residents on its preference list.
	
	At the end of Phase 1a, we delete from the 
	preferences of all 
	quota-one hospitals and residents holding 
	a proposal all agents to which they 
	prefer the held proposal. Subsequently, we restore the mutual acceptability 
	of agents
	by deleting for each agent $a$ an agent $a'$ from its preference list
	if $a$ does not appear on the preference list of $a'$. Finally, we delete 
	all
	quota-two hospitals with at most one resident on their preference list from 
	the instance and the preferences of all~agents. 

	\medskip

	The intuitive
	reasoning behind Phase 1a is the following. If an agent 
	rejects the proposal of another agent, then the two can never be matched to
	each other in a stable matching. Thereby, no agent can be matched better 
	than the agent it proposes to.
	Thus, any agent receiving a proposal can be sure that it does not end up 
	worse than the proposal it currently holds in a stable matching, since it 
	forms a blocking pair with the agent issuing its proposal otherwise. 
	After Phase 1a, each resident and quota-one hospital issues a 
	proposal to the first agent on its preference list and holds a proposal 
	from the last agent on its preference list. 
	The formal correctness proof of this phase presented in \Cref{a:p1a}
	consists of proving that, in a stable matching, no agent can be matched to 
an agent which was 
	deleted from its preferences during Phase 1a and that 
	all matchings that are stable after applying Phase 1a are also stable 
	before applying Phase 1a. 
	
    \begin{figure}[t]
	\begin{minipage}{0.49\textwidth}
	\begin{tikzpicture}
\node (I) at (0, 0) {\textbf{I}};
      \node (xshift) at (2.5,0) {};
      \node (yshift) at (0, -0.5) {};
      \node (anchor1) at (0.5,-0.) {};
      \node (anchor2) at ($(anchor1) + (yshift)$) {};
      \node (anchor3) at ($(anchor2) + (yshift)$) {};
	  \node[anchor=west] (h1) at (anchor1) {$h_1 : r_3 \succ r_1$};
	  \node[anchor=west] (r1) at ($(anchor1) + (xshift)$) {$r_1 : h_1 \succ h_2$};
	  \node[anchor=west] (h2) at ($(anchor1) + (yshift)$) {$h_2 : r_1 \succ r_2$};
	  \node[anchor=west] (r2) at ($(anchor2) + (xshift)$) {$r_2 : h_4 \succ h_2 \succ h_3$};
	  \node[anchor=west] (h3) at ($(anchor2) + (yshift)$) {$h_3 : r_2 \succ r_3$};
	  \node[anchor=west] (r3) at ($(anchor3) + (xshift)$) {$r_3 : h_3 \succ h_1 \succ h_4$};
	  \node[anchor=west] (h4) at ($(anchor3) + (yshift)$) {$h_4  : r_2 \succ r_3$};
	 \end{tikzpicture}
	\end{minipage}
	\hfill
	\begin{minipage}{0.49\textwidth}
	\begin{tikzpicture}
\node (I) at (0, 0) {\textbf{II}};
      \node (xshift) at (2.5,0) {};
      \node (yshift) at (0, -0.5) {};
      \node (anchor1) at (0.5,-0.) {};
      \node (anchor2) at ($(anchor1) + (yshift)$) {};
      \node (anchor3) at ($(anchor2) + (yshift)$) {};
	  \node[anchor=west] (h1) at (anchor1) {$h_1 : r_3 \succ r_1$};
	  \node[anchor=west] (r1) at ($(anchor1) + (xshift)$) {$r_1 : h_1 \succ h_2$};
	  \node[anchor=west] (h2) at ($(anchor1) + (yshift)$) {$h_2 : r_1 \succ r_2$};
	  \node[anchor=west] (r2) at ($(anchor2) + (xshift)$) {$r_2 : h_2 \succ h_3$};
	  \node[anchor=west] (h3) at ($(anchor2) + (yshift)$) {$h_3 : r_2 \succ r_3 $};
	  \node[anchor=west] (r3) at ($(anchor3) + (xshift)$) {$r_3 : h_3 \succ h_1$};
	  \node[anchor=west] (h4) at ($(anchor3) + (yshift)$) {$h_4 :$};
\end{tikzpicture}
	\end{minipage}

	 \caption{An example for Phase 1.
	 Hospital $h_1 $ is a quota-one hospital, while the other three 
	 hospitals are quota-two hospitals with upper quota two.
	 In the beginning (see instance I), each resident and $h_1$ propose to the 
	 first agent on their preferences. All agents accept their received 
	 proposal.
	 Since $h_4$ receives a proposal from $r_2$ and $h_3$ receives a proposal 
	 from $r_3$, they get activated and 
	 propose to $r_3$ respectively $r_2$.
	 Resident $r_3$ rejects the proposal from $h_4$, as she holds the proposal 
	 of $h_1$, while $r_2$ accepts the proposal of $h_3$.
	 Then, the preferences of $h_4$ contain only one resident, and thus $h_4 $ 
	 rejects the proposal from $r_2$.
	 Consequently, $r_2$ proposes to $h_2$, which activates $h_2$.
	 Subsequently, $h_2$ proposes to $r_1$, who accepts the proposal. As no 
	 quota-two
	 hospital received two proposals, no hospital gets split and Phase 1 ends.
	 The resulting instance is depicted as instance II.
	 }
	 \label{fig:example-p1-small}
	\end{figure}

	\subsubsection{Phase 1b - Split Hospitals} In this phase, we identify 
	quota-two hospitals that are open in every stable matching 
	and replace them by quota-one hospitals: 
	
	\paragraph{Algorithm (Phase 1b).} We replace each quota-two hospital $h$ 
	holding at least two proposals by~$u(h)$ hospitals $\hone$,
	\dots, $h^{\uq (h)}$ with lower and upper quota one 
	with the same preferences
	as~$h$. In the preferences of all residents, $h$ is replaced by $h^1 \succ 
	\dots \succ h^{\uq (h)}$.
	
	\medskip 
	Phase~1b allows us to replace quota-two hospitals with at least two 
	proposals by quota-one hospitals.
	This makes Phase~2 simpler, as we do not have to distinguish between quota-two hospitals receiving at most one proposal and quota-two hospitals receiving at least two proposals.
	
	We show in \Cref{a:p1b} that there exists
    a one-to-one correspondence between stable matchings before and after the
    application of Phase 1b.
    To summarize,  Phase 
	1 consists of 
	applying Phase 1a and Phase 1b as long as at least one hospital was split 
	in the last execution of Phase 1b.
	An example for the execution of Phase 1 can be found in 
	\Cref{fig:example-p1-small}.
	
	\subsubsection{Phase 2 - Eliminate Generalized Rotations}
	We introduce some notation for the definition of a generalized rotation. We
	call a quota-two hospital with more than two residents on its preferences 
	\emph{\fl}
	and all other 
	quota-two hospitals 
	\emph{\infl}.
	Note that while we already know which
	residents will be assigned to an open \infl hospital (as the number of 
	residents on its preferences is equal to its lower quota), this is not 
clear for 
	open \fl hospitals.
	Given a resident $r$, we denote by~$h(r) $ the first hospital on $r$'s
	preferences.
	If $h(r)$ is flexible, then we define $g(r):= h(r)$. 
	Otherwise, we define $g(r)$ to be the second hospital on $r$'s preference 
	list.
	
  A \emph{generalized rotation} is a sequence
	$(a_1, b_1), \dots, 
	(a_k, b_k)$ consisting of residents and quota-one hospitals with $a_i\neq 
	a_j$ for all $i\neq j$ 
	such that (all following indices are taken modulo $k$):
	
	\medskip
	\noindent \textbf{Relationship between $a_i$ and $b_{i+1}$}: 
	\begin{description}
		  \setlength\itemsep{-0.3em}
		\item[\ab{1}] If $a_i$ is a quota-one hospital, then $b_{i+ 1}$ is 
		the 
		second resident on $a_i$'s preferences. 
		\item[\ab{2}] If $a_i$ is a resident and $h(a_i)$ is a \fl hospital, 
		then 
		$b_{i+1}$ is the second-most preferred
		\item[] \hspace*{1.4cm} resident on $h (a_i)$'s preferences who is not $a_i$. \smallskip
		\item[]\hspace*{-0.3cm} If $a_i$ is a resident and $h(a_i)$ is an
\infl 
		hospital or a quota-one hospital and \smallskip
			\item[\ab{3a}] \hspace*{0.5cm} if
			$g(a_i)$ is a quota-one hospital, then $b_{i+1}:=g(a_i)$.
			\item[\ab{3b(i)}] \hspace*{0.000cm} if $g(a_i)$ is a quota-two 
			hospital holding a 
			proposal $r$, then $b_{i+1} :=r$.
			\item[\ab{3b(ii)}] if $g(a_i)$ is a quota-two 
			hospital which does 
			not hold a proposal, then $b_{i+1}$ is~$g(a_i)$'s
			\item[] \hspace*{2.25cm}  most preferred
			resident who is not $a_i$.
	\end{description}

	\noindent \textbf{Relationship between $b_i$ and $a_i$}:
	\begin{description}
		\setlength\itemsep{-0.3em}
		\item[\ba{1}] If $b_i$ is a quota-one hospital, then $a_{i}$ is the 
		last resident on $b_i$'s preferences.
		\item[\ba{2}] If $b_i$ is a resident and the last hospital $h$ on 
		$b_i$'s 
		preferences is of quota one, then $a_{i} := h$.
		\item[\ba{3}] If $b_i$ is a resident and the last hospital $h$ on 
		$b_i$'s 
		preferences is of quota two, then $a_{i}$ is
		\item[]\hspace*{1.1cm} the resident with $h$ as top-choice,
		i.e., the resident proposing to 
		$h$.
	\end{description}
	
	Note that residents and quota-one hospitals are treated the same as ordinary agents in 
	Irving's algorithm, while \infl hospitals are treated as ``edges'' between 
	two residents.

	\begin{figure}[t]
	\begin{minipage}{0.49\textwidth}
	\begin{tikzpicture}
\node (I) at (0, 0) {\textbf{I}};
      \node (xshift) at (2.5,0) {};
      \node (yshift) at (0, -0.5) {};
      \node (anchor1) at (0.5,-0.) {};
      \node (anchor2) at ($(anchor1) + (yshift)$) {};
      \node (anchor3) at ($(anchor2) + (yshift)$) {};
      \node[anchor=west] (h1) at (anchor1) {$h_1  : r_3 \succ r_1$};
      \node[anchor=west] (r1) at ($(anchor1) + (xshift)$) {$r_1 : h_1 \succ h_2$};
	  \node[anchor=west] (h2) at ($(anchor1) + (yshift)$) {$h_2 : r_1 \succ r_2$};
	  \node[anchor=west] (r2) at ($(anchor2) + (xshift)$) {$r_2 : h_2 \succ h_3$};
	  \node[anchor=west] (h3) at ($(anchor2) + (yshift)$) {$h_3 : r_2 \succ r_3$};
	  \node[anchor=west] (r3) at ($(anchor3) + (xshift)$) {$r_3 : h_3 \succ h_1$};
	 \end{tikzpicture}
	\end{minipage}
	\hfill
	\begin{minipage}{0.49\textwidth}
	\begin{tikzpicture}
\node (I) at (0, 0) {\textbf{II}};
      \node (xshift) at (2.5,0) {};
      \node (yshift) at (0, -0.5) {};
      \node (anchor1) at (0.5,-0.) {};
      \node (anchor2) at ($(anchor1) + (yshift)$) {};
      \node (anchor3) at ($(anchor2) + (yshift)$) {};
	  \node[anchor=west] (h1) at (anchor1) {$h_1 : r_3$};
	  \node[anchor=west] (r1) at ($(anchor1) + (xshift)$) {$r_1 :  h_2$};
	  \node[anchor=west] (h2) at ($(anchor1) + (yshift)$) {$h_2 : r_1 \succ r_2$};
	  \node[anchor=west] (r2) at ($(anchor2) + (xshift)$) {$r_2 :  h_2$};
	  \node[anchor=west] (h3) at ($(anchor2) + (yshift)$) {$h_3 : r_3 $};
	  \node[anchor=west] (r3) at ($(anchor3) + (xshift)$) {$r_3 : h_3 \succ h_1$};
	 \end{tikzpicture}
	\end{minipage}

	 \caption{An example for Phase 2.
	 Hospital $h_1 $ is a quota-one hospital, while the other two hospitals are 
	 quota-two hospitals with upper quota two. Initially (see 
	 instance I), $h_1$ holds the proposal of $r_1$, $h_2$ the proposal
	 of $r_2$, and~$h_3$ the proposal of $r_3$.
	 The instance admits the following generalized rotation:
	 $(r_1, h_1), (r_3, r_2)$. Note that for $b_1=h_1$ case BA-1 applies, for
	 $a_1=r_1$ case AB$^+$-3b(i) applies, for $b_2=r_2$ case BA-3 applies, and 
	 for~$a_2=r_3$
	 case AB$^+$-3a applies.
	 Eliminating this generalized rotation results in the instance II.
	 Afterwards, Phase~1 is applied, and $h_3 $ rejects $r_3$ because $h_1$ has only $r_3$ in its preferences.
	 The algorithm ends with the stable matching $\{(h_1 , r_3), (h_2, \{r_1, r_2\})\}$.
	 }
	 \label{fig:example-p2-small}
	\end{figure}
	
	\paragraph{Algorithm (Phase 2).}
	If the preference list of every resident contains at most one agent, then these preference lists correspond to a stable matching.
	Otherwise, Phase 2 computes a generalized rotation by starting with an 
	arbitrary
	resident whose
	preference list has length at least two as $a_1$ and subsequently applying 
	the 
	relationships depicted above to find $b_2$, $a_2$, \dots~until this
	procedure cycles and a generalized rotation has been found.
	Subsequently, we \emph{eliminate} the found rotation by deleting, for
	all $i\in
	[k]$, the mutual 
	acceptability of $a_i$ and $b_i$ if one of them is a hospital, and 
	otherwise the mutual acceptability of hospital $h(a_i)$ and $b_i$.
	After that, Phase 1 is applied again to the resulting instance.
	\medskip
	
	An
	example for 
	Phase 2 can
	be found in \Cref{fig:example-p2-small}.
    The correctness proof of this phase presented in \Cref{a:p2} starts with 
    showing that if there exists a stable matching in
	the instance before the elimination of a generalized rotation, then there also exists a
	stable matching after its elimination (\Cref{lem:rotation}). For the other 
	direction, we show 
	that each matching that is stable after the elimination of a generalized 
	rotation, and matches all residents with non-empty preferences (for whom 
	it can be proven that they have to be matched in any stable matching), is
	also stable before its elimination (\Cref{lem:forward}). 

	In a ``classical'' rotation $(a_1, b_1), \dots,
	(a_k, b_k)$ for \textsc{Stable 
	Roommates}~\cite{DBLP:journals/jal/Irving85}, for all $i\in [k]$, 
	agent~$b_{i+1}$ is
	the second 
	agent on the preference list of~$a_{i}$ and~$a_i$ is the last agent on 
	the preference list of~$b_i$, which implies that $b_i$ is the top-choice of 
	$a_i$. Here, eliminating a rotation consists of deleting the mutual 
	acceptability of $a_i$ and~$b_i$ for all $i\in [k]$ and results in an 
	instance that admits a stable matching if the original instance admits one. 
	Part of the reason for this is that if we assume 
	that 
	there exists a stable matching $M$ which contains the pairs $(a_i, b_i)$ 
	for all 
	$i 
	\in [ k]$, then the matching~$M'$ arising by replacing these pairs by 
the pairs $(a_i,
	b_{i+1})$ is also stable:
	Each agent~$b_i$ prefers~$M'$ to~$M$, and agent $a_i$ can only form a 
	blocking pair with $b_i$, implying that no blocking pair has been 
	introduced.
	However, applying this classical definition to \HRLUQtwo, the 
	observation from above does not longer hold (for example, if a \infl 
	quota-two hospital would appear as $a_i$, then a feasible matching must 
	match it to both $b_i$ and $b_{i+1}$ or close it).
	Therefore, we generalize the definition of a rotation in a way such that no 
	quota-two 
	hospital appears in a generalized rotation, while keeping the intuition 
	that $a_i$ still corresponds to $b_i$'s ``least preferred option'' and 
	$b_{i+1}$ to $a_i$'s ``second-most preferred option'':
	For~\ba{1} and~\ba{2}, $a_i$ is still $b_i$'s least preferred agent, while 
	for~\ba{3} following the classical 
	definition, it would be necessary to set $a_i$ to a quota-two hospital~$h$. 
	Instead, we set $a_i$ to be the resident proposing to~$h$, which can be 
	interpreted as matching $b_i$ to~$h$ together with $a_i$. For the
	relationship between~$a_i$ and $b_{i+1}$, for \ab{1}, $b_{i+1}$ is still 
	the second agent in the preferences of $a_i$. For \ab{2}, it 
	is necessary to recall that for a 
	\fl hospital $h$ there exist multiple possibilities which residents
	are matched to~$h$ in a stable matching. The ``most preferred option'' of 
	a resident $r$ is to be 
	matched to $h(r)$ together with $h(r)$'s most preferred remaining resident, 
	while 
	her ``second-most preferred option'' is to be matched to $h(r)$ with 
	$h(r)$'s
	second-most preferred remaining resident. If~$h(r)$ is \infl, then 
	there
	exists only 
	one possibility for a resident $r$ to be matched to~$h(r)$ so her 
	second-most
	preferred alternative is to be matched to $g(r)$ with the necessary case 
	distinctions made in \ab{3a}, \ab{3b(i)}, and \ab{3b(ii)}, which 
	corresponds to treating~$h(r)$ like an edge between $r$ and the other 
	resident acceptable to $h(r)$.
	\subsection{Proof of Correctness of the Algorithm}
	In the following section, we prove the correctness of the presented 
	algorithm. We start by proving the correctness for each phase separately, meaning
	that there exists a stable
	matching in the instance after the execution of a phase if and only if
	there exists a stable matching before the execution of the phase.
	The correctness of both phases together then proves the correctness of
	the full algorithm.
		\subsubsection{Phase 1a}\label{a:p1a}
	We start proving the correctness of Phase 1a by showing that if an agent deletes another agent from its
	preference list, then they cannot be matched in a stable matching. To do 
	so, we iterate over all possible cases in Phase 1a where the 
	preferences of an agent get modified.

	In the following, we slightly abuse notation by calling a blocking 
	coalition $(h, \{r\})$ to open a quota-one hospital $h$ also a blocking 
	pair.
	\begin{restatable}{lemma}{firstphase}
		\label{lem:first-phase}
			If an agent $a^*$ rejects the proposal of another agent $a'$, then 
		$a^*$ cannot be matched to $a'$ in any stable matching.
	\end{restatable}
	
	\begin{proof}
		Assume that the lemma does not hold, and let $r$ be the first resident 
		and $h$ the first hospital such that one of them rejected the other, 
		but $M(r) = h$ in some a stable matching~$M$. There are 
		two possible reasons for this deletion: Either $r$ rejected a proposal 
		from $h$ or $h$ rejected a proposal from $r$. 
		
		{\bfseries Case 1: } $r$ rejected $h$.
		
		Then $r$ received a proposal from a hospital $h'$ it prefers to $h$.
		
		{\bfseries Case 1.1: } $h'$ is a quota-one hospital.
		
		Then $h'$ has been rejected from all residents it prefers to $r$.
		By the choice of $r$ and~$h$, it follows that $h'$ cannot be matched to 
		a resident it prefers to $r$ in $M$.
		Thus, $(r, h')$ is a blocking pair for $M$, a contradiction.
		
		{\bfseries Case 1.2: } $h'$ is a quota-two hospital.
		
		Then $h'$ has received a proposal from some resident $r'\neq r$ before 
		proposing to $r$.
		By the choice of~$h$ and~$r$, it follows that $r'$ cannot be matched 
		to a hospital it prefers to $h'$.
		It follows that if $h'$ is closed in $M$, then $(h',\{r,r'\})$ is a
		blocking coalition in $M$, a contradiction.
		Furthermore,  as $r$ received a proposal from $h'$, $r$ is one of the 
		first $\uq (h')$ residents on the 
		preferences of $h'$ which did not reject $h'$.
		By the choice of~$h$ and $r$, it follows that $h'$ is not matched to 
		$\uq (h')$ resident it prefers over $r$.
		Thus, if~$h'$ is open, then $(h',r)$ is a blocking pair in $M$, a 
		contradiction.
		
		{\bfseries Case 2: } $h$ rejected $r$.
		
		{\bfseries Case 2.1: } $h$ is a quota-one hospital.
		
		Then before rejecting $r$, hospital $h$ received a proposal from a resident $r'$
		it prefers over $r$.
		By the choice of $h$ and $r$, it follows that $r'$ prefers $h$ to 
		$M(r')$.
		Consequently,~$(h,r')$ is blocking pair in $M$, a contradiction.
		
		{\bfseries Case 2.2: } $h$ is a quota-two hospital.
		
		The hospital $h$ cannot reject $r$ because it has been rejected by all 
		other residents on its preference list, as otherwise $r'$ with $r\neq 
		r'\in M(h)$ has been rejected by~$h$ before $h$ rejected~$r$, 
		contradicting the choice of $h$ and $r$.
		Thus before rejecting $r$, hospital $h$ received proposals from~$\uq (h)$
		residents $r_1$, \dots, 
		$r_{\uq (h)}$ which $h$ prefers over $r$.
		By the choice of $h$ and~$r$, it follows that $r_i$ cannot be matched 
		to a hospitals it prefers over $h$ for all~$i \in [\uq (h)]$. As $r$ is 
		matched to~$h$ in $M$, there needs to exists at least one $j\in [\uq (h)]$
with~$M(r_j)\neq h$. Then $(h,r_j)$ forms a blocking pair in $M$, a
		contradiction. 
	\end{proof}
	Recalling that agents propose to other agents in order of their 
preference list, the
	following observation directly follows from \Cref{lem:first-phase}.
	\begin{observation}\label{ob:ma}
		No resident or quota-one hospital $a^*$ which issues a proposal to an  
agent~$a'$ can be
		matched to an agent $\widetilde{a}$ which she prefers to $a'$ in a 
stable matching. No 
		quota-two hospital~$h$ which issues a proposal to some residents 
		$r_1,\dots r_{u(h)}$ can be matched to a resident which it prefers to 
$r_i$
		for all $i\in [u(h)]$ and which does not propose to $h$ in a stable 
matching.
	\end{observation}
	Using this observation, it is possible to prove that an agent cannot be 
	matched worse than the proposal it holds in a stable matching.
	\begin{lemma} \label{le:deletion2}
		If a resident or a quota-one hospital~$a^*$ holds the proposal of 
		another 
		agent~$a'$, then there
		cannot exist a stable matching in which $a^*$ is matched to an agent 
		$\widetilde{a}$ to 
		which it prefers $a'$. 
	\end{lemma}
	\begin{proof}
		For the sake of contradiction, let us assume that there exists a stable 
		matching~$M$ in which a resident or quota-one hospital~$a^*$ is matched 
		to an agent~$\widetilde{a}$ while $a^*$ holds a
		proposal from an agent $a'$ it prefers to $\widetilde{a}$. 
		
		If $a^*$ is a quota-one hospital, then $\widetilde{a}$ and
		$a'$ are both residents and by \Cref{ob:ma}, resident~$a'$ needs to 
		prefer~$a^*$ to $M(a')$. Thus, $(a^*,a')$ forms a blocking pair.
		
		If $a^*$ is a resident and $a'$ is a quota-one hospital,
		then by \Cref{ob:ma}, hospital~$a'$ needs to prefer~$a^*$ to~$M(a')$. 
		Thus,
		$(a',a^*)$ forms a blocking pair. 
		
		If $a^*$ is a resident and $a'$ is a quota-two hospital
		which is closed in $M$, then as $a'$ issues a proposal there 
		needs to exists a resident $r$ issuing a proposal to $a'$. By 
		\Cref{ob:ma}, resident~$r$ needs to prefer $a'$ to~$M(r)$. Thus, $a^*$ 
		and $r$
		form together a blocking coalition to open~$a'$.
		
		If $a^*$ is a resident and $a'$ is a quota-two hospital
		which is open in $M$, then by \Cref{ob:ma}, hospital~$a'$ cannot be matched to
		$u(h)$ residents which $a'$ all prefers to $a^*$. Thus, $(a',a^*)$ is a 
		blocking pair for~$M$.
	\end{proof}

	Using the two proceeding lemmas, we are now able to prove that all changes 
	made to the 
	preferences during Phase 1a do not delete any stable matchings, i.e., the 
	mutual acceptability of no hospital-residents pair occurring in any stable 
	matching is deleted in Phase 1a. 
	\begin{restatable}{lemma}{firstphasee}
		\label{lem:first-phase2}
		If an agent $a^*$ deletes another agent $a'$ from its preference list, 
		then 
		$a^*$ cannot be matched to $a'$ in any stable matching.
	\end{restatable}
\begin{proof}
		There are four types of modifications applied to the agents' 
		preferences in Phase~1a. 
		First, an agent might delete another agent from its preferences
		because one of them rejected the proposal of the other. Here, the 
		correctness follows directly from \Cref{lem:first-phase}.
		 
		Second, a resident or quota-one hospital might delete an agent 
		from its preference list because it holds a proposal it prefers 
		to the deleted agent at the end of Phase~1a. Here, the correctness follows directly from 
		\Cref{le:deletion2}.
		
		Third, $a^*$ might delete $a'$ from its preference list because it does 
		not appear on the preference list of~$a'$. This implies that $a'$ has
		deleted $a^*$ from its preferences because of one of the two 
		proceeding cases. By \Cref{lem:first-phase} and \Cref{le:deletion2}, 
		this implies that $a^*$ and $a'$ cannot be matched in any stable 
		matching.
		
		Fourth, a quota-two hospital~$a^*$ is deleted from the instance and 
		thereby also from the the preferences of $a'$, if $a'$ is the only 
		remaining agent the preference list 
		of~$a^*$.
		As $a^*$ has only $a'$ on its preference list left, the
		proceeding three observation imply that $a^*$ cannot be open in
		any stable matching and thereby that $a'$ can never be matched to $a^*$
		in a stable matching.
\end{proof}

	Now, we show that the modifications of the preferences made in Phase 1a do 
	not 
	delete any blocking coalitions or pairs, i.e., each matching that is 
	stable after applying Phase~1a is also stable before applying Phase 1a:
	
	\begin{restatable}{lemma}{firstphaseback}
		\label{lem:first-phase-backward}
		Any matching $M$ that is stable in the instance $\iafter$ after 
		applying Phase~1a is also stable in the instance $\ibefore$ before 
applying 
		Phase~1a.
	\end{restatable}
	\begin{proof}
		Let $h$ and $r$ be the first hospital and resident pair such 
		that the 
		mutual acceptability of~$h$ and $r$ is deleted and both $h$ and $r$ 
		appear in a blocking coalition $(h, \{r,r_1\})$ for
		some $r_1\in R$ or
		as a 
		blocking pair $(h,r)$ for some 
		matching $M$ in $\ibefore$ that is stable in $\iafter$.
		We now iterate over all possible cases in which the mutual 
		acceptability of two agents is deleted.
		Note that all cases which lead to a deletion of a mutual acceptability 
		reduce 
		to only four possibilities that trigger a deletion of a mutual 
		acceptability, namely 
		\begin{itemize}
			 \setlength\itemsep{0em}
			\item  that $r$ received a better proposal than $h$ (which may lead 
			to a 
			deletion because $r$ rejects the proposal from $h$ or at the end of 
			Phase 1a because $r$ holds a proposal it prefers to $h$), or 
			\item that $h$ is a quota-one hospital and received a better 
			proposal than $r$ (which may lead to a deletion because $h$ rejects 
			the proposal from $r$ or at the end of Phase 1a because~$h$ holds a 
			proposal it prefers to $r$), or
			\item that $h$ is a quota-two hospital and received $u(h)$ 
			proposals it prefers to $r$ (which may lead to a deletion because 
			$h$ rejects the proposal of $r$), or 
			\item that $h$ is a quota-two hospital and has only $r$ left on 
			its preferences (which leads to a deletion because $h$ rejects 
			the proposal of $r$ or because $h$ is deleted from the instance).
		\end{itemize}

		{\bfseries Case 1: } $r $ received a proposal $h'$ it prefers to $h$.
		
		Thus, at the end of Phase 1a, resident~$r$ holds a proposal from a hospital
		$h''$ it prefers to $h$.
		However, as $r$ forms a blocking pair or coalition with $h$, resident $r$
		needs to be matched worse than~$h$ in $M$ and therefore also worse than
		$h''$.
		
		{\bfseries Case 1.1: } $h''$ is a quota-one hospital.
		
		Since $r$ is the first resident on the preferences of $h''$ in 
		$\iafter$, hospital $h''$ is matched worse than~$r$ or is not matched 
		at all in $M$, implying
		that $(h'', r)$ is
		a blocking pair or $(h'',\{r\})$ a blocking coalition for $M$ 
		in~$\iafter$
		(since $r$ holds a proposal of 
		$h''$, the mutual acceptability has not been deleted and they both 
		accept each other in $\iafter$.
		
		{\bfseries Case 1.2: } $h''$ is a quota-two hospital.
		
		Before proposing to $r$, hospital $h''$ received a proposal from a
		resident $r' \neq r$.
		Since $h''$ did not reject~$r$, it still holds at least one proposal
		from some resident~$r''$.
		Therefore $r''$ cannot be matched better
		than 
		$h''$ in $M$ (as $h''$ is the first hospital in the preferences of 
		$r''$ in \iafter), and as $h''$ proposes to~$r$, hospital $h''$ can 
		have at most $\uq (h'') -1$ agents it prefers to $r$ on its 
		preferences in \iafter.
		If~$h''$ is closed in $M$, then $(h'', \{r, r''\})$ is a blocking
		coalition for $M$ in~$\iafter$. 
		Otherwise, $(h'', r)$ is a blocking
		pair for $M$ in $\iafter$.
		
		{\bfseries Case 2: } $h$ is a quota-one hospital and received a 
		proposal it prefers to $r$.
		
		Then, $h$ holds a proposal from a resident $r'$ it prefers to $r$ at
		the end of Phase~1a.
		Moreover, hospital~$h$ is the first hospital in the preferences of $r'$ in
		\iafter.
		Thus, as $(h,r)$ is a blocking pair for~$M$ in~$\ibefore$, hospital $h$ cannot
		be matched to $r'$. Thus, $(h,
		r')$ is a blocking pair for $M$ in~$\iafter$, a contradiction.
		
		{\bfseries Case 3: } $h$ is a quota-two hospital and received $\uq (h)$ 
		proposals from residents it 
		prefers to~$r$.
		
		Let $s_1, \dots, s_{\uq (h)}$ be the residents whose proposal $h$ holds 
		at the end of Phase 1a.
		Since~$h$ is the first hospital in the preferences of $s_i$ for all 
		$i\in [\uq (h)]$, it follows that none of~$s_1, \dots, s_{\uq
			(h)}$ prefers~$M ( s_i)$ 
		to $h$. However, as $(h,r)$ is a blocking
		pair for $M$ in~$\ibefore$, hospital~$h$ is undersubscribed or there exist one
		resident that is 
		matched to~$h$ to which~$h$ prefers~$r$. Thereby, for at least one 
		$j\in [\uq 
		(h)]$, resident~$s_j$ is not matched to $h$ in $M$ which implies that 
		$(h,s_j)$
		is a blocking pair for $M$ in $\iafter$, a contradiction. 
		
		{\bfseries Case 4: } $h$ is a quota-two hospital with only $r$ on its 
		preference list.

		Then $h$ needs to be closed in $M$ (as it has been deleted from the 
		instance $\iafter$/ its preferences are empty in $\iafter$).
		Thus, by our assumption, $(h, \{r, r_1\})$ is a blocking coalition to 
		open~$h$ in $\ibefore$.
		However, for $h$ to have only $r$ on its preference list left, $r_1$ 
		needs to reject $h$. Then, we are again in Case 1 and can thereby
		conclude that $r_1$ and $h$ together cannot be part of a blocking coalition in~$M$.
	\end{proof}
	
	We finish the description and correctness proof of Phase 1a by drawing 
	two 
	conclusions about the set of agents matched in a stable matching: 

\begin{corollary}
	\label{cor:first-phase}
	After applying Phase 1a, 
	\begin{itemize}
	\setlength\itemsep{0em}
		\item[1.] all residents holding a proposal are
		matched in all
		stable 
		matchings, and
		\item[2.] if a hospital $h$
		holds
		$\ell (h)$ proposals, then $h$ is open in all stable matchings.
	\end{itemize}
\end{corollary}
	\begin{proof}
		1. Assume for the sake of contradiction that there exists a stable 
		matching $M$ in which a resident $r$ currently holding proposal $h$ is 
		unmatched. If $h$ is a quota-one hospital, then we know by 
		\Cref{ob:ma} that $h$ cannot be matched to a resident it 
		prefers to $r$. Thus, $(h,r)$ blocks~$M$. If $h$ is a quota-two
		hospital, let $r'\neq r$ be a resident that activated~$h$ by proposing 
		to it. By \Cref{ob:ma}, resident~$r'$ cannot be matched to a hospital
		she prefers to $h$. Thus, if $h$ is closed in $M$, then $(h, \{r,r'\})$ forms a
		blocking coalition to open~$h$. Now assume that $h$ is open. As $h$ 
		proposed to $r$, by 
		\Cref{ob:ma}, there cannot 
		exist $u(h)$ residents that $h$ prefers to~$r$ and that are matched 
		to $h$ in a stable matching. Consequently, if $h$ is open, $(h,r)$ is a 
		blocking pair.   
		
		2. Assume for the sake of contradiction that there exists a stable 
		matching $M$ in which a hospital~$h$ which received $\ell (h)$
		proposals from residents $r_1, \dots, r_{\ell (h)}$ is closed. From
		\Cref{ob:ma} it follows that $r_1, \dots, r_{\ell(h)}$ cannot be
		matched to a hospital they prefer to $h$. Thus, $\{r_1,\dots,r_{\ell (h)}\}$ form a
		blocking coalition to open $h$. 
	\end{proof} 
	
	\subsubsection{Phase 1b} \label{a:p1b}
	Recall that in Phase 1b, we replace each quota-two hospital $h$ 
	holding at least two proposals by~$u(h)$ hospitals $\hone$,
	\dots, $h^{\uq (h)}$ with lower and upper quota one 
	with the same preferences
	as~$h$. In the preferences of all residents, $h$ is replaced by $h^1 \succ 
	\dots \succ h^{\uq (h)}$.
	 We now prove that there exists a one-to-one 
	mapping between the set of stable matchings before and after the 
application 
	of Phase 1b:

	\begin{restatable}{lemma}{splitting}
		\label{lem:splitting}
		There exists a one-to-one mapping between the stable matchings in the 
		instance~$\ibefore$ before applying Phase 1b and in the instance 
		$\iafter$ 
		after applying Phase 1b. 
	\end{restatable}
	\begin{proof}
		Let $h$ be a quota-two hospital holding two proposals, and let 
		$\ibefore$ be the instance before splitting hospital~$h$ and $\iafter$ the
		instance after splitting $h$.
		We define a function $\sigma$ by mapping a matching~$M$ with $M(h) = \{r_1,
		\dots, r_k\}$ and $h$ preferring $r_i$ to $r_{i+1}$ for all $i\in [k]$ 
		in~$\ibefore$ to 
		the matching~$\sigma (M) := (M\setminus \{(h , \{r_1, \dots, r_k\})\})
		\cup \{ (h^1, r_1), \dots, (h^k, r_k) \}$ in $\iafter$.
		We now show that $\sigma$
		is a bijection
from the set 
		of stable matchings in $\ibefore$ to the stable matchings in $\iafter$, proving the lemma.
		
		Let $M$ be a stable matching in \ibefore.
		By \Cref{cor:first-phase}, $h$ is open in $M$.
    We claim that~$M' :=\sigma (M)$ is a stable matching in $\iafter$.
		Any blocking coalition or pair must contain one resident from~$r_1, r_2,
		\dots, r_k$ or a hospital from \hone, \dots, $h^{\uq (h)}$, as all other agents are matched the same in~$M$ and~$M'$.
		If a blocking coalition or pair involves $r_i$ for some $i\in [k]$, 
		then it must also involve a
		hospital $h^j$ for some~$j\in [\uq (h)]$, as there does not exist a hospital
		that $r_i$ prefers 
		to~$h^1$, \dots, $h^k$ which she does not prefer to $h$.
		However, since $h^\ell$ and $r_\ell$ prefer each other to $r_p$ and
		$h^p$ for $\ell < p$, there is no such blocking coalition or pair.
		Next, assuming that for some $r\in R\setminus\{r_1,\dots, r_k\}$ 
		and~$i\in 
		[\uq (h)]$, pair~$(h^i,r)$ is a blocking pair for $M'$, then $r$ 
		prefers $h^i$ to $M'(r)$ and therefore also $h$ to $M(r)$.
		If $ i \le k$, then, as $(h^i,r)$ blocks $M'$, $h^i$ (and thus $h$)
		prefers $r$ to $r_i $.
		If $i > k$, which implies that $h^i$ is closed in~$M'$, then $h$ is 
		undersubscribed in~$M$.
		In both cases, $(h,r)$
		is a blocking pair for $M$,
		a contradiction.
		Note that $\sigma$ is obviously injective.
		
		Vice versa, let $M'$ be a stable matching in \iafter.
		Note that by our initial assumption, $h$ holds at least two proposals 
		in~$\ibefore$ and that $h^1$
and~$h^2$ are the first choices of the two residents proposing to $h$.
		Therefore, hospitals~$\hone$ and~$\htwo$ are both matched in $M'$.
		Note that due to the stability of $M'$,
		there exists some $2\le k\le \uq (h)$ such that $h^i$ is matched to an agent $r_i$ for $i \le k$ and $h^i$ is unmatched for all $i  > k$.
		We claim that $M := \sigma^{-1} (M') =\bigl(M' \setminus \{(\{h^i, \{r_i\}) : i\in
		[k]\}\bigr) \cup \{(h, \{r_1, \dots, r_k\})\}$ is a stable matching in~\ibefore.
		Any blocking coalition or pair for $M$ in $\ibefore$ must involve $h$ 
		and a resident from~$R\setminus\{r_1, \dots, r_k\}$. As $h$ is open, 
		there
		cannot exist a blocking coalition in $M'$. 
		If there exists a blocking pair $(h, s)$ for some resident~$s\in
		R\setminus\{r_1, \dots, r_k\}$, then $s$ 
		prefers $h$ to $M(s)$ and $h$ is undersubscribed or for some $j\in 
		[k]$, hospital~$h$ prefers $s$ to $r_j$.
		In the first case, there exists some $h^j$ such that $h^j$ is unmatched, 
implying that $(h^j, \{s\})$ is a blocking coalition for $M'$ in~$\iafter$, 
contradicting the stability of $M'$.
		In the later case, $(h^j, s)$ is a blocking pair for $M'$ in $\iafter$, 
contradicting the stability 
		of $M'$.
		Note further that $\sigma(M)=M'$ which implies that $\sigma^{-1}$ is a 
		right inverse of $\sigma$ and thus that $\sigma$ is surjective. To see 
		this, recall that all hospitals $h^i$ for $i\in [u(h)]$ rank all 
		residents in the same order and that all residents accepting them prefer
		$h^i$ to $ h^{i+1}$ for all $i\in [u(h)-1]$. Thus, by the stability of 
		$M'$ 
		it needs to hold that $h$ prefers $r_i$ to $r_{i+1}$ for all $i\in 
		[k-1]$ 
		and thereby that $\sigma(M)=M'$.
	\end{proof} 
	
	We conclude with several observations about the situation after the 
	excessive 
	application of Phase 1 that we will use for the definition of a generalized rotation
	in Phase~2:
	\begin{lemma}
		\label{lem:first-phase-rotation}
		After the application of Phase 1,
		\begin{enumerate}
		\setlength\itemsep{0em}
			\item[(1)] each agent either holds and issues exactly 
			one proposal or neither holds nor issues a proposal.
			\item[(2)] each quota-two hospital 
			\begin{itemize}
			\setlength\itemsep{-0.1em}
				\item has upper quota two, and holds the proposal
				of one of its first two residents,
				\item holds a proposal and has only two agents left on its 
				preferences, or
				\item holds no proposal.
			\end{itemize}
			\item[(3)] if an agent $a^*$ appears on the preference list of 
			another 
			agent $a'$ and $a'$ has more than one agent on its preference 
			list, then 
			$a^*$ also has more than one agent on its preference list.
			
		\end{enumerate}
	\end{lemma}	
	\begin{proof}
		(1) No quota-two hospital holds more than one proposal, as a hospital
		holding two proposals gets split into multiple quota-one hospitals 
in Phase 
		1b. Thereby, every agent can hold at most one proposal.
		
		We claim that every agent that holds a proposal also issues a proposal. 
		If a resident or quota-one hospital does not issue a proposal, then its preferences are empty.
		Consequently, it cannot hold a proposal.
		For quota-two hospitals
		$h$ holding a proposal from a resident $r$, the claim holds, as in the 
		case where all residents except $r$ reject the proposal $h$ and thereby 
        $h$ does not 
		issue a proposal, $h$ rejects the proposal~$r$.
		
		Consequently, as no agent can hold more than one proposal and each 
		agent that holds a proposal also issues one, every agent issues exactly 
		one proposal and holds exactly one proposal or neither holds nor issues 
		a proposal. 
		
		(2) The second statement directly follows from the first one by
observing that a quota-two hospital only issues one 
		proposal in the case where its upper quota is two and it holds the 
		proposal of one of its first two residents or if it only has two 
		residents on its preferences left.
		
		(3)
		If the preferences of $a^*$ only contain $a'$, then $a^*$ is not a 
		quota-two hospital, $a^*$ proposes to $a'$ and by (1) $a^*$ also 
		receives a 
		proposal from $a'$.
		Thus, $a'$ is not a quota-two hospital (as any quota-two hospital receiving exactly one proposal~$r$ does not propose to~$r$, and there are no agents receiving multiple proposals by (1)).
		As $a'$ proposes to~$a^*$, agent~$a^*$ is the first agent in the 
		preferences of $a'$, and as $a'$ receives a proposal from $a^*$, all 
		agents after $a^*$ are deleted from the preferences of $a'$.
		Thus, the preferences of $a'$ only contain $a^*$.
	\end{proof} 
	
	\subsubsection{Phase 2}\label{a:p2}
		We start by proving that as 
	long as there exists a resident 
	with more 
	than one hospital on her preference list, a generalized rotation is 
	guaranteed to exist and can be found in linear time.
	
	\begin{lemma}
		\label{lem:one-proposal}
		Unless the preference list of every resident contains at most one
		hospital, a generalized rotation always exists and can be found in
		$\mathcal{O} (n)$ time, where $n$ is the number of residents.
	\end{lemma}
	\begin{proof}
        If there exists a resident $r$ with at least two hospitals on her 
preference list, it is always possible to find a generalized rotation by starting with
setting $a_1:=r$ and then use \ab{1} to \ab{3b(ii)}
as defined above to find $b_2$ and subsequently \ba{1} to \ba{3} to find $a_2$. 
We continue doing so until the process cycles, i.e., we have found some $i\neq 
j$
with $a_i=a_j$. Note that it is also enough to find some $i\neq j$ with 
$b_i=b_j$ as this directly implies that $a_i=a_j$. 
As each computed pair~$(a_i, b_i)$ contains at least one resident, it follows
that the running
time for finding a generalized rotation lies in $\mathcal{O}(n)$.

It remains to argue
that for residents and quota-one hospitals $s$ with at least two agents on 
their preference list, if we set $a_i=s$, then $b_{i+1}$ always exists, has 
at least two agents on its preference list and is unique. Moreover, we need to 
prove the same for $a_i$ if we set $b_{i}=s$. 
Note that for all \ab{1} to \ab{3b(ii)} and \ba{1} to \ba{2}, the successor 
clearly always exists and is unique. For \ba{3}, it needs to holds that resident
$b_i$ holds a proposal of the last hospital $h$ on her preference 
list. By \Cref{lem:first-phase-rotation}, this implies that $h$ also receives 
exactly one proposal. Thus $a_i$ is well-defined and unique.
Furthermore, all computed successors appear on the 
preferences of agents whose preference list has length at least two and, by 
\Cref{lem:first-phase-rotation},
also have at least two agents on their preference list. 
	\end{proof}

	Given an agent $a_i$ appearing in a generalized rotation, we say that 
\emph{\ab{$x$} 
applies to $a_i$} for~$x\in \{1, 2, \text{3a}, \text{3b(i)}, \text{3b(ii)}\}$ 
if 
case \ab{$x$} needs to be applied to compute $b_{i+1}$ from~$a_i$.
	Similarly, for an agent $b_i$ appearing in a generalized rotation, we say 
that 
\emph{\ba{$x$} applies to $b_i$} for~$x\in \{1,2,3\}$ if case \ba{$x$} needs to 
be applied to compute $a_{i}$ from $b_i$. Moreover, in the following, for two 
agents $a,a'$ of which one is a resident and the other is a quota-one hospital, 
we write $\{a,a'\}\in M$ to denote that $a$ and $a'$ are matched to each other 
in~$M$. We now show a list of statements needed to prove the correctness of 
Phase~2. 
We start by considering whether and where agents might repeatedly appear in a 
single generalized rotation: 

		\begin{observation}\label{obs:b}
			In any generalized rotation, $b_i \neq b_j$ holds for $i \neq j$.
		\end{observation}
		
		\begin{proof}
			If $b_i = b_j$ holds, then also $a_{i} = a_{j}$ holds, which 
			implies $i =  
			j$.
		\end{proof} 
		
		\begin{lemma}\label{lem:no-same-g}
			Let $(a_1, b_1)$, \dots, $(a_k, b_k)$ be a generalized rotation.
			If a stable matching $M$ contains $\{a_i, b_i\}$ or $(h (a_i), 
			\{a_i, b_i\})$ for all $i \in [k]$, then it holds for all $i,j \in 
[k]$ that
\begin{itemize}
	  \setlength\itemsep{0em}
  \item[(i)] $a_i \neq b_j$,
  \item[(ii)] $g (a_i)\neq g (a_j)$ for $i\neq j$ whenever $a_i$ and $a_j$ are 
  residents, and
  \item[(iii)] $h(a_i) \neq g(a_j)$ for $i\neq j$ whenever $a_i$ and $a_j$ are 
  residents.
\end{itemize}
		\end{lemma}
		
		\begin{proof}
            To prove (i), for the sake of contradicting, let us assume that 
$a_i = b_j$ for some~$i,j \in [k]$.
            By \ba{1} to \ba{3} it follows that $i \neq j$ (for \ba{3}, note 
that since $b_j$ is part of a rotation, it follows that $b_j$ has at least two 
hospitals in its preferences, and thus does not propose to $h$, the last 
hospital on its preferences).
            If $a_i$ is a quota-one hospital or $a_i$ is a resident and 
            $h(a_i)$ is a quota-one hospital, then it also holds that $b_i = 
a_j$, since $M$ contains $\{a_\ell, b_\ell\}$ for $\ell \in \{i,j\}$.
            Thus, $a_i$ and $b_i$ propose to each other, implying that their
            preference lists contain only one agent.
            It follows that they cannot be part of a generalized rotation.
            Otherwise $a_i$ is a resident and $h(a_i)$ is a quota-two hospital.
            Thus, $M$ contains both $(h(a_i), \{a_i, b_{i}\})$ and $(h(a_{j}),
            \{a_{j}, a_i\})$, implying that $h(a_i) = h(a_{j})$.
            It follows that both $a_i$ and $a_{j}$ issue a proposal to
            $h(a_i)$.
            Since every hospital with at least two proposals got split, it
            follows that
            $a_i = a_{j}$, a contradiction.
            Thus, $a_i \neq b_j$ for all $i,j \in [k]$ holds.
            
            To prove (ii), assume that $g:= g (a_i) = g (a_j)$ for some $i \neq 
j\in [k]$ such that $a_i$ and $a_j$ are residents.
            Note that $g$ cannot be a quota-one hospital, as from this it would 
            follow that \ab{3a} applies for both~$a_i$ and~$a_j$ which would 
            imply 
            that $b_{i+1} = b_{j+1}$, contradicting \Cref{obs:b}. Thus, we
            assume that 
            $g$ is a quota-two hospital and distinguish between the case where 
            $g$ received 
            a proposal and the case where $g$ did not receive a proposal.
            We will find for both
            cases a 
            resident $r$ appearing both as $a_i$ and $b_{j+1}$, contradicting 
            (i).
            
            {\bfseries Case 1: } $g$ did not receive a proposal.
            
            Then $g$ cannot appear as $h (a_p)$ for any $p\in [k]$, and
            \ab{3b(ii)} applies for both $a_i$ and~$a_j$.
            Let $r $ be the first resident on $g$'s preference list.
            Then $b_{\ell+1} = r$ if and only if $a_\ell \neq r$ for $\ell \in 
            \{i,j\}$.
            Since $a_i \neq a_j$ by the definition of a generalized rotation and $b_{i  +
            	1} \neq 
            b_{j+1}$ by \Cref{obs:b}, it follows without loss of generality
            that  $r = a_i$ and~$r\neq a_j$.
            Since $r\neq a_j$, it follows that $b_{j+1} = r = a_i$.
            
            {\bfseries Case 2: } $g$ received a proposal from a resident $r$.
            
            Then $b_{\ell+1} = r$ if and only if $a_\ell \neq r$ for $\ell \in 
            \{i,j\}$ (note that \ab{2} only applies if $a_\ell = r$).
            It follows by \Cref{obs:b} that w.l.o.g.\ $r = a_i$ and $r\neq 
a_j$. Since $r\neq a_j$, it follows that $b_{j+1} = r  = a_i$.

          To see (iii), assume that there exists some $i\neq j$ such that $h(a_i) = g (a_j)$.
          By (ii), we have that $g(a_j) = h(a_i) \neq g(a_i)$.
          Therefore, $h(a_i)$ is \infl.
          Thus, it follows that $a_i$ and $b_i$ are the only agents in the preferences of $h(a_i)$ and thus $\{a_i, b_i\} = \{a_j, b_{j+1}\}$.
          By (i), it follows that $a_i = a_j$, a contradiction.
		\end{proof}
		
		The definition of a generalized rotation and \Cref{lem:first-phase-rotation} imply the following
		observations.
		\begin{observation} \label{ob:ab2}
			If for some $a_i$ case \ab{2} applies, i.e., $a_i$ is a resident and
			$h(a_i)$ 
			is a flexible hospital, then 
			\begin{enumerate}
			\setlength\itemsep{-0.2em}
				\item $b_i$ and $b_{i+1}$ are also residents, and
				\item $a_i$ and $b_i$ are the two first residents on $h(a_i)$'s 
				preference list and $b_{i+1}$ is the third resident on 
				$h(a_i)$'s 
				preference list.
			\end{enumerate}
		\end{observation}
		\begin{proof}
			As \ab{2} applies for $a_i$, agent~$b_{i+1}$ is the second-most preferred
			resident 
			on the preferences of $h(a_{i})$ 
			which is not $a_{i}$.
			By the definition of~$h(a_i)$, resident $a_i$ proposes to~$h(a_i)$, 
and is therefore by \Cref{lem:first-phase-rotation} among the first two 
residents in the preferences of~$h(a_i)$. 

            We claim that \ba{3} applies for $b_i$. 
			If \ba{1} applies for~$b_{i}$, then~$a_i$ is the last choice of
quota-one hospital~$b_i$ and therefore proposes to $b_i$, but we have $b_i \neq 
h(a_i)$ since $h(a_i)$ is a quota-two hospital. If \ba{2} applies for
$b_i$, then $a_i$ is a quota-one hospital, a contradiction.
			Thus, \ba{3} applies for~$b_i$.
			As the preferences of resident $b_i$ are non-empty, it holds a 
            proposal from the last hospital on its preferences.
			Thus, by \ba{3}, it holds that $h(a_i)$ proposes to $b_i$.
			By \Cref{lem:first-phase-rotation}, $h (a_i)$ has upper quota two. 
			From this, it follows
that $b_i$ and $a_i$ are among the first two residents in the preferences of $h(a_i)$.
		\end{proof}
		
		From \Cref{lem:first-phase-rotation}, we can conclude the following 
		observation:
		\begin{observation} \label{ob:hai}
			Let $(a_1, b_1), \dots, (a_k, b_k)$ be a generalized rotation. For all $i\in
			[k]$ such that $a_i$ is a resident and $h(a_i)$ is a quota-two hospital, hospital
			$h(a_i)$ holds the proposal of $a_i$ and if $h(a_i)$ is open in a 
			stable matching $M$, then exactly two residents are matched to 
			$h(a_i)$,
			one of which must be $a_i$.
		\end{observation}
		\begin{proof}
			Let $a_i$ be a resident and $h(a_i)$ a quota-two hospital for 
			some $i\in [k]$. Then, by the definition of Phase 1a, $h(a_i)$ 
			holds the proposal of $a_i$. There exist two possibilities. If 
			$h(a_i)$ has upper quota two, then $a_i$ is one of the fist two 
			residents on the preference list of $h(a_i)$ by 
			\Cref{lem:first-phase-rotation}. Thus, if $h(a_i)$ is open in a 
			stable matching, then exactly two residents need to be matched to 
			$h(a_i)$. One of these two needs to be $a_i$, as otherwise 
			$(a_i,h(a_i))$ is a blocking pair.
			Otherwise, by \Cref{lem:first-phase-rotation}, $h(a_i)$ has only 
			two residents on its preferences left one of which is $a_i$. From 
			this, it directly follows that if $h(a_i)$ is open, then $a_i$ and 
			the other resident left on the preferences of  $h(a_i)$
			need 
			to be matched to it.
		\end{proof}
	As the final step before proving the correctness of Phase 2, we prove that, 
	similar as for classical rotations in the \textsc{Stable Roommates} 
	problem, in a stable matching,
	either each agent pair in a generalized rotation 
	is matched (either to each other or to a specific hospital) or none of 
them.  
		\begin{lemma}\label{le:rotation-eli}
			Let $(a_1, b_1), \dots, (a_k, b_k)$ be a generalized rotation.
			If a stable matching~$M$ contains $\{a_i, b_i\}$ or $(h (a_i), 
			\{a_i, b_i\})$ for some $i\in [k]$, then $M$ contains $\{a_i,
			b_i\}$ or~$(h (a_i), \{a_i, b_i\})$ for all~$i \in [k]$.
		\end{lemma}
		
		\begin{proof}
			Note that by \Cref{ob:hai}, the hospital $h(a_i)$
			can be matched to at most two residents for every $i \in [k]$.
			Assume, for the sake of contraction, that for some $i\in [k]$, $M$ 
			contains $\{a_i, b_i\}$ or $(h 
(a_i), \{a_i, 
			b_i\})$ and 
			that the following two assumptions hold:
			\begin{align}
			\{a_{i-1}, b_{i-1}\}\notin M \tag{1}\label{a1}\\
			( h (a_{i-1}), \{a_{i-1}, b_{i-1}\}) \notin M \tag{2}\label{a2}
			\end{align}
			
			We will now make a case distinction over all possible assignments for $a_i$, $b_i$,
			$a_{i-1}$, and~$b_{i-1}$ and argue that $\{a_i,
			b_i\}\in M$ or $(h (a_i), \{a_i, b_i\}) \in M$ together with
			Assumption~(1) and
			Assumption~(2) implies a blocking pair or coalition for~$M$, a contradiction to the stability of~$M$.

			{\bfseries Case 1: } $a_{i-1} $ is a quota-one hospital.

			Then $b_i$ is a resident.
          Since $a_i$ or $h(a_i)$ is the last hospital in the preferences of~$b_i$, matching~$M$ matches $b_i$ to the last hospital on her preferences, and thus $b_i$ prefers~$a_{i-1}$ to $M(b_i)$.
          Note that $b_{i-1}$ is the best resident on $a_{i-1}$'s preferences and $b_i$ is the second-best resident on $a_{i-1}$'s preferences.
			As $M(a_{i-1} ) \neq b_{i-1}$ by Assumption~(1), it follows that $a_{i-1}$ prefers $b_i$ to~$M(a_{i-1})$.
			Thus, $\{a_{i-1}, b_i\}$ is a blocking pair, a contradiction.

			{\bfseries Case 2: } $a_{i-1} $ is a resident.

			\textbf{Case 2.1: } $b_i$ is a quota-one hospital.

			Then $a_i$ is the last resident on the preferences of $b_i$ and thus $b_i$ prefers $a_{i-1}$ to $a_i$.
			As $b_i$ is a hospital, hospital~$h(a_{i-1})$ is not \fl.
			Thus, Assumptions~(1) and~(2) imply that $a_{i-1}$ is not matched to~$h (a_{i-1})$.
			As $a_{i-1}$ is not matched to the first hospital~$h(a_{i-1})$ in its preferences and also not to the second hospital~$b_i$, it follows that $\{a_{i-1}, b_i\}$ is a blocking pair.

			\textbf{Case 2.2: } $b_i$ is a resident.

			Then $a_i$ or $h(a_i)$ is the last hospital on $b_i$'s preferences, and thus $b_i$ prefers $g(a_{i-1})$ to~$M (b_i)$.
			If $g(a_{i-1}) = h(a_{i-1})$, then $g(a_{i-1})$ is the first 
			hospital in $a_{i-1}$'s preferences. Otherwise, $h (a_{i-1})$ is 
			\infl or a quota-one hospital, and thus by Assumptions~(1) and~(2) 
			$M(a_{i-1}) \neq h(a_{i-1})$.
			In both cases, it follows that $a_{i-1}$ does not prefer $M(a_{i-1})$ to $g(a_{i-1})$.
			If $g(a_{i-1}) $ is closed, then $(g(a_{i-1}), \{a_{i-1}, b_i\})$ is a blocking coalition to open $g(a_{i-1})$.
			Otherwise $b_i$ is among the first two residents in the preferences 
			of $g(a_{i-1})$, and thus, $\{g(a_{i-1}), b_i\}$ is a blocking pair.
	\end{proof} 
	
	Proving the correctness of Phase 2, we start by showing that if there 
	exists a stable matching before the elimination of a generalized rotation, there also
	exists a stable matching after the elimination which matches the same set of residents.
	\begin{restatable}{lemma}{rot}
		\label{lem:rotation}
		Let $(a_1, b_1), \dots, (a_k, b_k)$ be a generalized rotation.
		If an instance admits a stable matching~$M$, then it still admits a
		stable matching~$M'$ which matches the same residents and quota-one 
hospitals 
		after the elimination of this generalized rotation.
		Furthermore, $M$ and $M'$ open the same number of hospitals.
	\end{restatable}
	\begin{proof}
		Let $M$ be a stable matching. We distinguish two cases.
		
		{\bfseries Case 1: } $M$ contains neither $\{a_i, b_i\}$ nor $(h (a_i), 
		\{a_i,
		b_i\})$ for all $i\in [k]$.
		
		We claim that $M$ is also contained in the reduced preferences after 
		the elimination of the generalized rotation. For the sake of contradiction, we
		assume that the mutual acceptability of a hospital-residents pair that 
		is part of $M$ was deleted in the elimination of the generalized rotation and show
		that this leads to a contradiction.
		
		Assume that the mutual acceptability of a resident $r $ and a 
		hospital $h$ with $M(r) = h$ was deleted.
		If $h$ is a quota-one hospital, then the only residents that were 
		deleted from the preferences of $h$ are residents $r$ which appear 
		together with $h$ in a 
		generalized rotation pair, i.e., $(h,r)=(a_i,b_i)$ or $(r,h)=(a_i,b_i)$ for some
		$i\in [k]$. However, by our initial assumption that $M$ does not 
		contain~$\{a_i,b_i\}$ for $i\in [k]$, this implies that~$M(r)\neq h$,
		a 
		contradiction.
		Thus $h$ is a quota-two hospital. A resident~$r$ is only deleted
		from the preferences of $h$ if $h= h (a_i)$ and $r= b_i$ for some $i 
		\in [k]$. As~$h(a_i)$ is the top-choice of~$a_i$, hospital $h(a_i)$
		holds the proposal of $a_i$. This implies by
		\Cref{lem:first-phase-rotation} that $a_i$ is ranked among the two best 
		residents in the  
		preferences of $h$, and that $M$ has either upper quota two or 
		exactly two residents on its preference list. As we have assumed 
		that~$M(b_i)=h(a_i)$ and~$(h(a_i),\{a_i,b_i\})\notin M$, it needs to 
		hold 
		that 
		$M(h) = \{b_i, r'\}$ for some~$r'\neq a_i$.
		As~$h(a_i)$ is the top-choice of $a_i$ and $a_i$ is among the two most 
		preferred residents of $h(a_i)$, pair~$(h, a_i)$
		blocks $M$, a contradiction.
		\medskip
		
		{\bfseries Case 2: } $M$ contains $\{a_i, b_i\}$ or $(h (a_i), \{a_i, b_i
		\})$ for all $i\in [k]$.

		\paragraph{Construction of new matching.} We claim that replacing, for 
		all $i\in [k]$, 
		$\{a_i, b_i\}$ or $(h 
		(a_i), \{a_i, b_i\})$ by~$(g(a_i), \{a_i,
		b_{i+1}\})$ if $a_i$ and $b_{i+1}$ are both residents and by
		$\{a_i,
		b_{i+1}\}$ otherwise, results in a stable matching $M'$. Note that in 
		the former case $g(a_i)$ is a quota-two hospital, as otherwise~$b_{i+1}$
		would be a hospital instead of a resident.
		Clearly, $M'$ matches the same set of residents and quota-one hospitals.
		
		\paragraph{Feasibility.} We first show that $M'$ does not violate any
		quotas. First of all note 
		that by \Cref{ob:hai} matching $M$ matches only $a_i$ and $b_i$ to 
		$h(a_i)$.
		All hospitals $h(a_i)$ are either closed or respect their 
lower quota in $M'$:
		If $h(a_i)$ is \fl, then it is also open in $M'$, as in this case, 
		$g(a_i) = h(a_i)$ and $(g(a_i), \{a_i, b_{i+1}\})\in M'$.
		Otherwise $h(a_i)$ is \infl, and then  $h(a_i) \neq g(a_j)$ for all $j\in [k]$ by 
		\Cref{lem:no-same-g}.
		Thus, $h (a_i)$ is closed in $M'$.

		Next, we show that every hospital obeys its upper quota in $M'$.
		Matching~$M'$ can only violate the upper
		quota of a hospital~$h = g(a_i)$ that is also open in $M$, as all replacements
		described above do not increase the number of residents matched to all other hospitals.
		So assume for a contradiction that there exists some hospital~$h = g(a_i)$ such that the upper quota of $h$ is violated in~$M'$.
		If $h (a_i)$ is \fl, then, as $g (a_i) \neq g(a_{i'}) 
		$ for all $i \neq i' \in [k]$ by \Cref{lem:no-same-g}, it follows that $M' (h) = (M(h) \setminus\{b_i\}) \cup \{b_{i+1}\}$ and $h$ obeys its upper quota also in~$M'$.
		Otherwise $h(a_i)$ is \infl.
		We will now show that $h$ is closed in $M'$, implying that $h$ 
		obeys its upper quota in $M'$.
		By \Cref{lem:no-same-g}, it holds that $h \neq h(a_j)$ for all $j\in [k]$.
		This
		implies that either \ab{3b(i)} or \ab{3b(ii)} applies for~$a_i$. If
		\ab{3b(i)} applies, then $h$ holds the proposal of $b_{i+1}$ which 
		implies by \Cref{lem:first-phase-rotation} that $b_{i+1}$ is among the 
		first two residents on the preference list of $h$. If \ab{3b(ii)} 
		applies, then by definition, $b_{i+1}$ is one of the two most 
		preferred residents  on the preference list of $h$.
		In both cases, $h$ prefers $b_{i+1}$ to one of
		the residents it is matched to in $M$.
		Since $h$ is on $b_{i+1}$'s preference list, it follows that 
		$b_{i+1}$ prefers $h$ to~$M$, as we have assumed that
		either $\{a_{i+1}, b_{i+1}\}$ or $(h (a_{i+1}), \{a_{i+1}, b_{i+1}\})$ 
		is part of~$M$
		which in both cases implies that $b_{i+1}$ is matched to its least 
		preferred hospital in $M$.
		Therefore, $h$ is closed in $M$, as otherwise $(h,b_{i+1})$ forms a blocking pair in 
		$M$.
		It follows that $h$ does not violate its upper quota.
		Therefore $M'$ is a feasible matching.

		Next we show that for any pair $(r, h)\in M'$ the acceptability of $r$ and $h$ has not been deleted.
		Since $a_i \neq b_j $ for all $i,j \in [k]$ by \Cref{lem:no-same-g}, this can only occur if $g(a_{i-1}) = h(a_i)$.
		Thus, $a_i$ issues a proposal to $g(a_{i-1}) = h (a_i)$.
		Since $a_i\neq a_{i-1}$ and every hospitals holds at most one proposal, 
		this implies that resident $a_{i-1}$ cannot issue a proposal to 
		$g(a_{i-1})$ and thereby that $h(a_{i-1})\neq g(a_{i-1})$ needs to 
		hold. 
		It follows that \ab{3b(i)} applies for $a_{i-1}$, and thus, $b_i= a_i$, 
		contradicting \Cref{lem:no-same-g}.

		To see that the number of open hospitals in $M'$ equals the number of open hospitals in $M$, first note that the set of open quota-one hospitals stays the same.
		For every \fl hospital~$h (a_i)$, we have that $h(a_i) = g(a_i)$ is also open in $M'$.
		For every \infl hospital $h(a_i)$, we have that $h(a_i) \neq g(a_j)$ for all $j\in [k]$ by \Cref{lem:no-same-g}.
		It follows that $h(a_i)$ is closed in $M'$, but $g(a_i)$ is open in~$M'$.
		In $M$, hospital $h(a_i)$ is open, while we have already seen that $g(a_i)$ is closed in~$M$.
		It follows that $M$ and $M'$ open the same number of hospitals.
		
		It remains to show that $M'$
		is 
		stable. 
		
		\paragraph{Stability.} Any blocking pair or coalition for $M'$ obviously needs to
		include an
		agent that is matched differently in $M'$ than in $M$, i.e., $a_i$, 
		$b_i$, $h(a_i)$, or $g(a_i)$ for some $i\in [k]$. We now show one after 
		each other for each of these four types of agents that they cannot be 
		part of a blocking pair or coalition in $M'$. However, we start by 
		observing the following: 
		
		{\bfseries Claim 1: } For all $i\in [k]$, agent $b_i$ is matched differently
		in $M'$ than in $M$ and $b_i$ prefers~$M'$ to $M$.
		
		By the definition of \ba{1} to 
		\ba{3}, $\{a_i,b_i\}\in M$ or $(h(a_i), \{a_i,b_i\})\in M$ implies that
		$M(b_i)$ is the last agent in the preferences of~$b_i$.
		Thus, it is enough to show that $b_i$ is matched differently in~$M$ and $M'$.
		If $b_i$ is a hospital, then $M(b_i) = a_i$ and $M' (b_i) = a_{i-1}$ are both residents, and $a_i \neq a_{i-1}$ implies that $M(b_i)\neq M' (b_i)$.
		If $b_i$ is a resident and~$a_i$ is a hospital, then again $M(b_i) = M' 
		(b_i)$ implies $a_{i} = a_{i-1}$, a contradiction.
		If $b_i$ is a resident and $a_i$ is also a resident, then $M(b_i) = h(a_i)$.
		If $M(b_i) = M' (b_i)$ holds, then, since $M(b_i)$ is a quota-two 
		hospital, it holds that $M' (b_i) = g (a_{i-1})$. For the sake of 
		contradiction, assume that $h(a_i)=M(b_i)=M'(b_i)=g (a_{i-1})$. This 
		leads to a contradiction as argued in the last paragraph in the part on 
		feasibility in this proof.
		
		{\bfseries Claim 2: } For all $i\in [k]$, agent $a_i$ is neither part 
		of a blocking
		pair nor of a blocking coalition in~$M'$.
		
		If $a_i$ is a quota-one hospital, then \ba{2} and \ab{1}
		apply for $b_i$ and $a_i$ which implies that $a_i$ is the last hospital on
		$b_i$'s preferences and $b_i$ is the first resident and $b_{i+1}$ is 
		the second resident on the preferences of $a_i$. Thereby, $a_i$ is 
		matched to its second-most preferred resident in $M'$. This implies
		that the only possible blocking pair including 
		$a_i$ in $M'$ is $(a_i,b_i)$. However, as by Claim 1, resident~$b_i$ prefers
		$M'$ to $M$, 
		this cannot be a blocking pair.
		
		If $a_i$ is a resident, then $a_i$ can be either part
		of a blocking pair $(a_i, h)$ with some hospital~$h\in H$ or part of a
		blocking coalition to open some hospital $h\in H$.
		If $h(a_i)$ is \fl, then $a_i$ is matched to her first choice in $M'$ 
		and thus is neither part of a blocking coalition nor a blocking pair.
		Hence, we assume that $h(a_i)$ has quota one or is \infl.
		As $a_i$ is matched
		to her second-most preferred hospital $g(a_i)$ in $M'$, the hospital in any blocking coalition or pair containing~$a_i$ is $h := h(a_i)$.
    We distinguish
		two cases based on whether $h(a_i)$ is a quota-one or quota-two hospital.
		
		If $h(a_i)$ is a quota-one hospital, then \ba{1} applies for $b_i = h(a_i)$, and $a_i$ is $b_i$'s least preferred
		resident. However, as $b_i$ prefers $M'(b_i)$ to $a_i$ as shown in
		Claim 1, $(h(a_i), a_i)$ cannot be a blocking pair for $M'$.
		
		If $h(a_i)$ is a quota-two hospital, then $h(a_i)$ is an \infl hospital as $h(a_i) \neq g(a_i)$.
		Thus, the preferences of $h(a_i)$ contain only two
		residents, namely $a_i$ and~$b_i$. As $a_i$
		cannot be matched to~$h(a_i)$ in $M'$, hospital~$h(a_i)$ is closed in $M'$.
		However, we have observed above that $b_i$ prefers~$M'(b_i)$ to~$M(b_i)=h(a_i)$ and thereby~$b_i$ is not part of a blocking
		coalition.
		Thus, no blocking coalition to open $h(a_i)$ exists, and~$a_i$ is 
		neither part of a blocking pair nor a blocking coalition.
		
		{\bfseries Claim 3: } For all $i\in [k]$ such that $a_i$ is a resident, hospital $h(a_i)$ is neither part of a
		blocking pair nor a blocking coalition in $M'$.
		
		If $h (a_i)$ is a quota-one hospital, then $h(a_i) = b_i$ and Claim 1 
		implies 
		that $h (a_i)$ prefers~$M'$ to~$M$. Moreover, Claim 2 implies that
		all residents that prefer $M$ to $M'$ are not part of a
		blocking pair or coalition. Thus, any blocking pair or coalition in $M'$ involving~$h(a_i)$ also blocks~$M$ and thus, such blocking pairs and coalitions do not exist.
		
		Thus, $h(a_i)$ is a quota-two hospital. Note that this implies that 
		\ba{3} 
		applies for $b_i$.
		Thereby $b_i$ is a resident holding a
		proposal from $h(a_i)$ and by \Cref{lem:first-phase-rotation} that $a_i$
		and $b_i$ are the first two residents on the preferences of $h(a_i)$.

		If $h(a_i)$ is an \infl hospital, then $h(a_i)$ is closed in $M'$, as 
		by Claim 1, $b_i$ is not matched to~$M(b_i)=h(a_i)$ in $M'$. The only
		possible 
		blocking coalition to open $h(a_i)$ is~$(h,\{a_i,b_i\})$. However, as 
		$b_i$
		prefers $M'(b_i)$ to $M(b_i)=h(a_i)$, resident $b_i$ cannot be part of this
		coalition.
		Thus, $h(a_i)$ is contained neither in a blocking coalition nor a blocking pair.
		
		If $h(a_i)$ is a flexible hospital, then by 
		\Cref{ob:ab2},
		$h(a_i) $ has upper quota two, $b_{i+1}$ is a resident and $a_i$,
		$b_i$, and $b_{i+1}$ are the first three residents in the preferences 
		of $h(a_i)$. Moreover, it holds that
		$h(a_i) = g(a_i)$. Thus, $h(a_i)$ is matched to $a_i$ and $b_i$ 
		in $M$ and to $a_i$ and $b_{i+1}$ in~$M'$. Consequently, the only 
		possible blocking pair involving $h(a_i)$ is $(h(a_i), b_i)$.
		As by Claim~1, resident~$b_i$ prefers $M'(h)$ to $M(h)=h(a_i)$, this pair cannot
		be blocking.   

		{\bfseries Claim 4: } For all $i\in [k]$ such that $a_i$ is a resident, 
		hospital $g(a_i)$ is neither part of a blocking pair nor a blocking 
		coalition in $M'$.

		Since $g(a_i)$ is open in $M'$, it cannot be part of a blocking coalition.
		Assume that there exists a blocking pair $(r, g(a_i))$ in $M'$. Note 
		that it needs to hold that $r\neq b_{i+1}$, as $b_{i+1}$ is matched to~$g(a_i)$ in~$M'$.
		Furthermore, $r\neq a_j$ for all $j\in [k]$ since by Claim~2 agents~$a_j$
		are not part of a blocking pair in $M'$.
		
		If $g(a_i)$ is closed in $M$, then $(g(a_i), \{b_{i+1}, r\})$ is a 
		blocking coalition in $M$: By Claim~1, resident~$b_{i+1}$ prefers $g(a_i)= M' (b_{i+1})$ to
		$M(b_{i+1})$.
		If $r=b_i$ for some
		$i\in [k]$, then~$r$ is matched better in $M'$ than in $M$, which 
		implies that $(g(a_i), \{b_{i+1}, r\})$ blocks $M$.
		Otherwise we have~$r\neq a_i,b_i$
		for all $i\in [k]$, and then $r$ is matched to the same hospital in~$M$ and
		$M'$, which again implies that $(g(a_i), \{b_{i+1}, r\})$ blocks $M$, a contradiction.
	
    Therefore, $g(a_i)$ is open in $M$.
		Recall that we have proven in the first
		part of this proof in the paragraph on feasibility that if a hospital 
		$g(a_i)$ is open in $M$, then it needs to hold that $g(a_i) = h(a_j)$ 
		for some $j\in [k]$.
		Thus, by Claim 3, $g(a_i) = h(a_j) $ is not part of a  blocking coalition or pair, a contradiction.
		
		{\bfseries Claim 5: } For all $i\in [k]$, agent $b_i$ cannot be part of 
		a blocking pair or a blocking coalition in $M'$. 
		
		For all $i\in [k]$, by Claim 2-4, we know that $b_i$ cannot form a 
		blocking pair or coalition with any agents appearing as $a_j$, 
		$h(a_j)$, or $g(a_j)$ for some $j\in [k]$. Thus, $b_i$ needs to form a 
		blocking pair or coalition with agents that are matched to the 
		same partners in~$M$ and in $M'$. However, as by Claim 1, agent~$b_i$ prefers
		$M'(b_i)$ to $M(b_i)$, this implies that the such a blocking pair or 
		coalition would also block $M$, a contradiction.
	\end{proof} 
	
	Before we prove that no stable matching is created by the elimination of a 
	generalized rotation, we identify a sufficient criterion for the non-existence of a
	stable matching, which directly follows from \Cref{cor:first-phase} and 
	\Cref{lem:rotation}.
	\begin{corollary}\label{cor:empty-preferences}
	 If the preference list of a resident gets empty during Phase 2
	 or contains only quota-two hospitals which have only one resident on their 
preferences, or 
	 the preference list of a quota-one hospital gets empty, then the instance 
	 does not admit a stable matching. 
	\end{corollary}
	Thus, we reject an instance as soon as the preference list of a resident
	becomes empty by eliminating a rotation. This also implies that in the 
	following we can assume that the set of residents
	with non-empty preferences is the same before and after eliminating a 
	rotation. 

	\begin{restatable}[]{lemma}{forward}
		\label{lem:forward}
		Let $\ibefore$ be the instance before the elimination of a generalized 
		rotation (but after applying Phase 1) and $\iafter$ the instance after 
		eliminating a generalized rotation from~$\ibefore$.
		Any stable matching $M$ in $\iafter$ which matches all residents 
		with non-empty preferences in~$\ibefore$ is also stable in
		$\ibefore$. 
	\end{restatable}
	\begin{proof}
		Let $M$ be a stable matching in $\iafter$ matching all residents 
		with non-empty preferences in \ibefore. For the sake of
		contradiction, let us assume that $M$ is not stable in~$\ibefore$. 
		Let $C$ be a blocking coalition or blocking pair involving some 
		hospital $h$ for~$M$ in~$\ibefore$. The only possibility that $C$ is 
		not blocking in $\iafter$ is that the mutual acceptability of $h$ and 
		some resident occurring in $C$ has been deleted. We now make a case 
		distinction whether $h$ is a quota-one or a quota-two hospital and 
		argue for both cases that $C$ cannot block $M$ in $\ibefore$. 

		{\bfseries Case 1: } Hospital $h$ is a quota-one hospital. 
	
		Let $r$ be a resident such that $h$ and $r$ together form a blocking
		pair or coalition for $M$ in~$\ibefore$. For the mutual acceptability 
		of $h$ and $r$ to be 
		deleted, it needs to hold that $\{h, r\} = \{a_i, b_i\}$ for some $i\in 
		[k]$.
		Agent~$a_i$ is then the last agent on the preferences of~$b_i$.
		Therefore, $\{h, r\}$ can only block~$M$ if $b_i$ is unmatched, which 
		we assume in the following.
		Since every resident with non-empty preferences in~$\ibefore$ is matched in $\iafter$, it follows that $b_i = h$ and $a_i = r$.
		We claim that $a_{i-1}$ prefers $b_i
		$ to
		$M(a_{i-1})$:
		The preferences of resident $a_{i-1}$ start with a quota-one or \infl
		hospital 
		$h (a_{i-1})$, followed by $g(a_{i-1}) = b_i$, as \ab{3a} has to apply 
		for $a_{i-1}$.
		Furthermore, $a_{i-1}$ is not matched to~$h(a_{i-1})$ in $M$ because 
		either $h(a_{i-1})$ is \infl, implying that $h(a_{i-1})$ has lower 
		quota two but only $a_{i-1}$ on its preferences in~$\iafter$ after the 
		elimination of the generalized rotation, or
		$h(a_{i-1})$ has lower quota one and the mutual acceptability of 
		$a_{i-1}$ and $h(a_{i-1}) = b_{i-1}$ has been deleted by the 
		elimination of the generalized rotation.
		Thus, as $b_i$ is unmatched in $M$, $(a_{i-1}, b_i)$ is a blocking pair 
		for $M $ in \iafter, a 
		contradiction to the stability of $M$.
		
    {\bfseries Case 2: } Hospital $h$ is a quota-two hospital. 
		
		Then, it either holds that $C = (h, r_1)$ is a blocking pair or $C  = 
		(h, \{ r_1, r_2\})$ is a blocking coalition for two residents $r_1,
		r_2\in R$ for $M$ in $\ibefore$. 
		Assume without loss of generality that the acceptability of $r_1$ and 
		$h$ has been deleted. This implies that $h=h(a_i)$ and $r_1=b_i$ for 
		some~$i\in[k]$. Note that $r_1$
		prefers all hospitals on her preference list to $h$.
		As $h(a_i)$ is contained in the preferences of $r_1$ in $\iafter$, the 
		preferences of $r_1 $ are non-empty and thus, by our assumption on $M$, 
		resident $r_1$ is matched in $M$.
		This implies that
		$r_1$ does not prefer $h$ to her partner in~$M$, a contradiction.
	\end{proof} 

\subsubsection{Proof of \Cref{t:q2}}
We are now ready to prove that the full algorithm works correctly.
	\qq*
	\begin{proof}
		Each application of Phase 1a or 2 removes the mutual
		acceptability of at least one resident and at least one hospital, while 
		each execution of Phase 1b reduces the number of quota-two hospitals.
		Since we may assume that each hospital has upper quota at most $n$, it 
		follows that there are at most $\mathcal{O}(mn)$ hospitals at any stage 
		of the algorithm.
		Thus, there are at most $\mathcal{O}(n \cdot mn)= 
		\mathcal{O}(n^2 m)$ mutual acceptabilities at any stage of the 
		algorithm, and all executions of Phase 1b together can be performed in 
		$\mathcal{O}(n^2 m)$.
		Since all except the first proposal an agent receives results in a 
		rejection and reduces the number of mutual acceptabilities, it follows 
		that all executions of Phase 1a can be performed in $\mathcal{O}(n^2 
m)$ time.
		A generalized rotation can be found in $\mathcal{O}(n)$ time by \Cref{lem:one-proposal}.
		Thus, all executions of Phase 2 can be executed in $\mathcal{O}(n\cdot n^2 m)$ time.
		Thus, the runtime $\mathcal{O}(n^3 m)$ follows.

		If \Cref{alg} returns a matching $M$, then $M$ matches all agents from 
		$S$.
		Thus, the preferences of a resident can only get empty in the first 
		application of Phase 1a, which implies that the set of residents with 
		empty preferences is the same before eliminating the first generalized 
		rotation and after eliminating the last generalized rotation.
		Consequently,
		\Cref{lem:first-phase-backward,lem:splitting,lem:forward} imply that 
		$M$ is also stable in the input instance.

		Vice versa, if the input instance contains a stable matching $M^*$, 
		then 
		\Cref{cor:first-phase} implies that $M^*$ matches all agents from $S$.
		Consequently, 
		\Cref{lem:splitting,lem:rotation,lem:first-phase2} ensure that there 
		also exists a 
		stable matching covering $S$ after all modifications performed by the algorithm.
		As the algorithm returns the only matching~$M$ which covers~$S$ and is still present at the termination of the algorithm, matching~$M$ is stable, and the correctness of the algorithm follows.
	\end{proof}

	From the correctness of our algorithm, we can also translate (a weaker 
	version of) the so-called Rural Hospitals 
	Theorem~\cite{RePEc:ucp:jpolec:v:92:y:1984:i:6:p:991-1016,10.2307/1913160} 
	to \HRLUQtwo.
	The Rural Hospitals Theorem states that in every Hospital Residents instance (without lower quotas), every stable matching matches the same residents and for each hospital, the same number of residents is matched to it in every stable matching.
	The latter part of this statement does not hold for \HRLUQ, as the set of 
	open hospitals can differ for two stable matchings.
	For instance, consider the following instance with two residents $r_1$ and $r_2$ as well as two hospitals $h_1 $ and $h_2$ with lower and upper quota two, and the following preferences: $r_1 : h_1 \succ h_2 $ and $r_2 : h_1 \succ h_2$.
	Then both $(h_1, \{r_1 , r_2\})$ and $( h_2, \{r_1, r_2\})$ are stable matchings.
	However, we can still show that every stable matching in an \HRLUQtwo 
	instance matches the same set of residents and opens the same number of 
	hospitals:
	
	\begin{proposition}
	  Given an instance of \HRLUQtwo, every stable matching matches the same set of residents and opens the same number of hospitals.
	\end{proposition}

	\begin{proof}
	  \Cref{cor:first-phase} implies that the set of matched residents is the same in every stable matching.
	  
	  Assume for a contradiction that there exist two stable matchings which 
	  open a different number of hospitals.
	  Then \Cref{lem:first-phase2,lem:splitting,lem:rotation} imply that the 
	  algorithm preserves that there are two stable matchings opening a 
	  different number of hospitals.
	  However, at the termination of the algorithm, there is only one stable matching, a contradiction.
	\end{proof}

    \section{Introducing Ties} \label{se:ties}
   
	It is also possible to extend the models that we have considered so far by 
	allowing for ties in the preferences of residents 
	and hospitals (if they have preferences). We call the resulting problems 
	\HRLUQT, \HRLQT, and \HALUQT. Notably, \HRLUQT now also generalizes
	\HALUQT, as it is possible to make the hospitals indifferent among all 
	residents. In the following, we start by considering \HALUQT before turning
	to \HRLQT and finally to \HRLUQT.
	Note that the stability concept we apply here, that is, all agents need to 
	be better off after a deviation and not just indifferent, is usually called 
	\emph{weak} stability in the literature (opposed to \emph{strong} and 
	\emph{super-}stability, where one or both agents of a blocking pair might 
	be indifferent between the blocking pair and the matching).
	
	All hardness results for  
	\HRLUQI clearly also hold for \HRLUQIT. Moreover, the only positive result 
	for \HRLUQ, that is, the FPT result for the parameterization by the number 
	$m$ of hospital, still holds, as it is 
	possible to adapt the ILP constructed in \Cref{pr:HRLUQI-FPTM} in a
	straightforward way (residents are still fully characterized by their 
	preference list and the number of different preference lists can be bounded 
	in a function of
	$m$): 
	\begin{corollary} \label{co:HRLUQIT-m}
		Parameterized by the number $m$ of hospitals, \HRLUQIT is 
		fixed-parameter tractable.
	\end{corollary}
	
	We now show that that both \HRLQT and
	\HRLUQT are NP-complete even if all hospitals 
	have lower (and upper) quota two (in contrast to our polynomial-time algorithm for \HRLUQtwo).
	To do so, we reduce from \textsc{Complete Stable Marriage with Ties and 
	Incomplete Preferences}.
	In this problem, we are given a set $U$ of $t$ men and a set $W$ of $t$ 
	women, where each man has preferences over a subset of women acceptable to him
	and each woman has preferences over a subset of men acceptable to her and 
	preferences may contain ties. We refer to the combined
	set of men and women as agents. A matching~$M$ is called \emph{complete} if
	each man is matched to a woman acceptable to him and is called
	\emph{stable} if no man-woman pair~$(m,w)\in U\times W$ exists such that both $m$
	and $w$ find each
	other acceptable and prefer each other to the agent they are
	matched to in $M$. The goal is to decide the existence of a
	complete and stable matching.
	Notably, \textsc{Complete Stable Marriage with Ties and Incomplete 
	Preferences} is equivalent to \HRLUQT if all hospitals have upper quota one 
	and the task is to decide whether the instance admits a stable matching 
	which matches all residents.
	
	Reducing from this problem, we now show that \HRLQT is already NP-complete, 
	even 
	if all hospitals have lower quota at most two. The general idea of the 
	reduction is similar to the one sketched at the beginning of 
	\Cref{sec:strict} combined with an appropriate penalizing component. 
	\begin{proposition}
		\label{p:q2T}
		\HRLQT, \HRLUQT, and \HALUQT are NP-complete, even if all hospitals 
		have lower 
		quota two, each hospital is accepted by at most two residents and each 
		resident accepts at most four hospitals.
	\end{proposition}
	\begin{proof}
      Consider an instance of \textsc{Complete Stable Marriage with Ties and 
      Incomplete Preferences} where each agent accepts at most three other 
      agents.
      This problem is known to be NP-complete~\cite{IrvingMO09}.
      
      \textbf{Construction:} We construct an instance of \HRLQT as follows. For 
      each man $m\in U$, we add a resident $r_m$, and for each woman $w\in W$, 
      we add a resident $r_w$.  
      For each man-woman pair~$(m, w)\in U\times W$ where both $m$ and $w$ find 
      each other acceptable, we add a hospital~$h_{(m, w)}$ (matching $(m,w)\in 
      U\times W$ in the given instance corresponds to matching both $m$ and $w$ 
      to~$h_{(m, w)}$).
      Furthermore, for each man~$m\in U$, we add a penalizing component 
      consisting of four hospitals $h^*_m$, $h_m$, $h_m'$, and $h_m''$  with 
      lower quota 
      two and four residents~$r^*_m$, $r_m'$, $r_m''$, and~$r_m'''$.
      
      The preferences of the residents are as follows.
      For each man~$m\in U$, the preferences of resident~$r_m$ arise from the 
      preferences of $m$ by replacing each woman~$w\in W$ by the 
      hospital~$h_{(m, w)}$. 
      At the end, we append the hospital~$h^*_m$.
      For each woman~$w\in W$, the preferences of resident~$r_w$ arise from the 
      preferences of~$w$ by replacing each man~$m\in U$ by the 
      hospital~$h_{(m,w)}$.
      For each~$m\in U$, the preferences of the residents in the penalizing 
      components are: $$r^*_m 
      : h^*_m \succ h_m, \qquad r_m' : h_m \succ h_m', \qquad r_m'' : h_m' 
      \succ h_m'', \qquad r_m''' : h_m'' \succ h_m.$$
      Note that the residents $r_m'$, $r_m''$, and $r_m'''$ together with the 
      hospitals $h_m$, $h_m'$, and $h_m''$
      correspond to the instance from 
      \Cref{ob:counter} not admitting a stable matching. This directly implies 
      that all hospitals~$h^*_m$ need to be closed in a stable matching and 
      thus all men $m\in U$ need to be matched to a hospital $h_{(m, w)}$ on 
      their preferences.

      Since each man~$m\in U$ and each woman~$w\in W$ in the given instance has 
      at most three agents 
      on its preferences, it follows that each resident $r_m$ has at most four 
      hospitals on her preferences, and each resident $r_w$ has at most three 
      hospitals on her preferences.

     {\bfseries ($\Rightarrow$)} Given a complete stable matching~$N$ in the
      \textsc{Complete Stable Marriage with Ties and 
      	Incomplete Preferences} instance, we get a stable 
      matching $M$ in the constructed \HRLQT instance by taking $(\{r_m, r_w\}, 
      h_{(m, w)})$ for each $(m, w)\in N$ 
      and $(\{r^*_m, r_m'\}, h_m)$ and $(\{r_m'', r_m'''\}, h_m'')$ for each 
      man~$m\in U$.
      Since every open hospital is matched to all resident it accepts, there 
      cannot be blocking pairs.
      Moreover, there cannot be a blocking coalition of the form $(\{m, w\}, 
      h_{(m, w)})$ for some $(m,w)\in U\times W$, since otherwise $(m, w)$ 
      would be 
      a blocking pair for $N$.
      Thus, any blocking coalition must involve a hospital from a penalizing component.
      However, since each man~$m\in U$ is matched in $N$, it follows that $r_m$ 
      prefers the hospital she is matched to over $h^*_m$. 
      Thus,
      there cannot be a blocking coalition to open $h^*_m$. Moreover, as $r_m'$ 
      is 
      matched to her top-choice, there is no blocking coalition to open $h_m'$.

      {\bfseries ($\Leftarrow$)} Vice versa, given a stable matching~$M$ in the 
      constructed 
      \HRLQT instance, first obverse 
      that for each $m\in U$, matching~$M$ matches resident~$r_m$ to a hospital~$h_{(m,
      w)}$ for a 
      woman~$w\in W$, as otherwise $r_m$ is matched to $h^*_m$ together with 
      $r^*_m$ and there is a blocking coalition in the corresponding
      penalizing component (see \Cref{ob:counter}).
      Thus, the matching $N:= \{(m, w)\mid (\{r_m, r_w\}, h_{(m,w)})\in M\}$ 
      matches all men, and it remains to show that it is stable.
      Assume for a contradiction that there exists a blocking pair~$(m, w)\in U 
      \times W$ for $N$.
      Then, both $r_m$ and $r_w$ prefer $h_{(m, w)}$ over their assigned 
      hospital in $M$, and it follows that $(\{r_m, r_w\}, h_{(m, w)})$ is a 
      blocking coalition in $M$.
	\end{proof}
	Notably, the construction falls under \Cref{ob:equiv}, implying that all
	three problems are still hard if stability only requires that no blocking 
	coalition exists. 
	
	Recall that for both \HRLQ and \HRLUQ, it is possible to decide whether 
	there exists a stable matching in which exactly 
	a given set of hospital is open in polynomial time. This result also 
	implies fixed-parameter 
	tractability with respect to the number
	of hospitals. In the following, we investigate whether these positive results can be extended
	if ties in the preferences are allowed. In fact, for \HRLQT, it is still 
	possible to decide 
	whether there exists a stable matching opening a given set of hospitals 
	$H_{\open}\subseteq H$ in polynomial time. The underlying idea is that each 
	resident needs to be
	matched to 
	one of her most preferred hospitals from $H_{\open}$ in every stable
	matching~$M$. Thus, we already know for each resident the set of hospitals
	she prefers to the hospital she is matched to in $M$. Using this, one 
	can easily check whether $M$ admits a 
	blocking coalition. The remaining task is then to find an assignment of 
	residents to one of their most preferred hospitals from $H_{\open}$ that 
	respects the lower 
	quota of each hospital from $H_{\open}$. This task can be formulated as an 
	instance of bipartite maximum matching.
	\begin{proposition}\label{pr:LTies}
		\label{pr:HRLQT-H'}
		Given a subset of hospitals $H_{\open}\subseteq H$, deciding whether 
		there exists a stable matching in an \HRLQT instance $(H,R)$ in which 
		exactly the hospitals from $H_{\open}$ are open is solvable in 
		$\mathcal{O}(nm + n^{2.5})$ time.
		Parameterized by the number $m_{\quota}$ of hospitals with non-unit 
		lower quota, 
		\HRLQT is solvable in 
		$\mathcal{O}\big((nm+n^{2.5})2^{m_{\quota}}\big)$ 
		time.
	\end{proposition}
	\begin{proof}
		Let $(H,R)$ be the given \HRLQT instance and let 
		$H_{\open}$ be the set of hospitals to open. For every resident~$r\in R$, let $H_r \subseteq H_{\open}$
denote the
		set of hospitals from $H_{\open}$ to which $r$ does not prefer 
		any 
hospital
		from $H_{\open}$ (the set of her top-choices in $H_{\open}$). Note that 
		$r$ needs to be assigned to one of the hospitals 
		from $H_r$ in every stable matching, as otherwise she forms a blocking 
		pair with any open hospital from $H_r$.
		As $r$ is indifferent
		among the hospitals in~$H_r$, for the stability of the 
		resulting matching, it does
		not matter to which of the hospitals from $H_r$ she is assigned. Thus, 
		we 
		return NO if there exists a hospital~$h\in H \setminus H_{\open}$ and a
		set of $l(h)$ residents~$r_1,\dots r_{l(h)}$ with $r_i$
		preferring~$h$ to the hospitals in $H_{r_i}$ for all~$i\in [l(h)]$.
		This can be done in $\mathcal{O} (mn)$.
		After that, the problem reduces to finding a feasible matching where 
		each 
		resident $r$ is assigned to a hospital from $H_r$ and each hospital
		from~$H_{\open}$ meets its lower quota. Note that if such a matching 
		exists, then this matching is stable, as the fact that we have not 
		rejected $H_{\open}$ before implies
		that there does not exist a blocking coalition. 
		The task of finding a feasible matching can be solved by computing a 
		maximum matching in a bipartite graph~$G$ which we construct as 
		follows: We add one vertex for each resident $r\in R$ and $\lowerq (h) 
		$ vertices for each hospital~$h\in H_{\open}$ and an edge between a 
		vertex corresponding to a resident $r\in R$ and a vertex corresponding 
		to a hospital $h\in H$ if $h \in H_r$.
		Note that we may assume that $\sum_{h\in H_{\open}} \lowerq (h) \le n$ (otherwise we have a trivial No-instance).
		Therefore, $G$ has $\mathcal{O} (n)$ vertices and $\mathcal{O} (n^2)$ 
		edges, and thus, using Hopcroft and Karp's 
		algorithm~\cite{DBLP:journals/siamcomp/HopcroftK73}, a matching in $G$ 
		can be computed in 
		$\mathcal{O} (n^{2.5})$ time. There exists a bipartite matching which 
		matches all vertices corresponding to hospitals if and only if there 
		exists a feasible matching of the residents to the hospitals where each 
		resident $r\in R$ is matched to a hospital from $H_r$. 
		By iterating over all possible subsets of~$H$ and subsequently 
		employing the algorithm described before, it is possible to check 
		whether there exists a stable matching in a \HRLQT instance 
		in~$\mathcal{O}((n m + n^{2.5})2^m)$~time. 
		
		Let $H^{\quota}\subseteq H$ be the set of  hospitals with non-unit 
		lower quota.
		It is possible to improve the FPT algorithm from above by instead 
		iterating 
		over all subsets $H_{\open}^{\quota}\subseteq H^{\quota}$ and 
		subsequently deciding  whether there exists a stable matching where the 
		set of open hospitals with a non-unit lower quota is exactly 
		$H^{\quota}_{\open}$. To answer this question, we need to modify the 
		algorithm from above slightly. That is, at the beginning of the 
		algorithm, we assign a resident~$r\in R$ to a quota-one hospital~$h\in 
		H$ (and delete her from the set of residents) if $r$ prefers $h$ to all 
		all hospitals from $H^{\quota}_{\open}$. The rest of the algorithm 
		remains unchanged leading to an overall running time 
		of~$\mathcal{O}((nm+n^{2.5})2^{m_{\quota}})$. 
	\end{proof}
	Note that this result implies that deciding whether there exists a stable 
	matching with~$m_{\open}$ open ($m_{\text{closed}}$ 
	closed) 
	hospitals in a \HRLQT instance lies still in XP with respect 
	to~$m_{\open}$~($m_{\text{closed}}$). 

	As \HRLUQT is a generalization of \HALUQ, by 
	\Cref{th:HRLUQIH'}, it is NP-hard to decide whether there exists a stable 
	matching in a \HRLUQT instance where exactly some given set of hospitals is 
	open. Turning to the 
	complexity of \HRLUQT parameterized by the number of $m$ hospitals, recall 
	that the basic idea of our ILP in \Cref{pr:HRLUQI-FPTM} for \HALUQ was to 
	bound the number of
	residents types in $m$.
	However, this does not work for \HRLUQT, as different hospitals may
	rank residents differently.
	In fact, it turns out that, in contrast to \HRLQT and \HALUQT, \HRLUQT is W[1]-hard with respect to the number of
	hospitals. To prove this result, we first show that deciding whether there 
	exists a stable matching in a \HRUQT instance (all hospitals only have an 
	upper quota they need to obey) which matches all residents 
	(\textsc{Com} \HRUQT) is
	W[1]-hard with respect to the number of hospitals, a result which may also 
	be of
	independent~interest.
	
	In particular, Meeks and Rastegari~\cite{MeeksR20} asked for the 
	computational complexity of \textsc{Complete Stable Marriage with Ties and 
	Incomplete Lists} when men are of a few different ``types'' and each man of the
	same type has the same preferences and women are indifferent between two 
	men of the same type.
	This corresponds to \textsc{Com} \HRUQT, as one hospital can be seen as its 
	upper quota many identical agents. Thus, the number of different types
	in Meeks and Rastegari's model corresponds to the number of hospitals in 
	our model. The following result shows that parameterized by the number of 
	different types of men, \textsc{Complete Stable Marriage with Ties and 
	Incomplete Lists} is W[1]-hard.
	
	\begin{proposition}
		\label{pr:COMHRT}
		Parameterized by the number $m$ of hospitals, \textsc{Com} \HRUQT is 
		W[1]-hard.
	\end{proposition}
	\begin{proof}
		We reduce from \textsc{Multicolored Clique} which we introduced in 
		\Cref{thm:ha-m-closed} and 
		which is W[1]-hard parameterized by the number $k$ of different colors~\cite{Pietrzak03}.
		Let 
		$G=(V,E)$ be an
		undirected graph with a partitioning of the vertices in $k$ different 
		colors $(V^1,\dots V^k)$. 
		Moreover, for each $c< d\in [k]$, let
		$E^{c, d}$ denote the set of edges with one endpoint colored
		with color $c$ and the other endpoint colored with color $d$.  Without
loss of
		generality and to simplify notation, we assume that all colors contain 
		$\ell$ vertices and for 
		all color combinations~$c<d\in [k]$, the number of edges with one
		endpoint colored in $c$ and one endpoint colored in $d$ is $p$.
		For~$c\in [k]$, we write~$V^c=\{v^c_1,\dots v^c_\ell\}$.   
		For a 
		vertex $v\in V$, let $\delta(v)$ denote the set of edges that are 
		incident to~$v$ in~$G$. Moreover, for a set of residents $R'$, let 
		$[R']$ 
		denote an arbitrary strict linear ordering of all residents in $R'$
		and $(R')$ a single tie containing all residents from $R'$.	
		
		\textbf{Construction:} We introduce a vertex selection  
		gadget for each color $c\in [k]$ and an edge selection gadget 
		for each pair of colors $c< d\in [k]$. For each $c\in
		[k]$, the vertex selection gadget consists of two 
		hospitals $h^c$ and $\bar{h}^c$ both with upper quota $\ell$ and
		$2\ell$ 
		\emph{selection residents}~$r^c_1,\dots, r^c_\ell$ and $s^c_1,\dots,
		s^c_\ell$. In 
		such a gadget, the vertex with color $c$ of the clique should be 
		selected. For 
		each~$c< d\in [k]$, the edge gadget should pick the edge between
		the selected vertex of color $c$ and that selected vertex of color 
		$d$. For this, we introduce one
		hospital $g^{c,d}$ with upper quota~$p-1$ and one hospital
		$\bar{g}^{c,d}$ with upper quota $1$ together with an edge
		resident 
		$t_e$ for each edge $e\in E^{c, d}$.
		 The 
		preferences of all residents and hospitals are as follows. 
		\begin{align*}
		\intertext{Vertex gadgets:}
		\bar{h}^c \colon & r^c_1 \succ \dots \succ r^c_\ell \succ 
		[t_e]_{e\in \delta (v^c_\ell)}\succ s^c_\ell
		\succ 
		[t_e]_{e\in \delta 
			(v^c_{\ell-1})}\succ s^c_{\ell-1} \succ \dots \succ  
			[t_e]_{e\in \delta 
			(v^c_1)}\succ s^c_{1}, &
		\forall c\in
		[k]
		\\
		h^c  \colon & r^c_1 \succ\dots \succ r^c_\ell \succ s^c_1 \succ 
		[t_e]_{e\in \delta (v^c_1)} \succ s^c_2 
		\succ [t_e]_{e\in \delta (v^c_2)} 
		\succ \dots \succ  s^c_\ell \succ [t_e]_{e\in \delta 
		(v^c_\ell)}, &
		\forall c\in 
		[k]\\ 
		r^c_i  \colon & (h^c, \bar{h}^c),\qquad s^c_i \colon (h^c,
		\bar{h}^c), \qquad  \forall i\in [\ell], \forall c\in [k]\\
		\intertext{Edge gadgets:}
		g^{c,d} \colon & (t_e)_{e\in E^{c, d}},\qquad
		\bar{g}^{c,d}  \colon (t_e)_{e\in E^{c, d}}, \qquad \forall c< d \in
[k]\\
		t_e \colon & g^{c,d}\succ h^c\succ h^{d} \succ \bar{h}^c
		\succ \bar{h}^{d}
		\succ \bar{g}^{c,d}, \qquad \forall c<d \in [k],\forall e\in
E^{c, d}
		\end{align*}
		See \Cref{fig:w-hard-com-hrt} for an example.
		For each color $c\in [k]$, selecting vertex $v^c_i\in V^c$ corresponds
		to 
		matching $s^c_1,\dots, s^c_{i}$ and $r^c_{i+1}, \dots, r^c_{\ell}$ to
		$h^c$ 
and $r^c_1, \dots, r^c_i$ and
		$s^c_{i+1},\dots, s^c_{\ell}$ to $\bar{h}^c$. For each color 
		pair~$c< d\in [k]$, selecting edge $e\in E^{c, d}$ corresponds to
		matching 
		$t_e$ to $\bar{g}^{c, d}$.
		
		\begin{figure}
		 \begin{minipage}{0.2\textwidth}
		  \begin{center}
		   \begin{tikzpicture}
		   \node (hor-dist) at (1, 0) {};
		   \node (ver-dist) at (0, 1) {};
		   
		    \node[vertex, label=180:$v_1^c$] (v11) at (0,0) {};
		    \node[vertex, label=180:$v_2^c$] (v21) at ($(v11) + (ver-dist)$) {};
		    \node[vertex, label=0:$v_1^d$] (v12) at ($(v11) + (hor-dist)$) {};
		    \node[vertex, label=0:$v_2^d$] (v22) at ($(v21) + (hor-dist)$) {};
		    
		    \draw (v11) -- (v12);
		    \draw (v11) -- (v22);
		    \draw (v21) -- (v22);
		    
		    \draw (v21) -- (v12);
		   \end{tikzpicture}

		  \end{center}

		 \end{minipage}
		 \begin{minipage}{0.7\textwidth}
		  \begin{center}
		   \begin{tikzpicture}
		   \node (hor-dist) at (1.5, 0) {};
		   \node (ver-dist) at (0, 1.5) {};
		   
			      \node[vertex, label=180:$r_{1}^c$] (r11) at (0, 0) {};
			      \node[vertex, label=180:$r_{2}^c$] (r21) at ($(r11) + 
			      (ver-dist)$) {};
			      
			      \node[vertex, label=180:$s_{1}^c$] (s11) at ($(r21) + ( 
			      ver-dist)$) {};
			      \node[vertex, label=180:$s_{2}^c$] (s21) at ($(s11) + 
			      (ver-dist)$) {};
			      
			      \node[vertex, label=0:$r_{1}^d$] (r12) at ($(r11) + 
			      6*(hor-dist)$) {};
			      \node[vertex, label=0:$r_{2}^d$] (r22) at ($(r12) + 
			      (ver-dist)$) {};
			      
			      \node[vertex, label=0:$s_{1}^d$] (s12) at ($(r22) + ( 
			      ver-dist)$) {};
			      \node[vertex, label=0:$s_{2}^d$] (s22) at ($(s12) + 
			      (ver-dist)$) {};
			      
			      \node[vertex, label=270:$t_{v_1^c, v_1^d}$] (e1) at 
			      ($0.5*(r11) + 0.5*(r12)$) {};
			      \node[vertex, label=270:$t_{v_1^c, v_2^d}$] (e2) at ($(e1) + 
			      (ver-dist)$) {};
			      \node[vertex, label={[xshift=0.cm]90:$t_{v_2^c, v_1^d}$}] 
			      (e3) at ($(e2) + (ver-dist)$) {};
			      \node[vertex, label=90:$t_{v_2^c,v_2^d}$] (e4) at ($(e3) + ( 
			      ver-dist)$) {};
			      
			      \node[squared-vertex, label=270:$h^c$, 
			      label={[yshift=-0.4cm]270:${[1,2]}$}] (h1) at ($(r11)+ 
			      (hor-dist) - 0*(ver-dist)$) {};
			      \node[squared-vertex, label={[yshift=0.4 cm]90:$\bar{h}^c$}, 
			      label=90:${[1,2]}$] (h1b) at ($(h1) + 3*(ver-dist)$) {};
			      
                  \draw (h1) edge node[pos=0.3, fill=white, inner sep=2pt] 
                  {\scriptsize $1$} node[pos=0.84, fill=white, inner sep=2pt] 
                  {\scriptsize $1$} (r11);
                  \draw (h1) edge node[pos=0.3, fill=white, inner sep=2pt] 
                  {\scriptsize $2$} node[pos=0.84, fill=white, inner sep=2pt] 
                  {\scriptsize $1$} (r21);
                  \draw (h1) edge node[pos=0.3, fill=white, inner sep=2pt] 
                  {\scriptsize $3$} node[pos=0.84, fill=white, inner sep=2pt] 
                  {\scriptsize $1$} (s11);
                  \draw (h1) edge node[pos=0.3, fill=white, inner sep=2pt] 
                  {\scriptsize $6$} node[pos=0.84, fill=white, inner sep=2pt] 
                  {\scriptsize $1$} (s21);
                  
                  \draw (h1b) edge node[pos=0.3, fill=white, inner sep=2pt] 
                  {\scriptsize $1$} node[pos=0.84, fill=white, inner sep=2pt] 
                  {\scriptsize $1$} (r11);
                  \draw (h1b) edge node[pos=0.3, fill=white, inner sep=2pt] 
                  {\scriptsize $2$} node[pos=0.84, fill=white, inner sep=2pt] 
                  {\scriptsize $1$} (r21);
                  \draw (h1b) edge node[pos=0.3, fill=white, inner sep=2pt] 
                  {\scriptsize $8$} node[pos=0.84, fill=white, inner sep=2pt] 
                  {\scriptsize $1$} (s11);
                  \draw (h1b) edge node[pos=0.3, fill=white, inner sep=2pt] 
                  {\scriptsize $5$} node[pos=0.84, fill=white, inner sep=2pt] 
                  {\scriptsize $1$} (s21);

                  \draw (h1) edge node[pos=0.2, fill=white, inner sep=2pt] 
                  {\scriptsize $4$} node[pos=0.84, fill=white, inner sep=2pt] 
                  {\scriptsize $2$} (e1);
                  \draw (h1) edge node[pos=0.2, fill=white, inner sep=2pt] 
                  {\scriptsize $4$} node[pos=0.84, fill=white, inner sep=2pt] 
                  {\scriptsize $2$} (e2);
                  \draw (h1) edge node[pos=0.2, fill=white, inner sep=2pt] 
                  {\scriptsize $7$} node[pos=0.84, fill=white, inner sep=2pt] 
                  {\scriptsize $2$} (e3);
                  \draw (h1) edge node[pos=0.2, fill=white, inner sep=2pt] 
                  {\scriptsize $7$} node[pos=0.84, fill=white, inner sep=2pt] 
                  {\scriptsize $2$} (e4);
                  
                  \draw (h1b) edge node[pos=0.2, fill=white, inner sep=2pt] 
                  {\scriptsize $6$} node[pos=0.84, fill=white, inner sep=2pt] 
                  {\scriptsize $4$} (e1);
                  \draw (h1b) edge node[pos=0.2, fill=white, inner sep=2pt] 
                  {\scriptsize $6$} node[pos=0.84, fill=white, inner sep=2pt] 
                  {\scriptsize $4$} (e2);
                  \draw (h1b) edge node[pos=0.2, fill=white, inner sep=2pt] 
                  {\scriptsize $3$} node[pos=0.84, fill=white, inner sep=2pt] 
                  {\scriptsize $4$} (e3);
                  \draw (h1b) edge node[pos=0.2, fill=white, inner sep=2pt] 
                  {\scriptsize $3$} node[pos=0.84, fill=white, inner sep=2pt] 
                  {\scriptsize $4$} (e4);
			      
			      \node[squared-vertex, label=270:$h^d$, 
			      label={[yshift=-0.4cm]270:$[1,2]$}] (h2) at ($(r12) - 
			      (hor-dist)$) {};
			      \node[squared-vertex, label={[yshift=0.4cm]90:$\bar{h}^d$}, 
			      label={90:$[1,2]$}] (h2b) at ($(h2) + 3*(ver-dist)$) {};
			      
                  \draw (h2) edge node[pos=0.3, fill=white, inner sep=2pt] 
                  {\scriptsize $1$} node[pos=0.84, fill=white, inner sep=2pt] 
                  {\scriptsize $1$} (r12);
                  \draw (h2) edge node[pos=0.3, fill=white, inner sep=2pt] 
                  {\scriptsize $2$} node[pos=0.84, fill=white, inner sep=2pt] 
                  {\scriptsize $1$} (r22);
                  \draw (h2) edge node[pos=0.3, fill=white, inner sep=2pt] 
                  {\scriptsize $3$} node[pos=0.84, fill=white, inner sep=2pt] 
                  {\scriptsize $1$} (s12);
                  \draw (h2) edge node[pos=0.3, fill=white, inner sep=2pt] 
                  {\scriptsize $6$} node[pos=0.84, fill=white, inner sep=2pt] 
                  {\scriptsize $1$} (s22);
			      
                  \draw (h2) edge node[pos=0.2, fill=white, inner sep=2pt] 
                  {\scriptsize $4$} node[pos=0.84, fill=white, inner sep=2pt] 
                  {\scriptsize $3$} (e1);
                  \draw (h2) edge node[pos=0.2, fill=white, inner sep=2pt] 
                  {\scriptsize $7$} node[pos=0.84, fill=white, inner sep=2pt] 
                  {\scriptsize $3$} (e2);
                  \draw (h2) edge node[pos=0.2, fill=white, inner sep=2pt] 
                  {\scriptsize $4$} node[pos=0.84, fill=white, inner sep=2pt] 
                  {\scriptsize $3$} (e3);
                  \draw (h2) edge node[pos=0.2, fill=white, inner sep=2pt] 
                  {\scriptsize $7$} node[pos=0.84, fill=white, inner sep=2pt] 
                  {\scriptsize $3$} (e4);
                  
                  \draw (h2b) edge node[pos=0.3, fill=white, inner sep=2pt] 
                  {\scriptsize $1$} node[pos=0.84, fill=white, inner sep=2pt] 
                  {\scriptsize $1$} (r12);
                  \draw (h2b) edge node[pos=0.3, fill=white, inner sep=2pt] 
                  {\scriptsize $2$} node[pos=0.84, fill=white, inner sep=2pt] 
                  {\scriptsize $1$} (r22);
                  \draw (h2b) edge node[pos=0.3, fill=white, inner sep=2pt] 
                  {\scriptsize $8$} node[pos=0.84, fill=white, inner sep=2pt] 
                  {\scriptsize $1$} (s12);
                  \draw (h2b) edge node[pos=0.3, fill=white, inner sep=2pt] 
                  {\scriptsize $5$} node[pos=0.84, fill=white, inner sep=2pt] 
                  {\scriptsize $1$} (s22);
			      
                  \draw (h2b) edge node[pos=0.2, fill=white, inner sep=2pt] 
                  {\scriptsize $6$} node[pos=0.84, fill=white, inner sep=2pt] 
                  {\scriptsize $5$} (e1);
                  \draw (h2b) edge node[pos=0.2, fill=white, inner sep=2pt] 
                  {\scriptsize $3$} node[pos=0.84, fill=white, inner sep=2pt] 
                  {\scriptsize $5$} (e2);
                  \draw (h2b) edge node[pos=0.2, fill=white, inner sep=2pt] 
                  {\scriptsize $6$} node[pos=0.84, fill=white, inner sep=2pt] 
                  {\scriptsize $5$} (e3);
                  \draw (h2b) edge node[pos=0.2, fill=white, inner sep=2pt] 
                  {\scriptsize $3$} node[pos=0.84, fill=white, inner sep=2pt] 
                  {\scriptsize $5$} (e4);
                  
			      \node[squared-vertex, label=90:$g^{c,d}$, 
			      label={[yshift=-0.1cm]270:${[1,3]}$}] (g) at 
			      ($0.5*(h1)+0.5*(h1b)$) {};
			      \node[squared-vertex, label=90:$\bar{g}^{c,d}$, 
			      label={[yshift=-0.1cm]270:${[1,1]}$}] (gb) at ($0.5*(h2) + 
			      0.5*(h2b)$) {};
			      
                  \draw (g) edge node[pos=0.2, fill=white, inner sep=2pt] 
                  {\scriptsize $1$} node[pos=0.84, fill=white, inner sep=2pt] 
                  {\scriptsize $1$} (e1);
                  \draw (g) edge node[pos=0.2, fill=white, inner sep=2pt] 
                  {\scriptsize $1$} node[pos=0.84, fill=white, inner sep=2pt] 
                  {\scriptsize $1$} (e2);
                  \draw (g) edge node[pos=0.2, fill=white, inner sep=2pt] 
                  {\scriptsize $1$} node[pos=0.84, fill=white, inner sep=2pt] 
                  {\scriptsize $1$} (e3);
                  \draw (g) edge node[pos=0.2, fill=white, inner sep=2pt] 
                  {\scriptsize $1$} node[pos=0.84, fill=white, inner sep=2pt] 
                  {\scriptsize $1$} (e4);
                  
                  \draw (gb) edge node[pos=0.2, fill=white, inner sep=2pt] 
                  {\scriptsize $1$} node[pos=0.84, fill=white, inner sep=2pt] 
                  {\scriptsize $6$} (e1);
                  \draw (gb) edge node[pos=0.2, fill=white, inner sep=2pt] 
                  {\scriptsize $1$} node[pos=0.84, fill=white, inner sep=2pt] 
                  {\scriptsize $6$} (e2);
                  \draw (gb) edge node[pos=0.2, fill=white, inner sep=2pt] 
                  {\scriptsize $1$} node[pos=0.84, fill=white, inner sep=2pt] 
                  {\scriptsize $6$} (e3);
                  \draw (gb) edge node[pos=0.2, fill=white, inner sep=2pt] 
                  {\scriptsize $1$} node[pos=0.8, fill=white, inner sep=2pt] 
                  {\scriptsize $6$} (e4);
		   \end{tikzpicture}

		  \end{center}

		 \end{minipage}
          \caption{An example for the reduction showing W[1]-hardness of 
          	\textsc{Com} \HRUQT parameterized by the number of hospitals from 
          	\Cref{pr:COMHRT}.
          The input graph with coloring $(V^c=\{v_1^c, v_2^c\}, V^d=\{v_1^d, 
          v_2^d\})$ is depicted on the left, while the output is depicted on 
          the right.}
          \label{fig:w-hard-com-hrt}
		\end{figure}
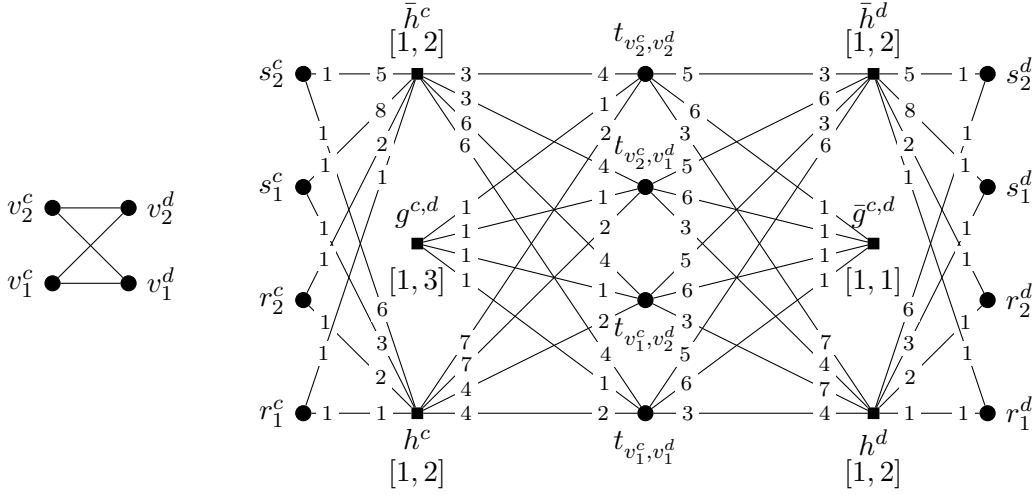

		{\bfseries ($\Rightarrow$)}
		Assume that $V'=\{v^1_{i_1}, \dots, v^k_{i_k}\}$ is a clique of 
		size 
		$k$ in $G$. We claim that the following matching $M$ is 
		a stable matching which matches all residents in the constructed 
		instance: 
		\begin{align*}
			M= & \{(g^{c,d}, \{t_e \mid e\in 
			 E^{c, d}\setminus\{\{v^c_{i_c},v^{d}_{i_{d}}\}\}\})\mid c<d \in [k]
			\}\cup \{(\bar{g}^{c,d},t_{\{v^c_{i_c},v^{d}_{i_{d}}
			\}})\mid c< d \in
			[k]\} \\
			&\cup \{(\bar{h}^c,\{r^c_1, \dots,
			r^c_{i_c},s^c_{i_c + 1},\dots, s^c_\ell \})\mid c\in
			[k]\} 
			\\
			& \cup \{(h^c,\{r^c_{i_c+1}, \dots, r^c_\ell,s^c_1, \dots, s^c_{i_c}
			\})\mid
			c\in [k]\}
		\end{align*}
		The matching $M$ is well defined, since $V'$ is a clique and clearly a 
		complete matching which respects the upper quota of all hospitals. To 
		show 
		stability, note that all residents except the edge 
		residents~$t_{\{v_{i^c_c},v_{i^{d}_{d}}\}}$ for all $c< d \in [k]$, 
		i.e., those
		corresponding to the edges lying in the clique,  
		are 
		matched to their top-choice.
		
		Thus, the only possible blocking pairs are
		$(h^c,t_{\{v_{i^c_c},v_{i^{d}_{d}}\}})$,
		$(h^{d},t_{\{v_{i^c_c},v_{i^{d}_{d}}\}})$,
		$(\bar{h}^{c},t_{\{v_{i^c_c},v_{i^{d}_{d}}\}})$,
		$(\bar{h}^{d},t_{\{v_{i^c_c},v_{i^{d}_{d}}\}})$, and $(g^{c, d}, 
		t_{\{v_{i^c_c},v_{i^{d}_{d}}\}})$ for some $c< d \in [k]$.
		However, by the construction of the matching~$M$, for each $c\in [k]$
		both $h^c$ and $\bar{h}^c$ prefer all residents assigned to them over
		all residents corresponding to the edges that are incident to 
		$v^c_{i_c}$.
		Furthermore, $g^{c, d}$ is full and is indifferent among all residents
		it accepts and thus cannot be part of a blocking pair.
		
		{\bfseries ($\Leftarrow$)}
		Let $M$ be a complete stable matching.
		Since the sum of the capacities of all hospitals equals the number of 
		residents, it follows that for each $c< d\in [k]$, one edge resident
		$t_e$ with $e=\{v^c_i,v^{d}_j\}$ for some $i,j\in [p]$ has to be
		matched to $\bar{g}^{c,d}$.
		We claim that from this it follows that the vertex-selection gadget for 
		$V^c$ has to select 
		$v^c_i$ and the vertex-selection gadget for $V^{d}$ has to select
		$v^{d}_j$, from which the correctness of the reduction easily follows.
		Without loss of generality, let us look at~$v^c_i$.
		Due to the stability of $M$, hospitals $h^c$ and $\bar{h}^c$ have to 
		prefer all 
		residents matched to them to $t_{\{v^c_i,v^{d}_j\}}$. Thereby, 
		$\ell$ residents from $\{r^c_j : j\in [\ell]\} 
		\cup \{s^c_j : j\in [i]\}$ have to be matched to $h^c$, and 
		$\ell$ 
		residents from $\{r^c_j : j\in [\ell]\} \cup \{s^c_j: j\in [i + 1, 
		\ell]\}$ have 
		to be matched to~$\bar{h}^c$.
		Thus, the agents~$\{s^c_j: j \in [i]\}$ are matched to $h^c$, and 
		$\{s^c_j: 
		j\in [i+1, \ell]\}$ are matched to $\bar{h}^c$ which corresponds to 
		selecting $v^c_i$ as the vertex of color $c$ in the clique.
	\end{proof}
	
	It is possible to reduce \textsc{Com} \HRUQT to \HRLUQT by introducing a 
	penalizing component that ensures that each resident needs to be matched to 
	a hospital from the original instance in every stable matching. This 
	implies the following result:
	\begin{proposition}
		\label{pr:HRT-W}
		Parameterized by the number $m$ of hospitals, \HRLUQT is W[1]-hard, 
		even 
		if only four hospitals have non-unit lower quota.
	\end{proposition}
	\begin{proof}
		We reduce from 	\textsc{Com} \HRUQT, which is W[1]-hard with respect to 
		the number
		of hospitals as proven in \Cref{pr:COMHRT}. 
		
		\textbf{Construction:} From an \textsc{Com} \HRUQT instance, 
		we construct a \HRLUQT instance by adding a lower quota of one to each 
		of the hospitals. Moreover, we insert a hospital $h^*$ with lower quota 
		and upper quota two which we insert after all other hospitals in the 
		preference of 
		each 
		resident. Furthermore, we add a penalizing component consisting of 
		three hospitals $h_1$, $h_2$, and $h_3$ with lower quota two and four 
		residents $r^*$, 
		$r_1$, $r_2$, and $r_3$ with the following preferences:  
		$$r^*: h^*\succ h_1, \quad r_1:h_1\succ h_2, \quad r_2: h_2\succ h_3, 
		\quad 
		r_3:h_3\succ h_1.$$ Finally, we set the preferences of $h^*$ such that 
		$h^*$ accepts all residents and has $r^*$ as its unique top-choice. 
		Note that the residents $r_1$, $r_2$, $r_3$ together with the hospitals 
		$h_1$, $h_2$, $h_3$ correspond to the instance from \Cref{ob:counter} 
		not 
		admitting a stable matchings. This implies that the 
		hospital $h^*$ is closed in all stable matchings. 
		
		{\bfseries ($\Rightarrow$)}
		If there exists a stable matching in the given \textsc{Com} \HRUQT 
		instance matching 
		all residents, then adding $(h_1,\{r^*, r_1\})$ and 
		$(h_3,\{r_2, r_3\})$ to the matching results in a stable matching in 
		the constructed \HRLUQT instance.
		
		{\bfseries ($\Leftarrow$)}
		If there exists a stable matching $M$ in the constructed \HRLUQT 
		instance, then $h^*$ needs to be closed in $M$, 
		as 
		otherwise $r^*$ is matched to it and consequently $r_1$, $r_2$, and 
		$r_3$ cannot be matched in a stable way. This implies that all 
		residents from the given \textsc{Com} \HRUQT instance need to be 
		matched to hospitals 
		from the original instance. Thus, $M$ restricted to the 
		original hospitals and residents is a stable matching in the given 
		\textsc{Com} \HRUQT instance where all residents are matched.
	\end{proof} 

	On the positive side, by guessing for each hospital the least preferred 
	resident assigned to it in a stable matching, it is again possible to 
	bound the number of 
	different resident types in a function of $m$ and subsequently solve 
	the problem using an ILP. This approach results in an XP algorithm for 
	\HRLUQT:
	
	\begin{proposition}
		\label{pr:HRT-XP}
		Parameterized by the number $m$ of hospitals, \HRLUQT lies in XP.
	\end{proposition}
	\begin{proof}
		Let $(H,R)$ be the given \HRLUQT instance.
		Similar as in \Cref{pr:HRLUQI-FPTM}, we solve this problem using an
		ILP. Let $\succsim_{t_1}$, $\succsim_{t_2}$, \dots, 
		$\succsim_{t_q}$ be a list of all 
		weak incomplete orders over $H$.
		For each $i\in [q]$, denote by $A(i)$ the set of hospitals contained in 
		the weak incomplete order~$\succsim_{t_i}$.
		The number $q$ lies in
		$\mathcal{O}(m\cdot 
		m!\cdot 2^{m})$, as each weak incomplete order can be created by
		picking a strict ordering of the $m$ elements, deleting the last~$i\in 
		[m]$
		elements and subsequently dividing the remaining elements into 
		equivalence 
		classes (by guessing the set of residents which are for some tie the first residents of this tie in the ordering).
		Unfortunately, it is not possible to
		directly bound the number of 
		different resident types by the number of different preference 
		relations because different hospitals might rank the 
		residents differently and this information is relevant when checking for
		blocking 
		pairs. Therefore, we start by guessing the subset
		of open hospitals $H_{\open}\subseteq H$ and for each such
		hospital $h\in H_{\open}$
		the least preferred resident~$r_h\in R$
		that is matched to
		$h$; in fact, we do not enforce that $r_h$ is matched to $h$. Instead, 
		we only enforce that either $r_h$ or a resident $r\neq r_h$ for which 
		it holds that $h$ is indifferent between $r$ and $r_h$ is matched to 
		$h$.
		Furthermore, we guess the set~$H_{\full} \subseteq H_{\open}$ of full hospitals.
		For each resident~$r\in R$ and hospital $h\in H_{\open}$,
		define $z_r^h$ to be~$1$ if $h$ prefers $r$ to
		$r_h$, to be 
		$0$
		if $h$ is 
		indifferent between $r$ and $r_h$, and to be $-1$ if $h$ strictly  
		prefers 
		$r_h$ to $r$. 
		We call the tuple~$(z_r^h)_{h\in H_{\open}}$ the \emph{hospital
		signature} of 
		resident~$r\in R$.
		For the sake of finding a stable matching, a resident is fully 
		characterized by her preference relation and her hospital signature.  
		For each $h\in H_{\open}$, $i\in [q]$, and $\mathbf{z}\in
		\{-1,0,1\}^m$, we 
		introduce a variable $x_{i,h}^\mathbf{z}$ denoting the number of 
		residents with preference relation $\succsim_{t_i}$ and hospital 
		signature 
		$\mathbf{z}$ that are assigned to hospital $h$. Further, let 
		$n_i^\mathbf{z}$ 
		denote the number of resident with preference relation $\succsim_{t_i}$ 
		and 
		hospital signature~$\mathbf{z}$ in the given instance.
		It is possible to check
		whether there 
		exists a stable matching respecting the current guess by solving the 
		following ILP (below the ILP, we explain the purpose of the different 
		constraints):
		
		\begin{align*}
		\sum_{\substack{h'\in H_{\open}: \\ h'\succsim_{t_i} h}} 
		x^{\mathbf{z}}_{i,h'}\geq
		n^\mathbf{z}_i, \qquad & \forall h \in H_{\open}, \forall i \in
		[q], \forall \mathbf{z}\in \{-1,0,1\}^m\text{ with } z_h=1 
		\tag{1a}\label{ILP:no-bpa2} \\
		\sum_{\substack{h'\in H_{\open}: \\ h'\succsim_{t_i} h}} 
		x^\mathbf{z}_{i,h'}
		\geq
		n_i^\mathbf{z}, \qquad & \forall h\in H_{\open} \setminus H_{\full} , \forall i \in
		[q] \text{ with } h\in A(i), \forall \mathbf{z} \in \{1,0,-1\}^m
		\tag{1b}\label{ILP:no-bpb2}\\
		\sum_{\substack{i\in [q], h'\in H_{\open}, \\ \mathbf{z}\in
				\{-1,0,1\}^m:  h \succ_{t_i} h'}}
		x^\mathbf{z}_{i,h'}\leq 
		l(h), \qquad & \forall h\in H\setminus H_{\open} \tag{2}
		\label{ILP:no-bc2}\\
		l(h)\leq \sum_{\substack{i\in [q], \\ \mathbf{z}\in
				\{-1,0,1\}^m}} x^\mathbf{z}_{i,h}  \leq u(h), \qquad& 
		\forall h\in 
		H_{\open} \tag{3}\label{ILP:quotas2}\\
		\sum_{h\in H_{\open}} x^\mathbf{z}_{i,h} \leq n_i^\mathbf{z}, 
		\qquad& \forall 
		i\in 
		[q],\forall \mathbf{z}\in \{-1,0,1\}^m \tag{4} 
		\label{ILP:num-residents2}
		\\
		x^\mathbf{z}_{i,h}  = 0, \qquad &\forall i \in [q], \forall h\in 
		H_{\open}\setminus
		A(t_i), \forall \mathbf{z}\in \{-1,0,1\}^m
		\tag{5} \label{ILP:acceptablitiy2}\\
		\sum_{i\in [q], \mathbf{z}\in \{-1,0,1\}^m } x^\mathbf{z}_{i,h}
		= u(h), \qquad & \forall h\in H_{\full} \tag{6} \label{ILP:full} \\
		\sum_{i\in [q], \mathbf{z}\in \{-1,0,1\}^m } x^\mathbf{z}_{i,h}
		\le u(h) - 1, \qquad & \forall h\in H_{\open} \setminus H_{\full} \tag{7} \label{ILP:undersubscribed} \\
		x^\mathbf{z}_{i,h}  = 0, \qquad &\forall i \in [q],\forall h \in 
		H_{\open}, \forall
		\mathbf{z}\in \{-1,0,1\}^m\text{ with } z_h=-1 
		\tag{8} \label{ILP:guess1}\\
		\sum_{\substack{i\in [q], \mathbf{z}\in 
			\{-1,0,1\}^m\\ \text{ with } z_h=0}}
		x^\mathbf{z}_{i,h}  \geq 1, \qquad &\forall h \in H_{\open}
		\tag{9} \label{ILP:guess2}\\
		x^\mathbf{z}_{i,h}\in\{0,1,\dots,n^\mathbf{z}_i\}, \qquad &\forall i\in [q] ,
		\forall h \in H_{\open}, \forall \mathbf{z}\in \{-1,0,1\}^m			
		\tag{10} \label{ILP:integrality2}
		\end{align*}
		Conditions (\ref{ILP:no-bpa2}) and (\ref{ILP:no-bpb2}) ensure 
		that no blocking pair exists. Condition (\ref{ILP:no-bpa2}) checks for 
		all~$h\in H_{\open}$ that 
		all residents that are preferred by $h$ to $r_h$ are matched 
		to a hospital that they find as least as good as $h$. 
		All other residents only form a blocking pair with $h$ if $h$ is 
		undersubscribed and they are matched to a hospital to which they 
		strictly prefer $h$. The non-existence of such a pair is enforced by 
		Condition (\ref{ILP:no-bpb2}). 
		Condition (\ref{ILP:no-bc2}) ensures that no blocking coalition to open
		a hospital in $H\setminus H_{\open}$ exists, while Condition~(\ref{ILP:quotas2}) ensures that all hospitals in $H_{\open}$ respect their lower and upper quota.
		Condition (\ref{ILP:num-residents2}) enforces that for each resident 
		type the number of matched residents of this type does not exceed the 
		number of 
		residents of this type from the given instance, while Condition~(\ref{ILP:acceptablitiy2}) enforces that no resident is matched to a
		hospital she does not accept.
		Conditions~(\ref{ILP:full}) and~(\ref{ILP:undersubscribed}) ensure that each hospital in~$H_{\full}$ is indeed full, while all other hospitals from~$H_{\open} \setminus H_{\full}$ are undersubscribed.
		Finally, Conditions~(\ref{ILP:guess1})
		and~(\ref{ILP:guess2}) ensure that for all hospitals~$h\in
		H_{\open}$, the guess for the worst matched resident $r_h$ is
		correct by enforcing that no resident $r\in R$ with 
		$r_h\succ_{h} r$ is matched to $h$ (Condition~(\ref{ILP:guess1})) and by enforcing that at least one
		resident $r\in R$ with $r\sim_h r_h$ is matched to $h$ 
		(Condition~(\ref{ILP:guess2})).
		
		For each guess, we solve the ILP from above and return YES if it is 
		feasible. Otherwise, we reject the current guess and continue with the 
		next guess and return NO 
		after all guesses have been rejected.
		
		Note that the total number of guesses is $3^m\cdot n^m$, and for each 
		guess, we need to solve the above~ILP. As the number of variables in
		the ILP lies 
		in $\mathcal{O}(m^2\cdot m!\cdot 2^{m}\cdot 3^m)$, by employing 
		Lenstra's algorithm~\cite{DBLP:journals/mor/Kannan87,DBLP:journals/mor/Lenstra83}, it is
		possible to 
		solve the problem in $\mathcal{O}(f(m)\cdot n^m)$ for some 
		computable function in $f$. 
	\end{proof}
	\section{Conclusion}
	\label{sec:conclusion}
	We conducted a thorough parameterized complexity
	analysis of the 
	Hospital
	Residents problem with lower and upper quotas. We have shown that the 
	hardness of this problem arises from choosing the set of open hospitals such
	that no blocking coalition exists, as the problem remains hard even if all 
	hospitals have only lower quotas and pairs cannot block an outcome, but it
	becomes easy as soon as the set of open hospitals is given. We have
	also 
	analyzed two variants of this problem.
		
	For future work, one could generalize \HRLUQ by having only hospitals, and 
	hospitals having preferences over other hospitals (similar to the 
	\textsc{Capacitated Stable Roommates} problem introduced by 
	Cechl{\'{a}}rov{\'{a}} and Fleiner~\cite{DBLP:journals/talg/CechlarovaF05}).
	Then, a feasible matching would be a set of pairs of hospitals where each hospital appears in either no or between its upper and lower quota many pairs.
	A blocking pair for a matching~$M$ would be a pair~$\{h_1, h_2\}$ of open 
	hospitals such that, for $i\in \{1,2\}$, $h_i $ is undersubscribed or 
	prefers $h_{3-i}$ to one hospital assigned to it by~$M$, while a blocking 
	coalition would be a closed hospital $h$ together with $\lowerq (h)$ 
	hospitals $h_1, \dots, h_{\lowerq (h)}$ such that, for $i\in [l(h)]$, $h_i 
	$ is undersubscribed or prefers $h$ to one of its partners in $M$.
	While all our hardness results clearly carry over to this setting, it would 
	be interesting to see whether this generalization is also solvable in 
	polynomial time if all lower quotas are at most two.

	Further, we believe that it would be 
	interesting to analyze other stable many-to-one 
	matching problems using a similar fine-grained parameterized approach as 
	taken in this paper to enrich our understanding of the 
	complexity of these problems.
	
	\bibliographystyle{splncs04.bst}

\end{document}